\documentclass[11pt]{article}

\usepackage[margin=1in]{geometry}
\usepackage{graphicx}
\usepackage[pdftex,colorlinks=true,linkcolor=blue,citecolor=blue,urlcolor=black]{hyperref}
\usepackage{amsmath, amsthm, amssymb}
\usepackage{subfigure}
\usepackage{comment}
\usepackage{url}
\usepackage{pdflscape}
\usepackage[ruled,lined,linesnumbered]{algorithm2e}
\usepackage{tikz}
\usetikzlibrary{decorations.pathreplacing}
\usepackage{youngtab}
\usepackage{braket}

\newcommand{\R}{\mathbb{R}}

\newcommand{\C}{\mathbb{C}}

\newcommand{\cH}{\mathcal{H}}
\renewcommand{\Re}{\operatorname{Re}}
\renewcommand{\Im}{\operatorname{Im}}
\renewcommand{\ge}{\geqslant}
\renewcommand{\le}{\leqslant}

\newcommand{\ip}[2]{\langle #1|#2 \rangle}

\newcommand{\proj}[1]{| #1 \rangle \langle #1 |}
\newcommand{\bracket}[3]{\langle #1|#2|#3 \rangle}
\newcommand{\sm}[1]{\left( \begin{smallmatrix} #1 \end{smallmatrix} \right)}

\DeclareMathOperator{\poly}{poly}

\DeclareMathOperator{\tr}{Tr}
\DeclareMathOperator{\rank}{rank}

\DeclareMathOperator{\linspan}{span}
\DeclareMathOperator{\eff}{eff}
\DeclareMathOperator{\od}{od}
\DeclareMathOperator{\hc}{h.c.}

\newcommand{\Antisymspace}{\mathcal{H}_{\operatorname{antisym}}(d-1)}
\newcommand{\haklt}{h^{\text{AKLT}}}
\newcommand{\Pitot}{\Pi_{\operatorname{tot}}}

\newcommand{\be}{\begin{equation}}
\newcommand{\ee}{\end{equation}}
\newcommand{\bea}{\begin{eqnarray}}
\newcommand{\eea}{\end{eqnarray}}
\newcommand{\bes}{\begin{equation*}}
\newcommand{\ees}{\end{equation*}}
\newcommand{\beas}{\begin{eqnarray*}}
\newcommand{\eeas}{\end{eqnarray*}}

\newcommand{\cS}{\mathcal{S}}

\tikzstyle{qubit}=[circle,draw,fill,thick,
inner sep=0pt,minimum size=2mm]

\makeatletter
\newtheorem*{rep@theorem}{\rep@title}
\newcommand{\newreptheorem}[2]{%
\newenvironment{rep#1}[1]{%
 \def\rep@title{#2 \ref{##1} (restated)}%
 \begin{rep@theorem}}%
 {\end{rep@theorem}}}
\makeatother

\newtheorem{thm}{Theorem}
\newtheorem*{thm*}{Theorem}
\newtheorem{cor}[thm]{Corollary}

\newtheorem{lem}[thm]{Lemma}
\newtheorem*{lem*}{Lemma}

\newtheorem{dfn}{Definition}
\newreptheorem{thm}{Theorem}
\newreptheorem{lem}{Lemma}
\newreptheorem{cor}{Corollary}

\begin{document}

\title{Universal qudit Hamiltonians}
\author{Stephen Piddock\thanks{{\tt stephen.piddock@bristol.ac.uk}.}\ \ and Ashley Montanaro\thanks{{\tt ashley.montanaro@bristol.ac.uk}.}\\[3pt] {\small School of Mathematics, University of Bristol, UK}}
\maketitle

\begin{abstract}
A family of quantum Hamiltonians is said to be universal if any other finite-dimensional Hamiltonian can be approximately encoded within the low-energy space of a Hamiltonian from that family. If the encoding is efficient, universal families of Hamiltonians can be used as universal analogue quantum simulators and universal quantum computers, and the problem of approximately determining the ground-state energy of a Hamiltonian from a universal family is QMA-complete. One natural way to categorise Hamiltonians into families is in terms of the interactions they are built from. 
%This is a generalisation of the classical notion of defining constraint satisfaction problems (CSPs) in terms of their allowed constraints.
Here we prove universality of some important classes of interactions on qudits ($d$-level systems):
\begin{itemize}
\item We completely characterise the $k$-qudit interactions which are universal, if augmented with arbitrary 1-local terms.
 %(a generalisation of fixing individual variables in a CSP). 
We find that, for all $k \ge 2$ and all local dimensions $d \ge 2$, almost all such interactions are universal aside from a simple stoquastic class.
\item We prove universality of generalisations of the Heisenberg model that are ubiquitous in condensed-matter physics, even if free 1-local terms are not provided. We show that the $SU(d)$ and $SU(2)$ Heisenberg interactions are universal for all local dimensions $d \ge 2$ (spin $\ge 1/2$), implying that a quantum variant of the Max-$d$-Cut problem is QMA-complete. We also  show that for $d=3$ all bilinear-biquadratic Heisenberg interactions are universal. One example is the general AKLT model.
\item We prove universality of any interaction proportional to the projector onto a pure entangled state.
\end{itemize}
%These results are proven via extensive use of perturbative gadgets to simulate one Hamiltonian with another. Analysis of these gadgets requires novel techniques, including fourth-order perturbation theory and the theory of quadratic Casimir operators.
\end{abstract}

\pagebreak

% ------------------------------------------------------------------------------

\section{Introduction}
\label{sec:intro}

What does it mean to say that a class of (quantum-)physical systems is complex? One perspective is to look at the physical phenomena displayed by that type of system. If these phenomena are rich and complex, then the system arguably can be said to be complex itself. Another perspective is to look at the computational power of the system: the ability to build a universal computer using the system would serve as strong evidence that the system is complex.

Interestingly, in some cases these notions of complexity are equivalent. Recent work by us, together with Cubitt, introduced and characterised the notion of universality in many-body quantum Hamiltonians~\cite{cubitt17}. A family of Hamiltonians is said to be universal if any other quantum Hamiltonian can be simulated arbitrarily well by some Hamiltonian in that family. By ``simulate'', we mean the following (see Section \ref{sec:prelims} below for a formal definition): Hamiltonian $A$ simulates Hamiltonian $B$ if the low-energy part of $A$ is close to $B$ in operator norm, up to a local isometry (i.e.\ a map which associates each subsystem of the $B$ system with a discrete set of subsystems of the $A$ system).

This notion of simulation is very strong, as it implies that the low-energy part of $A$ reproduces all physical properties of $B$ (such as eigenvalues, ground states, partition functions, correlation functions, etc.)~\cite{cubitt17}. Universality is correspondingly a very strong notion. As a universal family $\mathcal{F}$ of Hamiltonians can simulate any other quantum Hamiltonian, any physical phenomenon that can occur in a quantum system must occur within Hamiltonians picked from $\mathcal{F}$. This implies that the ability to implement Hamiltonians in $\mathcal{F}$ allows universal ``analogue'' simulation of arbitrary quantum systems~\cite{georgescu14,cirac12}. In addition, if one also assumes that the simulation can be computed efficiently (as is usually the case), universal families of Hamiltonians are computationally universal, in a number of senses~\cite{cubitt17}.  First, they can be used to perform arbitrary quantum computations, either by preparing a simple initial state, evolving according to $H \in \mathcal{F}$ for some time and measuring, or via adiabatic evolution. Second, the problem of approximately computing the ground-state energy of Hamiltonians from $\mathcal{F}$ is QMA-complete, where QMA is the quantum analogue of the complexity class NP~\cite{bookatz14,gharibian15}, and hence expected to be computationally hard.

A natural way to classify physical systems is in terms of the types of interactions that they are built from. Let $\cS$ be a set of interactions on up to $k$ qudits ($d$-level subsystems), i.e.\ each element of $\cS$ is a Hermitian operator on $(\C^d)^{\otimes l}$ for some $l \le k$. Then we say that an $n$-qudit Hamiltonian $H$ is an $\cS$-Hamiltonian if
\be \label{eq:generalform} H = \sum_i \alpha_i H^{(i)}, \ee
where for all $i$, $\alpha_i \in \R$ and the non-trivial part of $H^{(i)}$ is picked from $\cS$. That is, $H^{(i)} = H \otimes I$ for some $H \in \cS$. $H$ is a so-called $k$-local Hamiltonian. We stress that the $\alpha_i$ coefficients can (usually) be either positive or negative. We also say that $H$ is an $\cS$-Hamiltonian with local terms if it can be written in the form (\ref{eq:generalform}) by adding arbitrary 1-local operators. The form (\ref{eq:generalform}) encompasses a vast array of the Hamiltonians studied in condensed-matter physics, such as the general Ising model ($\cS = \{ Z \otimes Z\}$) and the general Heisenberg model ($\cS = \{X \otimes X + Y \otimes Y + Z \otimes Z\}$). In the case where $\cS = \{h\}$ for some $h$, we just call $H$ an $h$-Hamiltonian.

Determining the complexity of $\mathcal{S}$-Hamiltonians is a natural quantum generalisation of the long-running programme in classical complexity theory of classifying constraint satisfaction problems (CSPs) according to their complexity. Beginning with Schaefer's famous 1978 dichotomy theorem for boolean CSPs~\cite{schaefer78}, which has been extended in many different directions since (see e.g.~\cite{creignou01,thapper16} for references), this project aims to pinpoint, for each possible set of constraints $\mathcal{S}$, the complexity of a CSP that uses only constraints from $\mathcal{S}$ (perhaps weighted, to give an optimisation problem). A quantum generalisation of this question is to determine the complexity of approximately computing the ground-state energy of $\mathcal{S}$-Hamiltonians up to $1/\poly(n)$ precision~\cite{gharibian15}. This problem, which we call simply {\sc $\mathcal{S}$-Hamiltonian}, is a special case of the {\sc Local Hamiltonian} problem, which in general is QMA-complete~\cite{Kempe-Kitaev-Regev,Kitaev-Shen-Vyalyi} when $\cS$ contains all $k$-qubit interactions for any fixed $k \ge 2$. The classical special case of the {\sc $\mathcal{S}$-Hamiltonian} problem corresponds to $\mathcal{S}$ containing only diagonal interactions; such problems are known as ``valued'' or ``generalised'' CSPs, and a full complexity classification of these was only obtained in 2016, by Thapper and \v{Z}ivn\'y~\cite{thapper16}.

A full classification was given in~\cite{Cubitt-Montanaro} of the computational complexity of the {\sc $\mathcal{S}$-Hamiltonian} problem in the special case where all interactions in $\mathcal{S}$ are on at most 2 qubits; this was sharpened by~\cite{Bravyi-Hastings}, which showed that one complexity class in the classification was equivalent to the previously studied class StoqMA~\cite{bravyi06a}. It was later shown in~\cite{cubitt17} that each of the classes in~\cite{Cubitt-Montanaro} corresponds to a physical universality class. These results can be summarised as follows:
\begin{thm}[\cite{Science,jonsson00,Cubitt-Montanaro,cubitt17,Bravyi-Hastings}]
\label{thm:qubits}
  Let $\mathcal{S}$ be any fixed set of two-qubit and one-qubit interactions such that $\mathcal{S}$ contains at least one interaction which is not 1-local. Then:
  \begin{itemize}
  \item If there exists $U \in SU(2)$ such that $U$ locally diagonalises $\mathcal{S}$, then $\mathcal{S}$-Hamiltonians are universal classical Hamiltonian simulators~\cite{Science} and the {\sc $\mathcal{S}$-Hamiltonian} problem is NP-complete~\cite{jonsson00,Cubitt-Montanaro};
  \item Otherwise, if there exists $U \in SU(2)$ such that, for each 2-qubit matrix $H_i \in \mathcal{S}$, $U^{\otimes 2} H_i (U^\dag)^{\otimes 2} = \alpha_i Z^{\otimes 2} + A_i\otimes I + I \otimes B_i$, where $\alpha_i \in \R$ and $A_i$, $B_i$ are arbitrary single-qubit interactions, then $\mathcal{S}$-Hamiltonians are universal stoquastic Hamiltonian simulators~\cite{cubitt17} and the {\sc $\mathcal{S}$-Hamiltonian} problem is StoqMA-complete~\cite{Bravyi-Hastings,Cubitt-Montanaro};
  \item Otherwise, $\mathcal{S}$-Hamiltonians are universal quantum Hamiltonian simulators~\cite{cubitt17} and the {\sc $\mathcal{S}$-Hamiltonian} problem is QMA-complete~\cite{Cubitt-Montanaro}.
  \end{itemize}
\end{thm}
A stoquastic Hamiltonian is one whose off-diagonal elements in the standard basis are all nonpositive. Here we sometimes generalise this terminology slightly by also calling $H$ stoquastic if there exists a local unitary $U$ such that $U^{\otimes n} H (U^\dag)^{\otimes n}$ is stoquastic.

% ------------------------------------------------------------------------------

\subsection{Our results}

Here we continue the programme of classifying universality of Hamiltonians -- and hence the computational complexity of the {\sc $\mathcal{S}$-Hamiltonian} problem -- by generalising from qubit interactions to qudit interactions, i.e.\ local dimension $d > 2$, or equivalently spin $>1/2$. As well as being a natural next step from the perspective of computational complexity, this framework includes many important models studied in condensed-matter theory \cite{Affleck89,Beach09,Harada02,Kennedy90,Lauchli06,Lou09,Read89}. However, it is significantly more difficult than the qubit case. One reason for this is that in the case of qubits, there was a simple ``canonical form'' into which any 2-qubit interaction could be put by applying local unitaries~\cite{Cubitt-Montanaro}, which dramatically reduced the number of types of interaction that needed to be considered. No comparably simple canonical form seems to exist for $d > 2$~\cite{linden99}.

We first consider $\cS$-Hamiltonians with local terms. This is a more general setting than just $\cS$-Hamiltonians, and hence easier to prove universality results. From a computer science point of view, allowing free local terms corresponds to allowing arbitrary constraints or penalties on individual variables in a CSP. For conciseness, we say that $\cS$ is LA-universal (``locally assisted universal'') if the family of $\cS$-Hamiltonians with local terms is universal. Similarly, we say that $\cS$ is LA-stoquastic-universal if it can simulate any stoquastic Hamiltonian. Then our main result about universality with local terms is a complete classification theorem:

\begin{thm}
\label{thm:lamainklocal}
Let $\cS$ be a set of interactions, which are not all 1-local, between qudits of dimension $d$. Then $\cS$ is:
\begin{itemize}
\item stoquastic and LA-stoquastic-universal, if there exists $\ket{\psi}\in \C^d$ such that all interactions in $\cS$ are, up to the addition of 1-local terms, given by a linear combination of operators taken from the set $\{I, \proj{\psi},\proj{\psi}^{\otimes 2},\proj{\psi}^{\otimes 3},\dots \}$;
\item LA-universal, otherwise.
\end{itemize}
\end{thm}

We note some general consequences of this result for Hamiltonians assisted by local terms. First, we see that any nontrivial $k$-qudit interaction can be used to simulate an arbitrary stoquastic Hamiltonian. Second, almost any $k$-qudit interaction can actually be used to simulate arbitrary general Hamiltonians. Third, perhaps surprisingly, there exist Hamiltonians whose 2-local part is diagonal, but which are LA-universal.

We highlight some examples for $d=3$. Consider
\[ \cS_1 = \left\{ \begin{pmatrix} 1 & 0 & 0\\ 0 & -1 & 0\\ 0 & 0 & -1 \end{pmatrix}^{\otimes 2} \right\}, \;\;\;\; \cS_2 = \left\{ \begin{pmatrix} 1 & 0 & 0\\ 0 & -1 & 0\\ 0 & 0 & 0 \end{pmatrix}^{\otimes 2} \right\}. \]
The single interaction in $\cS_1$ is equal to $\proj{0}^{\otimes 2}$ plus some 1-local terms, so $\cS_1$ is stoquastic and LA-stoquastic-universal. On the other hand, the interaction in $\cS_2$ cannot be decomposed in this way, so $\cS_2$ is LA-universal. So, for example, given access to interactions of the form of $\cS_2$ and arbitrary local terms, one can perform universal quantum computation.

Next we consider the more general $H$-Hamiltonian problem, where the lack of ``free'' 1-local terms makes it much more challenging to prove universality results. Here we focus on qudit generalisations of the qubit Heisenberg (exchange) interaction ($h\propto X\otimes X+Y\otimes Y+Z\otimes Z$). Hamiltonians built from this interaction enjoy significant levels of symmetry, which made it one of the most difficult cases to prove universal in previous work~\cite{Cubitt-Montanaro,cubitt17}. The most symmetric such generalisation in local dimension $d$ is the $SU(d)$ Heisenberg model (often known as ``$SU(N)$ Heisenberg model'' in the literature~\cite{Lou09,Beach09}), where the interaction is
\be \label{eq:heisenbergsud} h=\sum_{a=1}^{d^2-1} T^a \otimes T^a \ee
for some $d\times d$ traceless Hermitian matrices $T^a$ such that $\tr(T^a T^b)=\frac{1}{2}\delta_{ab}$. Up to adding an identity term and rescaling, $h$ is just the swap operator, or the projector onto the symmetric subspace of two qudits,
\[ P_{\text{sym}} = \frac{1}{4} \sum_{i,j} (\ket{ij} + \ket{ji})(\bra{ij} + \bra{ji}). \]
$h$ is invariant under conjugation by local unitaries, implying that the eigenspaces of any Hamiltonian built only from $h$ interactions inherit this property. Nevertheless, we have the following result:

\begin{thm}
\label{thm:sud}
For any $d \ge 2$, the $SU(d)$ Heisenberg interaction $h:=\sum_{a} T^a \otimes T^a$, where $\{T^a\}$ are traceless Hermitian matrices such that $\tr(T^a T^b)=\frac{1}{2}\delta_{ab}$, is universal. This holds even if the weights $\alpha_i$ in the decomposition (\ref{eq:generalform}) are restricted to be non-negative.
\end{thm}

The special case $d=2$ of Theorem \ref{thm:sud} was shown in~\cite{cubitt17}. As a corollary of Theorem~\ref{thm:sud}, we obtain QMA-hardness of a quantum variant of the Max-$d$-Cut problem~\cite{frieze97} (equivalently, a quantum generalisation of the (classical) antiferromagnetic Potts model~\cite{wu82}). In the Max-$d$-Cut problem, we are given a graph where each edge $(i,j)$ has a non-negative weight $w_{ij}$, and are asked to partition the vertices into $d$ sets, such that the sum of the weights of edges between vertices in different sets is maximised. That is, we find a map $c$ from each vertex $i$ to an integer $c(i) \in [d]$ such that $\sum_{i<j} w_{ij} (1-\delta_{c(i)c(j)})$ is maximised. The natural ``quantum'' way of generalising this problem is to replace each vertex with a $d$-dimensional qudit, and replace each weighted edge across two vertices with a weighted projector onto the symmetric subspace across the corresponding qudits (equivalently, an interaction $h$). Then the task is to approximate the ground-state energy of the corresponding Hamiltonian $\sum_{i<j} w_{ij} h_{ij}$, up to precision $1/\poly(n)$. Call this problem {\sc Quantum Max-$d$-Cut}.

To see why this is a suitable (and non-trivial) generalisation, note that $P_{\text{sym}}$ gives an energy penalty to a pair of qudits that are both in the same computational basis state, similarly to the classical case, but that the behaviour of the quantum variant can sometimes be quite different. For example, consider the case $d=2$, and four vertices arranged in an unweighted cycle. Classically, the vertices can clearly be partitioned into two sets such that there are no edges between vertices in the same set. However, there is no quantum state that is simultaneously in the ground space of all corresponding projectors $P_{\text{sym}}$. This is because the unique ground state of $P_{\text{sym}}$ is maximally entangled, and each qubit cannot be maximally entangled with both of its neighbours simultaneously.

It is an immediate consequence of Theorem \ref{thm:sud} that:

\begin{cor}
\label{cor:maxdcut}
For any $d \ge 2$, {\sc Quantum Max-$d$-Cut} is QMA-complete.
\end{cor}

The special case $d=2$ of Corollary \ref{cor:maxdcut} was shown in~\cite{Piddock-Montanaro}.

Next, we consider the case where the interactions are of the form $P=\proj{\psi}$ for an entangled two qudit state $\ket{\psi}$.
\begin{thm}
\label{thm:proj}
Let $P= \proj{\psi}$  be the projector onto an entangled two-qudit state $\ket{\psi} \in (\C^d)^{\otimes 2}$. Then $\{P\}$-Hamiltonians are universal.
\end{thm}

In fact, Theorem \ref{thm:proj} holds even in the restrictive setting where all the interactions are required to sit on the edges of a bipartite interaction graph (see Section \ref{sec:proj} for a precise statement). 
Entanglement is a very well studied property of quantum systems, and is well known to be fundamental to many interesting quantum phenomena.
This result can be viewed as an intriguing and apparently tight link between entanglement and universality.

A perhaps more familiar, and also very well-studied, interaction we consider is another generalisation of the qubit Heisenberg interaction (e.g.~\cite{Affleck89,parkinson10,mattis93}): the $SU(2)$ Heisenberg interaction in local dimension $d$ (often just called the ``spin-$s$ Heisenberg interaction'', where $s=(d-1)/2$). Now the interaction is of the form
\[ h = S^x \otimes S^x + S^y \otimes S^y + S^z \otimes S^z, \]
where $S^x$, $S^y$, $S^z$ generate a $d$-dimensional irreducible representation of $\mathfrak{su}(2)$ and correspond to the familiar Pauli matices $X$, $Y$, $Z$ (up to an overall scaling factor). Note that, although the Lie algebra involved is the same as for the qubit case, the interaction $h$ may have very different properties for higher $d$; for example, it has $d$ distinct eigenvalues (see equation (\ref{eq:su2tensordecomp}) below). Nevetheless, this generalisation turns out to be universal too:

\begin{thm}
\label{thm:su2}
For any $d \ge 2$, the $SU(2)$ Heisenberg interaction $h= S^x \otimes S^x + S^y \otimes S^y + S^z \otimes S^z$, where $S^x$, $S^y$, $S^z$ are representations of the Pauli matrices $X$, $Y$, $Z$, is universal.
\end{thm}

Finally, we consider yet another well-studied generalisation of the Heisenberg model (see e.g.~\cite{AKLT,Harada02,Kennedy90,Lauchli06}): the general bilinear-biquadratic Heisenberg model in local dimension $d=3$ (spin 1). Here the interaction used is
\[ h^{(\theta)} := (\cos \theta) h + (\sin \theta) h^2, \]
where $\theta \in [0,2\pi)$ is an arbitrary parameter and $h$ is the spin-1 Heisenberg interaction, which can be written explicitly as
\be \label{eq:spin1intro} h = X_3 \otimes X_3 + Y_3 \otimes Y_3 + Z_3 \otimes Z_3 \ee
where
\[ X_3 = \frac{1}{\sqrt{2}} \begin{pmatrix} 0 & 1 & 0\\ 1& 0& 1\\ 0 & 1 & 0\end{pmatrix},\;\;\;\;Y_3 = \frac{i}{\sqrt{2}} \begin{pmatrix} 0 & -1 & 0\\ 1& 0& -1\\ 0 & 1 & 0\end{pmatrix},\;\;\;\;Z_3 = \begin{pmatrix} 1 & 0 & 0\\ 0& 0& 0\\0 & 0 & -1\end{pmatrix}. \]
The special case $\theta = \arctan 1/3$ corresponds to the famous Affleck-Kennedy-Lieb-Tasaki (AKLT) model~\cite{AKLT}. Our result here is as follows:

\begin{thm}
\label{thm:bilinbiq}
Let $h^{(\theta)} := (\cos \theta) h + (\sin \theta) h^2$, where $\theta \in [0,2\pi)$ is an arbitrary parameter and $h$ is the spin-1 Heisenberg interaction. For all $\theta$, $h^{(\theta)}$ is universal.
\end{thm}

We therefore see that, although different values of $\theta$ may correspond to very different physics~\cite{Lauchli06}, from a universality point of view they are all of equal power.

We remark that, in common with most previous work in this area~\cite{Cubitt-Montanaro,cubitt17}, we usually allow each interaction weight to be positive or negative. This can lead to physical systems built from the same interaction having very different physical properties (e.g.\ antiferromagnetism vs.\ ferromagnetism). It is sometimes possible to prove universality-type results for interactions whose weights all have the same sign~\cite{Piddock-Montanaro}; we achieve this in Theorem \ref{thm:sud}, but in general leave this extension for future work. Another interesting direction is to prove universality for systems with simpler interaction patterns~\cite{Oliveira-Terhal,Schuch-Verstraete,Piddock-Montanaro,cubitt17}, or with less heavily-weighted interactions~\cite{cao15}.

% ------------------------------------------------------------------------------

\subsection{Related work}

There has been a substantial amount of work characterising the complexity of various types of qubit Hamiltonians from the perspective of QMA-completeness; see~\cite{cubitt17,bookatz14,gharibian15} for references. In the case of qudits, rather than general classification results, most work has considered carefully designed special cases where QMA-completeness can be achieved. Indeed, it is often the case that these results aim to {\em reduce} the local dimension of a QMA-complete construction that achieves some other desiderata. For example, Aharonov et al.~\cite{aharonov09} gave a QMA-complete family of local Hamiltonians on a 1D line with $d=12$, later improved to $d=8$ by Hallgren, Nagaj and Narayanaswami~\cite{hallgren13}; Gottesman and Irani~\cite{gottesman13} gave a QMA$_{\text{EXP}}$-complete family of translationally invariant Hamiltonians on a 1D line with $d=O(10^6)$, later improved to $d \approx 40$ by Bausch, Cubitt and Ozols~\cite{bausch16}. The local dimension has been reduced even further to $d=4$, for a translationally invariant Hamiltonian on a 3D lattice~\cite{Bausch17}. We refer to~\cite{bookatz14} for further examples, including the more general case where the local dimension can vary across the system being considered. In all these cases, one fixes the dimension and then carefully tunes the types of interactions used to achieve the desired result. Here, by contrast, we begin with a fixed set of interactions and attempt to determine the complexity of Hamiltonians based on these interactions.

% ------------------------------------------------------------------------------

\subsection{Overview of proof of Theorem \ref{thm:lamainklocal}}

We now give an informal discussion of our LA-universality classification result. The majority of the work to prove Theorem \ref{thm:lamainklocal} is taken up by the special case of 2-local interactions, and sets $\cS$ containing only one interaction. To prove universality of an interaction $H$, we use simulations: showing that an interaction known to be universal~\cite{Cubitt-Montanaro,cubitt17} can be implemented using Hamiltonians consisting of $H$ terms and 1-local terms. Our simulations are all based on perturbative gadgets, as introduced in~\cite{Kempe-Kitaev-Regev} and used for example in~\cite{Bravyi-Hastings,cubitt17,Oliveira-Terhal}, to effectively implement one Hamiltonian within the ground space of another.
For example, a type of gadget we often use is a so-called mediator gadget. In this type of gadget, one or more ancilla (``mediator'') qudits are added to the system. Strong interactions within the mediator qudits effectively project these qudits into a fixed state. Then weaker interactions between the mediator and original qudits implement effective interactions between the original qudits. The interactions produced are determined rigorously via perturbation theory.
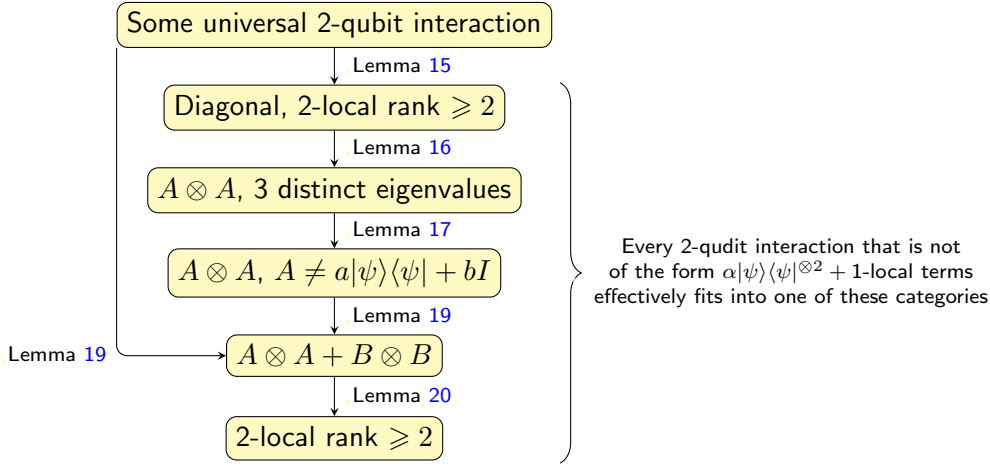
\begin{figure}
  \begin{center}
    \begin{tikzpicture}[yscale=1.1,font=\sffamily,every node/.style={draw,fill=yellow!30,rounded corners},label/.style={draw=none,fill=none,font={\sffamily,\scriptsize}},>=stealth]
      \node (uni2) at (0,0) {Some universal 2-qubit interaction};
      \node (diag) at (0,-1) {Diagonal, 2-local rank $\ge 2$};
      \node (aa) at (0,-2) {$A \otimes A$, 3 distinct eigenvalues};
      \node (aa2) at (0,-3) {$A \otimes A$, $A \neq a \proj{\psi} + b I$};
      \node (aabb) at (0,-4) {$A \otimes A + B \otimes B$};
            \node (2qudit) at (0,-5) {2-local rank $\ge 2$};
      \draw[->] (uni2.south) -- node[label,right] {Lemma \ref{lem:diagonal}} (diag);
      \draw[->] (diag.south) -- node[label,right] {Lemma \ref{lem:3evalues}} (aa);
            \draw[->] (aa.south) -- node[label,right] {Lemma \ref{lem:AAuniversal}} (aa2);
            \draw[->] (aa2.south) -- node[label,right] {Lemma \ref{lem:AA+BB}} (aabb);
            \draw[->,rounded corners] (uni2.south west) |- node[label,left] {Lemma \ref{lem:AA+BB}} (aabb.west);
                        \draw[->] (aabb.south) -- node[label,right] {Lemma \ref{lem:rank2ormore}} (2qudit);
                  \draw [decorate,decoration={brace,amplitude=8pt}] (3,-0.7) -- node[label,right]
        {\begin{tabular}{c}
         Every 2-qudit interaction that is not\\ of the form $\alpha \proj{\psi}^{\otimes 2} + \text{1-local terms}$\\ effectively fits into one of these categories\\\end{tabular}} (3,-5.3);
   \end{tikzpicture}
 \end{center}
 \caption[Sequence of simulations used in this work.]{%
   Sequence of simulations used in this work.
   An arrow from one box to another indicates that a Hamiltonian of the first type can be simulated by a Hamiltonian of the second type.
%   Where two arrows enter a box, this indicates that a Hamiltonian of this type can be simulated by one of the two target Hamiltonians, but not necessarily both.
   }
%   ``2SLD'' is short for ``the 2\nbd-local parts of all interactions in the set are simultaneously locally diagonalisable'', and $k,k' \ge 2$ are arbitrary integers such that $k \ge \lceil k' \log_2 d\rceil$.}
 \label{fig:reductions}
\end{figure}

First we consider the special case of diagonal interactions with 2-local rank $\ge 2$, where the 2-local rank of an interaction $H$ is informally defined as follows: Writing $H= H' + \text{1-local terms}$, and $H' = \sum_{a,b} M_{ab} T^a \otimes T^b$ for some basis $T^a$ of Hermitian matrices, the 2-local rank of $H$ is the rank of $M$. (For example, $H = X \otimes X + Y \otimes I$ has 2-local rank 1.) We can think of diagonal matrices symmetric under qudit interchange and with 2-local rank 2 as being of the form $D \otimes D + E \otimes E$ for some diagonal matrices $D$ and $E$. To show that such interactions are universal (a similar argument works for non-symmetric interactions), we use our free 1-local terms to apply a heavy interaction to each qudit which effectively projects it into a 2-dimensional subspace. Note that even though $D$ and $E$ commute, this need not be the case for the corresponding projected qubit interactions. This allows us to generate a 2-qubit effective interaction within this subspace which is universal~\cite{cubitt17}.

Remaining within the special case of diagonal interactions, the next step is to consider those with 2-local rank 1, which are of the form $A \otimes A$. To deal with this case, we split into two parts. When $A$ has at least 3 distinct eigenvalues, we design a gadget using an additional qudit to implement the effective interaction $A \otimes A^2 + A^2 \otimes A$, which is universal from the previous case. When $A$ has 2 distinct eigenvalues, but is not of the form $a\proj{\psi} + b I$, we show that another gadget can be used to simulate an interaction $B \otimes B$ where $B$ has 3 distinct eigenvalues. For the remaining diagonal case -- interactions of the form $A \otimes A$ for $A = a\proj{\psi} + b I$ -- we show that local unitary rotations can be used to transform any Hamiltonian built of such interactions into a stoquastic Hamiltonian, so we cannot expect this case to be universal.

We then move on to non-diagonal interactions. We first consider those of the form $A \otimes A + B \otimes B$ for some $B$ that does not commute with $A$ (otherwise we would be in the diagonal case). For all such interactions, we show there exists a gadget which projects the interaction onto a 2-qubit subspace on which the resulting interaction is universal. The non-commutativity makes this task simpler than in the diagonal case. The next step is interactions with 2-local rank $\ge 2$, but not of the form $A \otimes A + B \otimes B$. For these, we show that one can always produce an effective interaction of the form $A \otimes A + B \otimes B$ using two rounds of simulation.

All 2-qudit interactions $H$ can be handled using one of these lemmas. Considering the interaction $H'$ formed by deleting the 1-local parts from $H$, we know that $H$ is LA-universal if the 2-local rank of $H'$ is $\ge 2$. If not, then $H' = A \otimes B$ for some $A$ and $B$. Either $A \otimes B + B\otimes A$ has 2-local rank $\ge 2$, or $B$ is proportional to $A$. Either way, we are in one of the previously considered cases.

The final step to complete the proof of Theorem \ref{thm:lamainklocal} is to generalise to $k$-local interactions for $k>2$. To do so, we show that our free 1-local terms can be used to extract 2-local ``sub-interactions'' from the interactions we are given; this is a generalisation to $d>2$ of an analogous argument for qubits in~\cite{Cubitt-Montanaro}. Then either we can produce a universal sub-interaction, or all the sub-interactions of all interactions in $\cS$ are proportional to $\proj{\psi}^{\otimes 2}$, up to 1-local terms. In the latter case, the overall interactions must all have been of the form $\proj{\psi}^{\otimes \ell}$, so the whole Hamiltonian is stoquastic.

% ------------------------------------------------------------------------------

\subsection{Overview of proof of Theorems \ref{thm:sud}, \ref{thm:proj}, \ref{thm:su2} and \ref{thm:bilinbiq}}

The techniques required to prove universality of interactions without free local terms are very different, and in general this setting is much more challenging. Given the symmetry displayed by the interactions we consider, we need to consider some notion of encoding in order to implement arbitrary effective interactions. In the case of the $SU(d)$ Heisenberg interaction, we proceed by using a perturbative gadget to encode a qubit within the 2-dimensional ground space of a system of $2d$ qudits; this generalises a similar (but significantly simpler) gadget used for the case $d=2$ in~\cite{cubitt17}. Interactions across pairs of qudits within the gadget implement effective $X$ and $Z$ interactions, while interactions across two gadgets can be used to implement a non-trivial 2-qubit interaction, which is enough to prove universality using the results of~\cite{cubitt17,Piddock-Montanaro}. In order to analyse the gadget's behaviour, we need to use the representation theory of the Lie algebra $\mathfrak{su}(N)$, and in particular analysis of quadratic Casimir operators~\cite{FuchsSchweigert}, which are operators of the form $\sum_a R(T^a) R(T^a)$ for some representation $R$ of the generators $T^a$ of $\mathfrak{su}(d)$. The Hamiltonian corresponding to the $SU(d)$ Heisenberg interaction on the complete graph on $k$ qudits turns out to have a close connection to the Casimir operator corresponding to the representation $R(T^a) = \sum_{i=1}^k T^a_i$, whose spectral properties are well-understood, and which has beautiful algebraic features that enable suitable gadget weights to be determined for any $d$.

Theorem \ref{thm:proj} is proven using a gadget that shows that, when $P$ is the projector onto an entangled state of two qudits, $\{P\}$-Hamiltonians can simulate $\{P'\}$-Hamiltonians for some $P'=\proj{\psi'}$ where either $\ket{\psi'}$ is an entangled state of two {\em qubits}, in which case universality follows from Theorem \ref{thm:qubits}; or $\ket{\psi'} =\frac{1}{\sqrt{d}}\sum_i \ket{i}\ket{i}$, in which case universality can be shown to follow from universality of the $SU(d)$-Heisenberg interaction (Theorem \ref{thm:sud}).

The gadget for the $SU(2)$ Heisenberg interaction $h$ also relies on properties of the corresponding Casimir operator, but is more complicated than the $SU(d)$ case. Here the key technical step is to give a gadget that allows $h^2$ interactions to be simulated, given access to $h$ interactions; once this is achieved, it is not too hard to show that for any $d$, this allows the $SU(2)$ Heisenberg interaction to be simulated in local dimension 3 (qutrits). Applying the $h \mapsto h^2$ gadget again, we can produce the interaction $h + h^2$, which (in local dimension 3) is the same as the $SU(3)$ Heisenberg interaction, and hence universal. The analysis of this gadget depends on fourth-order perturbation theory, for which we need to prove a new general simulation lemma based on the Schreiffer-Wolff transformation~\cite{Bravyi-DiVincenzo-Loss}. Previous work gave general simulation lemmas for up to third-order perturbations~\cite{Bravyi-Hastings}, but extending this line of argument to fourth-order is more complex technically; in particular, there are non-trivial interference effects between different gadgets to take into account. We thus hope that this result will find other applications elsewhere.

We note that higher order perturbation theory has been considered before in the literature in slightly different settings, mostly in a framework where only the ground state energy is reproduced; for example \cite{jordan08} considers perturbation theory at arbitrary order. 
Although the contribution of the fourth order term in a Schreiffer-Wolff perturbative series has been considered before \cite{cao17}, we are not aware of any explicit demonstration of how the interactions must be chosen such that this fourth order term dominates as in Lemma~\ref{lem:fourthorder}.
Cross gadget interference has previously been seen before for certain parameter regimes of low strength Hamiltonians \cite{cao13}, where it can be easily shown to disappear simply by increasing the strength of the interactions; whereas in Lemma~\ref{lem:4thorderinterference}, the cross gadget terms are independent of the strength of the Hamiltonian.

Finally, for the remaining bilinear-biquadratic Heisenberg interactions in dimension 3, we use different gadgets depending on the value of $\theta$, which we can assume is within the range $[0,\pi]$ because we are free to choose the signs of interactions arbitrarily. When $\theta \in (0,\arctan 1/3) \cup (\pi/4, \pi)$ and $\theta \neq \arctan 2$, then there exists an entangled state $\ket{\psi}$ which is either the unique ground state or the unique highest excited state of $h^{(\theta)}$. Using a perturbative gadget to effectively project some qudits onto $\ket{\psi}$, we can obtain a new interaction $h^{(\theta')}$ for some $\theta' \neq \theta$. Taking a linear combination of these two interactions, we can simulate the $SU(3)$ Heisenberg interaction. When $\theta \in (\arctan 1/3,\arctan 5)$, $h^{(\theta)}$ has a 3-dimensional ground space. We encode a qutrit within this subspace of two physical qutrits, and use $h^{(\theta)}$ interactions across pairs of qutrits to simulate the $SU(3)$ Heisenberg interaction across logical qutrits. These ranges encompass all values of $\theta$ except $\theta=\arctan 1/3$. In this last special case, $h^{(\theta)}$ corresponds to the well-studied AKLT interaction~\cite{AKLT}. Here the ground space of $h^{(\theta)}$ is 4-dimensional, but we are able to construct a mediator qutrit gadget which effectively projects 3 qutrits into the unique ground state of a 3 qutrit AKLT Hamiltonian. This again allows us to simulate the $SU(3)$ Heisenberg interaction.

% ------------------------------------------------------------------------------

\section{Summary of techniques}
\label{sec:prelims}

%\subsection{Notation}

%We say that a Hamiltonian $H \in L((\C^d)^{\otimes n})$ is $k$-local (or just ``local'' if $k=O(1)$) if it can be written as a sum of terms such that each $h_i$ acts non-trivially on at most $k$ subsystems of $(\C^d)^{\otimes n}$.
%That is, $h_i\in L((\C^d)^{\otimes k})$ and $H = \sum_i h_i\otimes I$ where the identity in each term in the sum acts on the subsystems where that $h_i$ does not.

Next, we give the required definitions to state our results formally, describe previous results that we use, and exemplify our results by giving a simple example of a simulation. We then proceed to a full technical presentation of the remainder of our results.

% ------------------------------------------------------------------------------
\subsection{Definitions}

We first formally define the notions of simulation and universality that we will use. For an arbitrary Hamiltonian $H \in L(\C^d)$, we let $P_{\le \Delta(H)}$ denote the orthogonal projector onto the subspace $S_{\le \Delta(H)} := \linspan \{ \ket{\psi} : H\ket{\psi}=\lambda\ket{\psi}, \lambda \le \Delta \}$.
We also let $H'|_{\le \Delta(H)}$ denote the restriction of some other arbitrary Hamiltonian $H'$ to $S_{\le \Delta(H)}$, and write $H|_{\le \Delta} := H|_{\le \Delta(H)}$ and $H_{\le \Delta} := H P_{\le \Delta(H)}$. We let $L(\mathcal{H})$ denote the set of linear operators acting on a Hilbert space $\mathcal{H}$, and use the standard notation $[A,B] := AB-BA$ and $\{A,B\} := AB+BA$ for the commutator and anticommutator of $A$ and $B$, respectively.

\begin{dfn}[Special case of definition in~\cite{cubitt17}; variant of definition in~\cite{Bravyi-Hastings}] 
  \label{dfn:sim}
  We say that $H'$ is a $(\Delta,\eta,\epsilon)$-simulation of $H$ if there exists a local isometry $V = \bigotimes_i V_i$ such that:
  \begin{enumerate}
  \item 
    There exists an isometry $\widetilde{V}$ such that $\widetilde{V} \widetilde{V}^\dag = P_{\le \Delta(H')}$ and $\|\widetilde{V} - V\| \le \eta$;
  \item 
    $\| H'_{\le \Delta} - \widetilde{V}H\widetilde{V}^\dag \| \le \epsilon$.
  \end{enumerate}
  We say that a family $\mathcal{F}'$ of Hamiltonians can simulate a family $\mathcal{F}$ of Hamiltonians if, for any $H \in \mathcal{F}$ and any $\eta,\epsilon >0$ and $\Delta \ge \Delta_0$ (for some $\Delta_0 > 0$), there exists $H' \in \mathcal{F}'$ such that $H'$ is a $(\Delta,\eta,\epsilon)$-simulation of $H$.
  We say that the simulation is efficient if, in addition, for $H$ acting on $n$ qudits, $\|H'\| = \poly(n,1/\eta,1/\epsilon,\Delta)$; $H'$ is efficiently computable given $H$, $\Delta$, $\eta$ and $\epsilon$; and each isometry $V_i$ maps to $O(1)$ qudits.
\end{dfn}

The first part of Definition \ref{dfn:sim} says that $H$ can be mapped exactly into the ground space of $H'$ by some ``encoding'' isometry $\widetilde{V}$ which is close to a local isometry $V$. The second part says that the low-energy part of $H'$ is close to an encoded version of $H$. In~\cite{cubitt17} a more general notion of encoding was used, which allowed for complex Hamiltonians to be encoded as real Hamiltonians, for example; here we will not need this directly. (However, as we make use of the results of~\cite{cubitt17}, we do use this notion of encoding indirectly.)

\begin{dfn}[\cite{cubitt17}]
\label{dfn:universal}
  We say that a family of Hamiltonians is \emph{universal} if \emph{any} (finite-dimensional) Hamiltonian can be simulated by a Hamiltonian from the family.
  We say that the universal simulator is \emph{efficient} if the simulation is efficient for all local Hamiltonians.
\end{dfn}

Here all simulations we develop will be efficient, so whenever we say ``universal'', we mean ``efficiently universal'' in the above sense.

% ------------------------------------------------------------------------------

\subsection{Perturbative gadgets}

The main technique we will use to prove universality will be the remarkably powerful concept of perturbative gadgets~\cite{Kempe-Kitaev-Regev}. Let $\mathcal{H}_{\text{sim}}$ be a Hilbert space decomposed as $\mathcal{H}_{\text{sim}} = \mathcal{H}_+ \oplus \mathcal{H}_-$, and let $\Pi_{\pm}$ denote the projector onto $\mathcal{H}_{\pm}$.
For any linear operator $O$ on $\mathcal{H}_{\text{sim}}$, write
\begin{equation}
  O_{--} = \Pi_- O \Pi_-,\;\;\;\; O_{-+}
  = \Pi_- O \Pi_+,\;\;\;\; O_{+-}
  = \Pi_+ O \Pi_-,\;\;\;\; O_{++} = \Pi_+ O \Pi_+.
\end{equation}
Throughout, let $H_0$ be a Hamiltonian such that $H_0$ is block-diagonal with respect to the split $\mathcal{H}_+ \oplus \mathcal{H}_-$, $(H_0)_{--} = 0$, and $\lambda_{\min}((H_0)_{++}) \ge 1$, where $\lambda_{\min}(H)$ denotes the minimal eigenvalue of $H$.

Slight variants of the following lemmas were shown in~\cite{Bravyi-Hastings}, building on previous work~\cite{Oliveira-Terhal,Bravyi-DiVincenzo-Loss}:

\begin{lem}[First-order simulation~\cite{Bravyi-Hastings}]
  \label{lem:firstorder}
  Let $H_0$ and $H_1$ be Hamiltonians acting on the same space. Suppose there exists a local isometry $V$ such that $\Im(V)=\mathcal{H}_-$ and
  \begin{equation}
V H_{\operatorname{target}} V^\dag = (H_1)_{--}.
  \end{equation}
  Then $H_{\operatorname{sim}} = \Delta H_0 + H_1$ $(\Delta/2,\eta,\epsilon)$-simulates $H_{\operatorname{target}}$, provided that the bound $\Delta \ge O(\|H_1\|^2/\epsilon + \|H_1\| / \eta)$ holds.
\end{lem}

\begin{lem}[Second-order simulation~\cite{Bravyi-Hastings}]
  \label{lem:secondorder}
  Let $H_0$, $H_1$, $H_2$ be Hamiltonians acting on the same space, such that: $\max\{\|H_1\|,\|H_2\|\} \le \Lambda$; $H_1$ is block-diagonal with respect to the split $\mathcal{H}_+ \oplus \mathcal{H}_-$; and $(H_2)_{--} =0$.
  Suppose there exists a local isometry $V$ such that $\Im(V)=\mathcal{H}_-$ and 
  \begin{equation}
    V H_{\operatorname{target}} V^\dag = (H_1)_{--} - (H_2)_{-+} H_0^{-1} (H_2)_{+-}.
  \end{equation}
  Then $H_{\operatorname{sim}} = \Delta H_0 + \Delta^{1/2} H_2 + H_1$ $(\Delta/2,\eta,\epsilon)$-simulates $H_{\operatorname{target}}$, provided that $\Delta \ge O(\Lambda^6/\epsilon^2 + \Lambda^2/\eta^2)$.
\end{lem}

\begin{lem}[Third-order simulation~\cite{Bravyi-Hastings}]
  \label{lem:thirdorder}
  Let $H_0$, $H_1$, $H_1'$, $H_2$ be Hamiltonians acting on the same space, such that: $\max\{\|H_1\|,\|H_1'\|,\|H_2\|\} \le \Lambda$; $H_1$ and $H_1'$ are block-diagonal with respect to the split $\mathcal{H}_+ \oplus \mathcal{H}_-$; $(H_2)_{--}=0$.
  Suppose there exists a local isometry $V$ such that $\Im(V)=\mathcal{H}_-$ and 
  \begin{equation}
    V H_{\operatorname{target}} V^\dag = (H_1)_{--} + (H_2)_{-+} H_0^{-1} (H_2)_{++} H_0^{-1} (H_2)_{+-}
  \end{equation}
  and also that
  \begin{equation}
    (H_1')_{--} = (H_2)_{-+} H_0^{-1} (H_2)_{+-}.
  \end{equation}
  Then $H_{\operatorname{sim}} = \Delta H_0 + \Delta^{2/3} H_2 + \Delta^{1/3} H_1' + H_1$ $(\Delta/2,\eta,\epsilon)$-simulates $H_{\operatorname{target}}$, provided that $\Delta \ge O(\Lambda^{12}/\epsilon^3 + \Lambda^3/\eta^3)$.
\end{lem}

We will often apply the simulation results in these lemmas to many individual interactions within a larger overall Hamiltonian, in parallel. For the gadgets we will use, it was shown in~\cite[Lemma 36]{cubitt17} (following similar arguments in previous work, e.g.~\cite{Oliveira-Terhal,Bravyi-Hastings}) that the overall simulation produced is what one would expect (i.e.\ a sum of the individual simulated interactions, without unexpected interference between the terms). In addition, the simulations that we use will either associate a fixed number of ancilla (``mediator'') qudits with each interaction, or encode each logical qudit within a fixed number of physical qudits. In each such case, the overall isometry $V$ is easily seen to be a tensor product of local isometries as required for Definition~\ref{dfn:sim}; for readability, we leave this isometry implicit.

Later on, we will need a new {\em fourth-order} simulation lemma. As this is more technical to state (and its proof has some additional complications involving interference), we defer it to Section~\ref{sec:fourthorder}.

% ------------------------------------------------------------------------------

\subsection{Example: the AKLT interaction}

To see how the above simulation results can be used to prove universality, we give a simple example of how the AKLT interaction~\cite{AKLT} can simulate the $SU(3)$ Heisenberg interaction. The AKLT interaction $\haklt$ is defined in local dimension $d=3$ (spin 1) by $\haklt:=3h+h^2$, where $h$ is the $SU(2)$ Heisenberg interaction defined in (\ref{eq:spin1intro}).

\begin{lem}
The AKLT interaction $\haklt:=3h+h^2$ is universal.
\label{lem:AKLTmodel}
\end{lem}

\begin{proof}
We will use a gadget construction to show that $\haklt$ can simulate the $SU(3)$ invariant interaction $h+h^2$, which is shown to be universal in Theorem~\ref{thm:sud}.
We will use Lemma~\ref{lem:secondorder} and construct a second-order mediator qutrit gadget involving 3 mediator qutrits labelled 3, 4, 5 that will result in an effective interaction between qutrits 1 and 2.
Let $H_0=\haklt_{34}+\haklt_{45}+\haklt_{35}+6I$, which has a unique ground state $\ket{\psi}$ on qutrits 3, 4, 5 
\[\ket{\psi}=\frac{1}{\sqrt{6}}\left(\ket{012}+\ket{120}+\ket{201}-\ket{021}-\ket{210}-\ket{102}\right),\]
the completely antisymmetric state on 3 qutrits. Define $\Pi$ to be the projector onto the ground space of $H_0$, and let $H_2=\lambda_2\left(\haklt_{13}+\haklt_{23}-\frac{8}{3}I\right)$ for some $\lambda_2 \in \R$. Then one can check (either by hand or using a computer algebra package) that $\Pi H_2 \Pi=0$ and 
%\[-\Pi H_2 H_0^{-1} H_2 \Pi = -\frac{2\lambda_2^2}{27}\left(23h_{12}+h_{12}^2+\frac{272}{3}I\right)\Pi\]
\[-\Pi H_2 H_0^{-1} H_2 \Pi = -\frac{2\lambda_2^2}{27}\left(23h_{12}+h_{12}^2+\frac{136}{3}I\right)\Pi.\]
Let $H_1=\lambda_1 \haklt_{12}$ for some $\lambda_1 \in \R$ so that $\Pi H_1\Pi=\lambda_1\haklt_{12}\Pi$. Then by Lemma~\ref{lem:secondorder}, choosing $\lambda_1=22$ and $\lambda_2=\sqrt{27}$ we can simulate
\[\Pi H_1 \Pi -\Pi H_2 H_0^{-1} H_2 \Pi =20(h_{12}+h^2_{12})-\frac{272}{3}I\]
which one can check is the $SU(3)$ Heisenberg interaction as desired, up to rescaling and deletion of an identity term. Note that this can only produce positively-weighted interactions, but Hamiltonians of this restricted form are indeed proven universal in Theorem~\ref{thm:sud}.
 \end{proof}

% ------------------------------------------------------------------------------

\section{Fourth-order perturbative gadgets}
\label{sec:fourthorder}

We will need the following lemma, which we prove for the first time here (and hence state a bit more generally than the above simulation lemmas, although we will only need $\epsilon=0$ on the right-hand side of (\ref{eq:4thorderstate})). The proof is technical, and hence (as with the subsequent lemma) deferred to Appendix~\ref{app:fourthorderproofs}.

\begin{lem}[Fourth-order simulation]
\label{lem:fourthorder}
Let $H_0$, $H_1$, $H_2$, $H_3$, $H_4$ be Hamiltonians acting on the same space, such that: $\max\{\|H_1\|,\|H_2\|,\|H_3\|,\|H_4\|\} \le \Lambda$; $H_2$ and $H_3$ are block-diagonal with respect to the split $\mathcal{H}_+ \oplus \mathcal{H}_-$; $(H_4)_{--}=0$. Suppose there exists a local isometry $V$ such that $\Im(V)=\mathcal{H}_-$ and 
\begin{equation} \| V H_{\operatorname{target}} V^\dag - \Pi_-\left(H_1 +H_4 H_0^{-1} H_2 H_0^{-1} H_4 -H_4 H_0^{-1} H_4 H_0^{-1} H_4 H_0^{-1} H_4\right)\Pi_- \| \le \epsilon/2 
\label{eq:4thorderstate}
\end{equation}
and also that
\begin{equation} (H_2)_{--} = \Pi_- H_4 H_0^{-1}H_4 \Pi_- \quad \text{ and } \quad (H_3)_{--} = -\Pi_- H_4 H_0^{-1}H_4 H_0^{-1}H_4 \Pi_-. 
\label{eq:H2H3}
\end{equation}
Then $H_{\operatorname{sim}} = \Delta H_0 + \Delta^{3/4} H_4 + \Delta^{1/4} H_3 + \Delta^{1/2}H_2+H_1$ $(\Delta/2,\eta,\epsilon)$-simulates $H_{\operatorname{target}}$, provided that $\Delta \ge O(\Lambda^{20}/\epsilon^4+\Lambda^4/\eta^4)$.
\end{lem}

%We will often apply the simulation results in these lemmas to many individual interactions within a larger overall Hamiltonian, in parallel. For the gadgets we will use, it was shown in~\cite[Lemma 36]{cubitt17} (following similar arguments in previous work, e.g.~\cite{Oliveira-Terhal,Bravyi-Hastings}) that the overall simulation produced is what one would expect (i.e.\ a sum of the individual simulated interactions, without unexpected interference between the terms).

For fourth-order gadgets, unlike the gadgets analysed in previous work, it is unfortunately not the case that one can disregard interference between different gadgets applied in parallel; there are additional terms generated by interference between gadgets. We calculate this interference in the following lemma.

\begin{lem}
\label{lem:4thorderinterference}
Consider a Hilbert space $\cH=\cH_0 \otimes \bigotimes_{i\ge 1} \cH_i$ with multiple fourth-order mediator gadgets labelled by $i\ge 1$, each with heavy Hamiltonian $H_0^{(i)}$ which acts non-trivially only on $\cH_i$, and interaction terms $H_1^{(i)}$, $H_2^{(i)}$, $H_3^{(i)}$, $H_4^{(i)}$ which act non-trivially only on $\cH_i \otimes \cH_0$. Let $\Pi_-^{(i)}$ denote the projector onto the ground space of $H_0^{(i)}$, and $\Pi_+^{(i)} = I - \Pi_-^{(i)}$.
Suppose that for each $i$, these terms satisfy the conditions of Lemma~\ref{lem:fourthorder}; in particular, $H_0^{(i)} \Pi_-^{(i)} = 0$, $H_2^{(i)}$ and $H_3^{(i)}$ are block diagonal with respect to the $\Pi_-^{(i)}$, $\Pi_+^{(i)}$ split, $\Pi^{(i)}_- H_4^{(i)} \Pi^{(i)}_-=0$ and 
\[ \Pi_-^{(i)}H_2^{(i)}\Pi_-^{(i)} = \Pi^{(i)}_- H_4^{(i)} (H_0^{(i)})^{-1}H_4^{(i)} \Pi^{(i)}_- \text{ and }
\Pi^{(i)}_-H_3^{(i)} \Pi_-^{(i)} = -\Pi^{(i)}_- H_4^{(i)} (H_0^{(i)})^{-1}H_4^{(i)} (H_0^{(i)})^{-1}H_4^{(i)} \Pi^{(i)}_-.\]
%
%Let $H_0=\sum_{i}H_0^{(i)}$, $H_1=\sum_{i}H_1^{(i)}$, $H_2=\sum_{i}H_2^{(i)}$, $H_3=\sum_{i}H_3^{(i)}$ and $H_4=\sum_{i}H_4^{(i)}$.
For each $j \in \{0,\dots,4\}$, let $H_j = \sum_i H_j^{(i)}$, and let  $\Lambda \ge\max\{\|H_1\|,\|H_2\|,\|H_3\|,\|H_4\|\}$.

Suppose there exists a local isometry $V$ such that $\Im(V)$ is the ground space of $H_0$ and also $\|V H_{\operatorname{target}}V^{\dagger}-M\|\le \epsilon/2$, where
\begin{align*}M=\sum_i& \Pi_-\left(H_1^{(i)}+H_4^{(i)} (H_0^{(i)})^{-1} H_2^{(i)} (H_0^{(i)})^{-1} H_4^{(i)} -H_4^{(i)} (H_0^{(i)})^{-1} H_4^{(i)} (H_0^{(i)})^{-1} H_4^{(i)} (H_0^{(i)})^{-1} H_4^{(i)}\right)\Pi_-\\
&+\sum_{i\neq j}\Pi_-\Big(H_4^{(i)} (H_0^{(i)})^{-1} H_4^{(j)}(H_0^{(j)})^{-1}H_4^{(j)} (H_0^{(i)})^{-1} H_4^{(i)}\\[-11pt]
&\qquad\qquad\qquad-H_4^{(i)} (H_0^{(i)})^{-1} H_4^{(j)} (H_0^{(i)}+H_0^{(j)})^{-1} H_4^{(j)} (H_0^{(i)})^{-1} H_4^{(i)}\\
&\qquad\qquad\qquad-H_4^{(i)} (H_0^{(i)})^{-1} H_4^{(j)} (H_0^{(i)}+H_0^{(j)})^{-1} H_4^{(i)} (H_0^{(j)})^{-1} H_4^{(j)}\Big)\Pi_-
\end{align*}
and $\Pi_-$  is the projector onto the ground space of $H_0$.

Then $\Delta H_0 + \Delta^{3/4} H_4 + \Delta^{1/4} H_3 + \Delta^{1/2}H_2+H_1$ $(\Delta/2,\eta,\epsilon)$ simulates $H_{\operatorname{target}}$, provided that $\Delta \ge O(\Lambda^{20}/\epsilon^4+\Lambda^4/\eta^4)$.
\end{lem}

Note that the first line of the simulated Hamiltonian is what one would expect when summing the contributions of each of the gadgets separately. The other terms are in general not zero and may be thought of as the cross-gadget interference.

We will only need to use Lemma~\ref{lem:4thorderinterference} via the following simplified corollary.

\begin{cor}
\label{cor:4thorderinterference}
Suppose the conditions of Lemma~\ref{lem:4thorderinterference} hold, and in addition $H_0^{(i)}H_4^{(i)}\Pi_-=H_4^{(i)}\Pi_-$ for all $i$ (for example when $H_0^{(i)}$ is a projector).
Then the expression for $M$ is given by
\[M=\sum_i \Pi_-\left(H_1^{(i)}+H_4^{(i)} H_2^{(i)} H_4^{(i)} -H_4^{(i)}H_4^{(i)} (H_0^{(i)})^{-1} H_4^{(i)}  H_4^{(i)}\right)\Pi_- 
-\frac{1}{2} \sum_{i<j}\Pi_-\left[H_4^{(i)},H_4^{(j)}\right]^2 \Pi_- \]
\end{cor}
\begin{proof}
For $i \neq j$, by the additional assumption of the present corollary $(H_0^{(i)}+H_0^{(j)})^{-1}H_4^{(i)}H_4^{(j)}\Pi_-=\frac{1}{2}H_4^{(i)}H_4^{(j)}\Pi_-$, so the expression for the cross-gadget interference from Lemma~\ref{lem:4thorderinterference} simplifies to 
\begin{align*}
&\sum_{i \neq j} \Pi_- \left(H_4^{(i)} H_4^{(j)} H_4^{(j)} H_4^{(i)} -\frac{1}{2}\left(H_4^{(i)} H_4^{(j)} H_4^{(j)} H_4^{(i)} +H_4^{(i)} H_4^{(j)} H_4^{(i)} H_4^{(j)} \right)\right)\Pi_-\\
=&\frac{1}{2}\sum_{i \neq j} \Pi_- \left(H_4^{(i)} H_4^{(j)} H_4^{(j)} H_4^{(i)} -H_4^{(i)} H_4^{(j)} H_4^{(i)} H_4^{(j)} \right)\Pi_- =-\frac{1}{2} \sum_{i<j}\Pi_-\left[H_4^{(i)},H_4^{(j)}\right]^2 \Pi_-
\end{align*}
where we note that the sum over $i \neq j$ includes both cases $i < j$ and $i>j$.
\end{proof}

% ------------------------------------------------------------------------------

\section{LA-universal Hamiltonians}

We first prove LA-universality (or otherwise) of various classes of interactions, before bringing these results together into a full classification theorem by showing that every interaction fits into one of these classes. Before embarking on the proof, we observe that for any interaction $H$, we can delete its 1-local part by using our free 1-local terms. This corresponds to replacing $H$ with
\begin{equation}
\label{eq:2localpart} H' = H - \frac{I}{d}\otimes \tr_1(H) - \tr_2(H)\otimes \frac{I}{d} + \tr(H)\frac{I\otimes I}{d^2}. \end{equation}
We call $H'$ the 2-local part of $H$. For a fixed basis $T^a$ of Hermitian $d \times d$ matrices, we can decompose $H'=\sum_{a,b} M_{ab} T^a \otimes T^b$ for some real $d^2 \times d^2$ matrix $M$. We define the \emph{2-local rank} of $H$ to be the rank of $M$. 

Note that this definition is independent of the choice of basis $T^a$. Suppose we instead write $H'=\sum_{a,b}\tilde M_{ab} S^{a}\otimes S^{\prime b}$ for two other bases $\{S^a\}_a$ and $\{S^{\prime b}\}_b$ of Hermitian $d\times d$ matrices. Since these are bases there must exist invertible matrices $R$ and $R'$ such that $T^a=\sum_{b} R_{ab} S^a=\sum_{b} R'_{ab} S^{\prime b}$. Then 
\begin{align*}
H'&=\sum_{a,b}\tilde M_{ab} S^{a}\otimes S^{\prime b} \\
&=\sum_{c,d} M_{cd} T^c \otimes T^d=\sum_{a,b}( \sum_{c,d}R_{ca} M_{cd} R_{db})  S^a \otimes S^{\prime b}
\end{align*}
and thus  $\rank(\tilde{M})=\rank(R^{T}MR') =\rank(M)$ since $R$ and $R'$ are both full rank.

We now move on to the first case of the proof, diagonal interactions.
% ------------------------------------------------------------------------------

\subsection{Interactions diagonalisable by local unitaries}

\begin{lem}
\label{lem:diagonal}
Let $H$ be a nonzero diagonal 2-qudit interaction. If the 2-local rank of $H$ is $\geqslant 2$, then $H$ is LA-universal; otherwise, $H$ is LA-stoquastic-universal.
\end{lem}

%The interaction $H'$ occurring in the statement of this lemma is simply the interaction obtained by deleting the 1-local part from $H$. We can assume that we have access to $H'$ by using our free 1-local terms to effectively delete this part from $H$.

\begin{proof}
First note that we can use 1-local terms to replace $H$ with its 2-local part, as in (\ref{eq:2localpart}). This still results in a diagonal interaction and allows us to assume that $\tr_1(H)=0=\tr_2(H)$. Let $H$ be given by $H=\sum_{i,j=1}^d A_{ij} \proj{i} \otimes \proj{j}$ for some $d \times d$ matrix $A$. Then the 2-local rank of $H$ is given by $\rank(A)$.
Next observe that we can assume that the interaction $H$ is either symmetric or antisymmetric with respect to permuting the qudits on which it acts, because we can apply it in either direction, with positive or negative weights. So we obtain either $H_{ij} + H_{ji}$ or $H_{ij}-H_{ji}$, corresponding to mapping $A$ either to $A+A^T$ or $A-A^T$. This cannot affect the condition on the rank of $A$, because
\[ \rank(A) = \rank((A+A^T) + (A-A^T)) \le \rank(A+A^T) + \rank(A-A^T); \]
if $\rank(A) \ge 2$, then either $\max\{\rank(A+A^T),\rank(A-A^T)\} \ge 2$, or $\rank(A+A^T) = \rank(A-A^T) = 1$; but this latter possibility cannot occur because $A-A^T$ is skew-symmetric, so $\rank(A-A^T) \neq 1$.

We will apply Lemma~\ref{lem:firstorder} by using heavily-weighted local terms to effectively project each subsystem on which $H$ acts into a 2-dimensional subspace, which will encode a qubit. Such a projection can be described by a $2 \times d$ matrix $P$. We aim to produce an effective 2-qubit interaction $H'$ which is universal. As we can apply arbitrary local terms, we can project each qudit onto an arbitrary 2-dimensional subspace $S$ by choosing a ``heavy'' Hamiltonian $H_0 = \sum_i H^P_i$ in Lemma~\ref{lem:firstorder} such that $H^P$ has $S$ as its ground space. The local isometry $V$ in the lemma is just given by $P^{\dagger}$.

The result of projecting $H$ is the 2-qubit interaction
\[ H' = \sum_{i,j=1}^d A_{ij} \left( P \proj{i} P^{\dag}\right) \otimes \left( P \proj{j} P^\dag \right) = \sum_{i,j=1}^d A_{ij} \left( \sum_{k=0}^3 \beta_{ik} \sigma^k \right) \otimes \left( \sum_{\ell=0}^3 \beta_{j\ell} \sigma^\ell \right), \]
for some real coefficients $\beta_{ik}$ such that
\[ \beta_{ik} = \frac{1}{2} \tr[P \proj{i} P^\dag \sigma^k]. \]
Reordering the sums, we obtain
\[ H' = \sum_{k,\ell=0}^3 \left( \sum_{i,j=1}^d \beta_{ik} A_{ij} \beta_{j\ell} \right) \sigma^k \otimes \sigma^\ell = \sum_{k,\ell=0}^3 \bracket{\beta_k}{A}{\beta_{\ell}} \sigma^k \otimes \sigma^\ell, \]
where we define the unnormalised vector $\ket{\beta_k} = \sum_{i=1}^d \beta_{ik} \ket{i}$. We can write down explicit expressions for these vectors as
\[ \beta_{i1} = \Re( P_{1i}^* P_{2i} ),\;\;\;\; \beta_{i2} = \Im( P_{1i}^* P_{2i} ),\;\;\;\; \beta_{i3} = \frac{1}{2}\left( |P_{1i}|^2 - |P_{2i}|^2\right). \]
It was shown in~\cite{Cubitt-Montanaro,cubitt17} that an interaction of the form $\sum_{k,\ell=1}^3 M_{k\ell} \sigma^k \otimes \sigma^\ell$ is universal if the $3 \times 3$ matrix $M$ has rank at least 2. Our goal will be to choose the vectors $\ket{\beta_k}$ to achieve this.

If $A$ is symmetric, we can expand it as a weighted sum of projectors onto real, orthonormal eigenvectors $\ket{\eta_i}$; as $\rank(A) \ge 2$, there exist $\ket{\eta_1}$, $\ket{\eta_2}$ with nonzero eigenvalues. If $A$ is skew-symmetric, there exist real, orthonormal vectors $\ket{\eta_i}$ such that $\bracket{\eta_i}{A}{\eta_i} = 0$ for all $i$, and $\bracket{\eta_1}{A}{\eta_2} = -\bracket{\eta_2}{A}{\eta_1} \neq 0$ (see e.g.~\cite{thompson88}). Hence, in either the symmetric or skew-symmetric case, in order to achieve that $M$ has rank at least 2, it is sufficient to have $\ket{\beta_1} = \ket{\eta_1}$ and $\ket{\beta_3} = \ket{\eta_2}$. This fixes a $2 \times 2$ submatrix of $M$ to be either diagonal (and rank 2), or proportional to $\sm{0 & 1\\-1 & 0}$. So we want to produce a matrix $P$ that achieves $\beta_{i1} = \ip{i}{\eta_1}$, $\beta_{i3} = \ip{i}{\eta_2}$ for all $i$.

%Further, as $A'$ has been projected onto the orthogonal complement of $\proj{+}$, all of the $\ket{\eta_i}$ vectors with nonzero eigenvalues are orthogonal to $\ket{+}$. (This holds in the skew-symmetric case too: follows from the fact that $A'$ and $(A')^T$ both have $\ket{+}$ as a zero eigenvector.)

If we can find a {\em real} matrix $P$ that achieves this, it will automatically have orthonormal rows (up to an overall normalising constant), and also the entries of $M$ outside a $2\times 2$ submatrix will be zero. To see this, first note that $\ket{\eta_1}$ and $\ket{\eta_2}$ are orthogonal to $\ket{+} = \sum_{i=1}^d \ket{i}$. This holds because $\tr_1(H) = \sum_{j=1}^d \left(\sum_{i=1}^d A_{ij}\right) \proj{j} = 0$, and similarly for $\tr_2(H)$, so $A \ket{+} = A^T \ket{+} = 0$. So as $\ket{\beta_1} = \ket{\eta_1}$ and $\ket{\beta_3} = \ket{\eta_2}$, $\sum_i \beta_{i1} = \sum_i \beta_{i3} = 0$, implying that $\sum_i P_{1i} P_{2i} = 0$ and $\sum_i P_{1i}^2 = \sum_i P_{2i}^2$. We can find an explicit expression for each element of $P$ by solving the simultaneous equations
\[ P_{1i} P_{2i} = \gamma_i,\;\;\;\; \frac{1}{2}\left(P_{1i}^2 - P_{2i}^2\right) = \delta_i, \]
where we write $\gamma_i = \ip{i}{\eta_1}$, $\delta_i = \ip{i}{\eta_2}$. It can readily be verified that the following is a valid solution:
\[
\begin{cases}
P_{1i} = 0, P_{2i} = \sqrt{-2\delta_i} & \text{if $\gamma_i = 0$ and $\delta_i \le 0$}\\
%P_{1i} = 0, P_{2i} = \sqrt{\delta_i} & \text{if $\gamma_i = 0$ and $\delta_i \ge 0$}\\
%P_{1i} = \sqrt{\gamma_i}, P_{2i} = -\sqrt{\gamma_i} & \text{if $\gamma_i \le 0$ and $\delta_i = 0$}\\
%P_{1i} = \sqrt{\gamma_i}, P_{2i} = \sqrt{\gamma_i} & \text{if $\gamma_i \ge 0$ and $\delta_i = 0$}\\
P_{1i} = \sqrt{\delta_i + \sqrt{\gamma_i^2+\delta_i^2} }, P_{2i} = \frac{\gamma_i}{\sqrt{\delta_i + \sqrt{\gamma_i^2+\delta_i^2}}} & \text{otherwise.}
\end{cases}
\]
Thus $H$ is LA-universal. This completes the proof of the case $\rank(A) \ge 2$. If $\rank(A) = 1$, we know that there exists an eigenvector $\ket{\eta_1}$ with nonzero eigenvalue, and can take $\ket{\eta_2}$ to be an arbitrary orthogonal vector. Almost all the above steps go through, but we end up producing a matrix $M$ such that $\rank(M) \ge 1$. This case is known to be stoquastic-universal~\cite{Bravyi-Hastings,cubitt17}.
\end{proof}

\begin{lem}
\label{lem:3evalues}
 Let $H=A\otimes A$ be a 2-qudit interaction such that $A$ has three distinct eigenvalues. Then $H$ is LA-universal.
\end{lem}

\begin{proof}
We will use a third order mediator qudit perturbation involving three qudits labelled $1,2,3$, where $3$ will be a mediator qudit. We work in the eigenbasis of $A$, so that $A=\sum_i \lambda_i \proj{i}$ . By the addition of 1-local terms of the form $\mu A\otimes I +\mu I \otimes A+\mu^2 I\otimes I$, it is possible to shift the spectrum of $A$ by a constant $\mu$. Since $A$ has three distinct eigenvalues, we may therefore assume wlog (relabelling eigenvectors if necessary) that $A$ has eigenvalues $\lambda_0<0$ and $\lambda_1>0$ such that $\lambda_0 + \lambda_1 > 0$.

  Let $H_2= A_1A_3+A_2 A_3$ and let $H_0=I-\proj{\psi}$ act only on the mediator qudit $3$, where $\ket{\psi}=\sqrt{\lambda_1}\ket{0}+\sqrt{-\lambda_0}\ket{1}$ . Note that $\ket{\psi}$ has been chosen so that
  \begin{equation}
    \label{eq:expAzero}
  \bra{\psi}A\ket{\psi}=0,\;\;\;\; \bra{\psi}A^2\ket{\psi}>0,\;\;\;\; \bra{\psi}A^3\ket{\psi}> 0,
  \end{equation}
  which implies that $(H_2)_{--}=(A_1+A_2)\bra{\psi}A\ket{\psi}\otimes \proj{\psi}=0$.
 
Let $H_1'=\bra{\psi}A^2\ket{\psi}(2A_1A_2 +A_1^2+A_2^2)$ so that
 \[    (H_2)_{-+} H_0^{-1} (H_2)_{+-}= (A_1+A_2)\bra{\psi}A^2\ket{\psi}(A_1+A_2)\otimes \proj{\psi}=(H_1')_{--} \]
 as required, where we have used the fact that $H_0^{-1}A\ket{\psi}=A\ket{\psi}$ (since $A\ket{\psi}$ and $\ket{\psi}$ are orthogonal as shown in (\ref{eq:expAzero})).
 
Finally we calculate the third order term:
 \[(H_2)_{-+} H_0^{-1} (H_2)_{++} H_0^{-1} (H_2)_{+-}=(A_1+A_2)^3 \bra{\psi}A^3\ket{\psi} \otimes \proj{\psi}\]
 so by Lemma~\ref{lem:thirdorder} we can set $H_1=-(A_1^3+A_2^3)\bra{\psi}A^3\ket{\psi}$ and simulate an interaction of the form $A\otimes A^2 +A^2 \otimes A$ which is universal by Lemma~\ref{lem:diagonal} unless $A^2=\lambda A+ \mu I$ for some $\lambda, \mu \in \R$. But if A has three distinct eigenvalues, then it cannot be a root of any polynomial of degree 2. 
\end{proof}

\begin{lem}
\label{lem:AAuniversal}
 Let $H=A\otimes A$ be a 2-qudit interaction such that $A$ is \emph{not} of the form $a\proj{\psi}+bI$ for any $\ket{\psi} \in \C^d$, and $a,b \in \R$. Then $H$ is LA-universal.
\end{lem}

\begin{proof}
By assumption $A$ is not proportional to the identity so has at least two distinct eigenvalues. If $A$ has three distinct eigenvalues then $H$ is LA-universal by Lemma~\ref{lem:3evalues}. It remains to consider the case where $A$ has exactly two eigenvalues $\lambda_1 \neq \lambda_2$.

Since $A\neq  a\proj{\psi}+bI$, there must be at least two orthonormal eigenvectors for each eigenvalue of $A$. Let $\ket{\psi_i}$ and $\ket{\phi_i}$ be orthonormal eigenvectors with eigenvalue $\lambda_i$ for $i \in \{1,2\}$. Let $P$ be the projector onto span$\{\ket{\psi_1},\ket{\psi_2},\frac{\ket{\phi_1}+\ket{\phi_2}}{\sqrt{2}}\}$, and let $H_0=I-P$. Then by Lemma~\ref{lem:firstorder}, we can simulate interactions of the form $B \otimes B$ where 
\[B=PAP=\lambda_1 \proj{\psi_1} +\lambda_2 \proj{\psi_2} + \frac{\lambda_1+\lambda_2}{2}\left(\frac{\ket{\phi_1}+\ket{\phi_2}}{\sqrt{2}}\right)\left(\frac{\bra{\phi_1}+\bra{\phi_2}}{\sqrt{2}}\right)\]
which has three distinct eigenvalues $\lambda_1, \lambda_2, \frac{\lambda_1+\lambda_2}{2}$, so is LA-universal by Lemma~\ref{lem:3evalues}.
\end{proof}

We next show that the one remaining case that is not covered by Lemma \ref{lem:AAuniversal} corresponds to stoquastic Hamiltonians, so is unlikely to be universal.

\begin{lem}
\label{lem:stoqconj}
Let $H=A\otimes A$ be a 2-qudit interaction where $A$ is of the form $A=a \proj{\psi}$ for some $\ket{\psi} \in \C^d$ and 
$a \neq 0$. Then any Hamiltonian of the form $\sum_i M^{(i)} + \sum_{j \neq k} \alpha_{jk} H_{jk}$ -- where $M^{(i)}$ are arbitrary single qudit operators acting only on qudit $i$, $H_{jk}$ refers to the interaction $H$ applied to qudits $j$ and $k$, and $\alpha_{jk} \in \R$ -- is equivalent to a stoquastic Hamiltonian under conjugation by a local unitary operation.
\end{lem}

\begin{proof}
By conjugating $H$ by a local unitary $U\otimes U$ and rescaling, we may assume without loss of generality that $A=\proj{0}$. For each qudit, we demonstrate the existence of a local unitary acting on that qudit which leaves $\ket{0}$ unchanged, but rotates the 1-local term $M^{(i)}$ acting on that qudit into a stoquastic term (i.e.\ non-positive off-diagonal entries). First we conjugate by a unitary $U_1=\proj{0}+\widetilde{U}$ where $\widetilde{U}$ acts only on $S=\text{span}\{\ket{1},\dots \ket{d-1}\}$, such that $U_1M^{(i)}U_1^{\dagger}$ is diagonal on the space $S$; that is,
\[U_1M^{(i)}U_1^{\dagger}=\sum_{j=0}^{d-1}{w_j}\proj{j} +\sum_{j-1}^{d-1}a_j\ket{0}\bra{j}+a_j^*\ket{j}\bra{0}.\] 
Write $a_j=|a_j|e^{i\theta_j}$ and define $U_2=\proj{0}+\sum_{j=1}^{d-1}-e^{i\theta_j}\proj{j}$
so that \[U_2U_1M^{(i)}U_1^{\dagger}U_2^{\dagger}=\sum_{j=0}^{d-1}w_j\proj{j} +\sum_{j=1}^{d-1}- |a_j| \bigl(\ket{0}\bra{j}+\ket{j}\bra{0}\bigr).\]
This operator is clearly stoquastic.
\end{proof}

% ------------------------------------------------------------------------------

\subsection{Interactions not necessarily diagonalisable by local unitaries}

Having dealt with the diagonal case, we now need to consider other types of interactions. The first is interactions of the form $A\otimes A+B\otimes B$.

\begin{lem}
\label{lem:AA+BB}
Let $A$ and $B$ be single-qudit Hermitian operators such that the operators $A'=A-\tr(A)I/d$ and $B'=B-\tr(B)I/d$ are linearly independent, and write $H = A\otimes A+B\otimes B$. Then $H$ is LA-universal.
\end{lem}

\begin{proof}
If $A$ and $B$ commute, then $A$ and $B$ are simultaneously diagonalisable by the same unitary $U$. Conjugating $H$ by $U \otimes U$, the result follows from Lemma~\ref{lem:diagonal}. So suppose $A$ and $B$ do not commute. Then there must exist an eigenstate $\ket{\psi}$ of $A$ with eigenvalue $\lambda$ such that $AB\ket{\psi}\neq BA\ket{\psi}=\lambda B\ket{\psi}$. So $B\ket{\psi}$ is not in the eigenspace of $A$ corresponding to eigenvalue $\lambda$, and there must exist an orthogonal eigenstate $\ket{\phi}$ of $A$ with distinct eigenvalue $\mu \neq \lambda$, such that $\bra{\phi}B\ket{\psi} \neq 0$. By multiplying $\ket{\phi}$ by a phase $e^{i\theta}$, we may assume $\bra{\phi}B\ket{\psi}$ is real. 

We will apply a heavy term $H_0=I-\ket{\psi}\bra{\psi}-\ket{\phi}\bra{\phi}$ with ground space $S=\text{span}\{\ket{\psi},\ket{\phi}\}$ to each of the qudits on which $H$ acts. Then we can use first-order perturbation theory (Lemma~\ref{lem:firstorder}) to produce a logical 2-qubit interaction by projecting $H$ onto $S$. Let $P$ be the projector onto $S$, and identify $\ket{0_L}=\ket{\psi}$ and $\ket{1_L}=\ket{\phi}$ so that
\begin{align*}PAP&=\lambda\proj{\psi}+\mu\proj{\phi} = \frac{\lambda-\mu}{2}Z_L+\frac{\lambda+\mu}{2}I_L,\\
PBP&=aZ_L+\bra{\phi}B\ket{\psi}X_L+\frac{\bra{\psi}B\ket{\psi}+\bra{\phi}B\ket{\phi}}{2}I_L,
\end{align*}
where $a=(\bra{\psi}B\ket{\psi}-\bra{\phi}B\ket{\phi})/2$. So $P^{\otimes 2}HP^{\otimes 2} = \sum M_{ij} \sigma^i \otimes \sigma^j+\text{1-local terms}$, where $M$ is the matrix defined by
\[M=\left(\begin{array}{ccc}
\bra{\phi}B\ket{\psi}^2 & 0 & a\bra{\phi}B\ket{\psi} \\
0 & 0 & 0 \\
a\bra{\phi}B\ket{\psi} & 0 & a^2+(\lambda-\mu)^2/4 \\ 
\end{array}
\right),\]
which has rank 2 whenever $\bra{\phi}B\ket{\psi}(\lambda-\mu)\neq 0$. As shown in~\cite{Cubitt-Montanaro,cubitt17}, any such interaction is universal. Hence $H$ is LA-universal.
\end{proof}

Next we use Lemma~\ref{lem:AA+BB} to deal with almost all other types of interactions.

\begin{lem}
\label{lem:rank2ormore}
Let $H$ be a 2-qudit interaction with 2-local rank $\geqslant 2$. Then $H$ is LA-universal.
\end{lem}

\begin{proof}
Let $H'$ be the 2-local part of $H$, given by $H'=\sum_{a,b}M_{ab}T^a \otimes T^b$ where $\rank(M)\geqslant 2$ and $\{T^a\}_a$ is a basis for the space of of traceless Hermitian matrices. Let $S$ be a two-dimensional subspace of $\C^d$ spanned by orthonormal vectors $\ket{\psi}$ and $\ket{\phi}$ to be chosen later. Let $P$ be the projector onto $S$ and let $H_0=I-P$ act on qudit $j$. Then by Lemma~\ref{lem:firstorder}, for any $i$ we can simulate interactions of the form $F_{ij}=\sum_{a,b} M_{ab} T_i^a \otimes T_j^b|_S$.

Then, using another round of (second-order) perturbation theory, we choose $H_0=\proj{\phi}$ and $H_2=F_{13}+F_{23}$. The second-order term is given by 
\begin{align*}
-(H_2)_{-+} H_0^{-1} (H_2)_{+-}&=-\sum_{a,b,c,d}M_{ab} (T_1^a+T_2^a) \bra{\psi}T_3^b\proj{\phi}T_3^{d}\ket{\psi} M_{cd}(T_1^c+T_2^c)\otimes \proj{\psi}\\
&=-\left[\sum_{a,c}(R_{ac}+R_{ca})T_1^a T_2^c + \text{1-local terms} \right] \otimes \proj{\psi} 
\end{align*}
where $R_{ac}=\sum_{b,d} M_{ab} \bra{\psi}T_3^b\proj{\phi}T_3^{d}\ket{\psi} M_{cd}= \bra{\psi}K^a\proj{\phi}K^c\ket{\psi}$ where $K^a=\sum_b M_{ab}T_3^b$. Note that $R$ is positive semi-definite and rank 1. Since $R+R^T$ is symmetric, if we can choose $\ket{\psi}$ and $\ket{\phi}$ such that $\rank(R+R^T)=2$, then the simulated interaction must be of the form $-(A\otimes A+B \otimes B)$ and so is LA-universal by Lemma~\ref{lem:AA+BB}. 

Suppose for a contradiction that $\rank(R+R^T)\neq 2$ for any choice of  $\ket{\psi}$ and $\ket{\phi}$. 
Since $\rank(R)=1=\rank(R^T)$, this can only happen if $R=R^T$.
That is, for any $a$ and $c$ and any choice of orthogonal normalised states $\ket{\psi}$ and $\ket{\phi}$, 
\begin{equation}
\label{eq:Rsymmetric}
\bra{\psi}K^a\proj{\phi}K^c \ket{\psi}=\bra{\psi}K^c\proj{\phi}K^a \ket{\psi}.
\end{equation}
By the definition of $K^a$ and the fact that $M$ has rank at least 2, there must be a choice of $a$ and $c$ such that $K^a$ and $K^c$ are linearly independent. Fix this choice of $a$ and $c$ for the remainder of the proof. The contradiction we will show is that equation $(\ref{eq:Rsymmetric})$ implies that $K^a$ and $K^c$ are not linearly independent.

Fix $\ket{\psi}$ and extend it to an orthonormal basis $B_{\psi}=\{\ket{\psi},\ket{e_1},\dots,\ket{e_{d-1}}\}$. Then taking $\ket{\phi} = \ket{e_i}$ for any $i$, equation (\ref{eq:Rsymmetric}) holds. Taking the sum over all $i$ we have $\bra{\psi}K^a K^c\ket{\psi}=\bra{\psi}K^c K^a\ket{\psi}$. Since $\ket{\psi}$ was arbitrary, we conclude that $[K^a,K^c]=0$. So $K^a$ and $K^c$ are simultaneously diagonalisable. Let $\ket{\Phi}=\frac{1}{\sqrt{d}}\sum_i \ket{i}$, where $\{\ket{i}\}$ is an eigenbasis for both $K^a$ and $K^c$. We can decompose an arbitrary state $\ket{\psi}$ as $\ket{\psi} = \ket{\psi'}+b\ket{\Phi}$ where $\ket{\psi'}$ is an unnormalised vector orthogonal to $\ket{\Phi}$. Then 
\[\bra{\Phi}K^a\ket{\psi}=\bra{\Phi}K^a\ket{\psi'}+b\bra{\Phi}K^a\ket{\Phi}=\bra{\Phi}K^a\ket{\psi'}+b\frac{1}{d}\tr(K^a)=\bra{\Phi}K^a\ket{\psi'}\]
and similarly for $K^c$. So, setting $\ket{\phi}=\ket{\Phi}$, as $\ket{\psi'}$ is orthogonal to $\ket{\Phi}$ equation (\ref{eq:Rsymmetric}) holds for any choice of $\ket{\psi}$, and hence $K^a\proj{\Phi}K^c=K^c\proj{\Phi}K^a$.
Multiplying on the left by $\bra{i}$ and on the right by $\ket{j}$ this gives $\lambda_i\mu_j=\mu_i\lambda_j$ where $\lambda_i$ and $\mu_i$ are the eigenvalues corresponding to $\ket{i}$ of $K^a$ and $K^c$ respectively. This implies there exists $C\in \R$ such that $\lambda_i=C\mu_i$ for all $i$, and hence that $K^a=C K^c$ which is the contradiction we desired.
\end{proof}

We have now proven all the ingredients we need to show the following theorem, which is the 2-local, single-interaction special case of Theorem \ref{thm:lamainklocal}:

\begin{thm}
\label{thm:lamain}
Let $H$ be a 2-qudit interaction which is not 1-local. 
If, up to addition of 1-local terms, $H = \alpha \proj{\psi}^{\otimes 2}$ for some state $\ket{\psi}$ and some $\alpha \neq 0$, then $H$ is LA-stoquastic-universal.
Otherwise $H$ is LA-universal.
\end{thm}

\begin{proof}
Let $H'$ be the interaction obtained from $H$ by deleting its 1-local part. Then, by Lemma~\ref{lem:rank2ormore} $H$ is LA-universal unless $H' = A \otimes B$ for some $A$ and $B$. If $A$ and $B$ are linearly independent, then $A\otimes B+B \otimes A$ has 2-local rank 2 and so is LA-universal by Lemma~\ref{lem:rank2ormore}. Otherwise, $B = \beta A $ for some $\beta\neq 0$, so $H' = \beta A \otimes A$. Diagonalising $H$ using a local unitary $U \otimes U$ and using Lemma~\ref{lem:diagonal}, $H$ is LA-stoquastic-universal. In addition, if $A \neq a\proj{\psi}+bI$ for some $\ket{\psi} \in \C^d$, then $H$ is LA-universal by Lemma~\ref{lem:AAuniversal}.
\end{proof}

We do not expect any larger class of interactions to be LA-universal than in Theorem~\ref{thm:lamain}, as shown by Lemma~\ref{lem:stoqconj}. 

\subsection{Extension to $k$-local interactions}
In order to extend our results to interaction terms that act on more than 2 qudits, we first show how 1-local terms can be used to extract $(k-1)$-local interactions from $k$-local interactions.
\begin{lem}
\label{lem:ktok-1}
Let $H$ be a $k$-local interaction with a decomposition $H=\sum_{i=1}^l A_i \otimes B_i$ where the $A_i$ operators act on $k-1$ qudits and the $B_i$ operators are linearly independent. Then using $H$ interactions and additional 1-local terms we can simulate any interaction in $\linspan\{A_i\}_{i=1}^l$.
\end{lem}
\begin{proof}
Fix a single qudit state $\ket{\psi} \in \C^d$, and let $H_0=I-\proj{\psi}$. 
Then by Lemma \ref{lem:firstorder}, a first-order perturbation gadget applying $H_0$ to the $k$'th qudit will simulate a $(k-1)$-qudit interaction of the form $\sum_{i=1}^l A_i \bra{\psi} B_i \ket{\psi}$. 
Using different ancilla qubits projected into different states $\ket{\psi}$ we can produce a linear combination of such interactions.
It therefore suffices to prove that $\linspan\{ x^{(\psi)} : \ket{\psi} \in \C^d\}=\R^l$, where $x^{(\psi)}$ is the vector in $\R^l$ with coefficients given by $x^{(\psi)}_i=\bra{\psi} B_i \ket{\psi}$.

Suppose for a contradiction that the $x^{(\psi)}$ do not span the whole of $\R^l$, then there must exist some non-zero $\lambda \in \R^l$ which is orthogonal to $x^{(\psi)}$ for all $\ket{\psi}$, so
\[0=\sum_{i=1}^l \lambda_i x_i^{(\psi)}=\bra{\psi}\left(\sum_i \lambda_i B_i \right)\ket{\psi}  \quad\forall \ket{\psi} \quad \Rightarrow  \quad \sum_i \lambda_i B_i=0\]
contradicting the assumption that the $B_i$ are linearly independent.
\end{proof}

%\begin{lem}
%Let $H$ be a $k$-local interaction (with $k \geqslant 2$) satisfying $\tr_i(H)=0$ for all $i \in \{1,2, \dots ,k\}$. 
%Then $H$ is LA-universal unless there exists a state $\ket{\psi}$ such that $H$ is proportional to $(d\proj{\psi}-I)^{\otimes k}$.
%\end{lem}

%\begin{proof}
%We prove this by induction on $k$, noting that the $k=2$ case follows directly from the statement of Theorem \ref{thm:lamain}.
%Assuming the claim holds for $k-1$ qudits, then by Lemma \ref{lem:ktok-1}, $H$ is LA-universal unless $H= (d\proj{\psi}-I)^{\otimes k-1}\otimes B$ for some single qudit operator $B$. By applying Lemma \ref{lem:ktok-1} to a different qudit, we conclude that $B$ must be proportional to $(d\proj{\psi}-I)$ as required.
%\end{proof}

Let $H$ be a $k$-qudit Hamiltonian and $S$ be a subset of those $k$ qudits. Define $H_S$ to be the part of $H$ which acts non-trivially only on $S$ but does not have any part in its decomposition which acts trivially on any subset of $S$.
More precisely, take a basis $\{I, B_i\}$ of Hermitian matrices on $\C^d$, where the $B_i$ are traceless, and decompose $H$ as a linear combination of tensor products of terms from these bases; then $H_S$ is the sum of all terms which are non-identity on $S$ and identity elsewhere.
Note that $H=\sum_S H_S$ and $\tr_i(H_S)=0$ for any $i \in S$.

The following corollary is an easy consequence of Lemma \ref{lem:ktok-1}.
\begin{cor}
\label{cor:klocal}
Let $H$ be a $k$-qudit interaction, with a decomposition $H=\sum_S H_S$ where $H_S$ is defined as above. Then, using $H$ and additional 1-local terms, it is possible to simulate the interaction $H_S$ for any subset $S$. 
\end{cor}

\begin{proof}
 Let $H$ have a decomposition $H=A_0\otimes I+\sum_i A_i \otimes B_i$ where the $B_i$ are traceless Hermitian matrices acting nontrivially on a single qudit. Then, by Lemma~\ref{lem:ktok-1}, we can simulate $A_0$. This is the part of $H$ which acts trivially on the last qudit and can hence be expressed as $A_0 \otimes I =\sum_{S' \subseteq \{1,\dots k-1\}} H_{S'}$.
By applying Lemma~\ref{lem:ktok-1} repeatedly in this way, we can simulate any interaction of the form $H(S)=\sum_{S'\subseteq S}H_{S'}$ for an arbitrary set $S$.

We now prove the corollary by induction on $|S|$, noting that the base case $|S|=1$ is trivial since we have access to all 1-local terms. 
Assume the claim for all subsets of size $l$ and let $S$ be a subset of size $l+1$.
 By the induction hypothesis, we can simulate $H_{S'}$ for all subsets $S'\subset S$. Taking these away from $H(S)$ we are left with $H_S$ as desired.
\end{proof}

We are now ready to generalise Theorem \ref{thm:lamain} to $k$-local interactions.

\begin{repthm}{thm:lamainklocal}
Let $\cS$ be a set of interactions, which are not all 1-local, between qudits of dimension $d$. Then $\cS$ is:
\begin{itemize}
\item stoquastic and LA-stoquastic-universal, if there exists $\ket{\psi}\in \C^d$ such that all interactions in $\cS$ are, up to the addition of 1-local terms, given by a linear combination of operators taken from the set $\{I, \proj{\psi},\proj{\psi}^{\otimes 2},\proj{\psi}^{\otimes 3},\dots \}$;
\item LA-universal, otherwise.
\end{itemize}
\end{repthm}

\begin{proof}
First note that by the same argument as Lemma~\ref{lem:stoqconj}, the Hamiltonians given in the first case are stoquastic. Since not all interactions are 1-local, Lemma~\ref{lem:ktok-1} can be used to extract a 2-local interaction with non-zero 2-local part, which is LA-stoquastic-universal by Theorem~\ref{thm:lamain}.

It remains to prove that any other set of interactions is universal. Define $T_l$ to be the space of $l$-local interactions that have no $m$-local part in their decomposition for $m<l$, and which can be generated by repeated applications of Lemma~\ref{lem:ktok-1} to interactions $H \in \mathcal{S}$ (and taking linear combinations of such interactions). 
Given an interaction $H$ in $\mathcal{S}$, and a decomposition $H=\sum_{S} H_S$, $T_l$ includes all interactions $H_S$ such that $|S|=l$ by Corollary~\ref{cor:klocal}.
It will therefore suffice to prove that there exists $\ket{\psi}$ such that $T_l= \linspan\{(d \proj{\psi}-I)^{\otimes l}\}$ for all $l$, as then $H = \sum_S H_S$ will be of the desired form.

We prove this claim by induction on $l$. Note that $T_2$ is non-empty unless all interactions in $\mathcal{S}$ are 1-local. By Theorem~\ref{thm:lamain}, each interaction in $T_2$ must be proportional to $(d \proj{\psi}-I)^{\otimes 2}$ for some state $\ket{\psi}$. Moreover, the state $\ket{\psi}$ must be the same for all interactions in $T_2$, or we could simulate $(d \proj{\psi}-I)^{\otimes 2}+(d \proj{\psi'}-I)^{\otimes 2}$ for some $\ket{\psi} \neq \ket{\psi'}$, which is LA-universal by Lemma~\ref{lem:AA+BB}.

Assume now that the claim holds for $T_l$ and consider an interaction $F$ in $T_{l+1}$.
Write $F=\sum_i A_i \otimes B_i$, where $B_i$ are traceless single-qudit operators. Then, by Lemma~\ref{lem:ktok-1}, $\linspan\{A_i\} \subseteq T_l$. 
Therefore, by the induction hypothesis, $F=(d \proj{\psi}-I)^{\otimes l}\otimes B$ for some single-qudit operator $B$. 
By applying Lemma~\ref{lem:ktok-1} to a different qudit, we conclude that $B$ must also be proportional to $(d\proj{\psi}-I)$ as required.
\end{proof}

% ------------------------------------------------------------------------------

\section{$SU(d)$ Heisenberg interaction}

In the remainder of the paper we prove universality for some families of interactions where we are not assisted by free 1-local terms. We consider interactions that generalise the familiar Heisenberg interaction $h = X \otimes X + Y \otimes Y + Z \otimes Z$ for qubits. The Pauli matrices $X$, $Y$, $Z$ correspond to generators for the fundamental (2-dimensional) representation of the Lie algebra $\mathfrak{su}(2)$. So two natural ways to generalise the interaction $h$ are to consider $\mathfrak{su}(d)$ for $d > 2$, or to consider higher-dimensional representations of $\mathfrak{su}(2)$. We study both of these generalisations, beginning with the former.

We first review the mathematical aspects of these generalised Heisenberg models that will be important for us, and in particular the required concepts from representation theory.
Throughout this section,~\cite{FuchsSchweigert} will be a useful reference. 
The fundamental representation of the Lie algebra $\mathfrak{su}(d)$ is given by the space of traceless antiHermitian $d \times d$ matrices. We will follow the physics convention of considering a set of traceless Hermitian operators $\{T^a\}$ such that the real linear span of $\{iT^a\}$ gives the fundamental representation of $\mathfrak{su}(d)$. The basis can be chosen such that $\tr(T^a T^b)=\frac{1}{2}\delta_{ab}$ so that the structure constants $f_{abc}$, defined by $[T^a,T^b]=\sum_{c}if_{abc}T^c$, are completely antisymmetric. For example the Pauli spin matrices $iX/2,iY/2, iZ/2$ are such a basis of $\mathfrak{su}(2)$. The $SU(d)$ Heisenberg interaction $h$ is given by 
\be \label{eq:heisenbergsud} h:=\sum_{a} T^a \otimes T^a. \ee
which (up to rescaling and adding an identity term) is the only two-qudit operator which is invariant under conjugation by the unitary $U\otimes U$ for any matrix $U$ in $SU(d)$. 

\subsection{Notes on the representation theory of $\mathfrak{su}(N)$}
A representation of a Lie algebra $\mathfrak{g}$ is a vector space $\Lambda$ and a linear map $R:\mathfrak{g}\rightarrow L(\Lambda)$ from $\mathfrak{g}$ to the space of linear maps on $\Lambda$, such that $[R(x),R(y)]=R([x,y])$ for all $x,y \in \mathfrak{g}$.
The Lie algebra $\mathfrak{su}(N)$ is semi-simple, which means that any representation $R$ has a direct sum decomposition such that:
\begin{equation}
\label{eq:semisimpledecomp}
R=\bigoplus_{i} R_i \quad \text{ and } \quad \Lambda=\bigoplus_{i} \Lambda_i
\end{equation}
where each $R_i:\mathfrak{g}\rightarrow \Lambda_i$ is an irreducible representation.

The irreducible representations of $\mathfrak{su}(N)$ can be labeled with a Young diagram of at most $N$ rows.
The fundamental representation has a Young diagram of a single box. 
The antifundamental representation or conjugate representation has Young diagram of a single column of $N-1$ boxes, and is given by $R_{\operatorname{conj}}(T^a)=-(T^a)^*$ where $*$ denotes complex conjugation.
The trivial representation is a one dimensional representation in which $R_{\operatorname{trivial}}(T^a)=0$, with Young diagram consisting of a single column of $N$ boxes.
The adjoint representation is an $N^2-1$ dimensional representation in which $R_{\operatorname{adjoint}}$ acts on the Lie algebra itself with the action of the Lie bracket, $R_{\operatorname{adjoint}}(T^a) T^b=[T^a,T^b]$. The adjoint representation has a Young diagram of one column of $N-1$ boxes and a second column of a single box.

For a given representation $R$ of $\mathfrak{su}(N)$, the quadratic Casimir operator $C_R$ is defined by $C_R=\sum_{a}R(T^a)R(T^a)$. Note that $C_R$ commutes with all elements $R(T^b)$:
\begin{align*}[C_R,R(T^b)]&=\sum_a [R(T^a)R(T^a),R(T^b)]=\sum_a \left(R(T^a)[R(T^a),R(T^b)]+[R(T^a),R(T^b)]R(T^a)\right)\\
&=\sum_{a,c}if_{abc}\left(R(T^a)R(T^c)+R(T^c)R(T^a)\right)=0
\end{align*}
since $f_{abc}$ is antisymmetric in $a,c$ and $R(T^a)R(T^c)+R(T^c)R(T^a)$ is clearly symmetric in $a,c$.

When $R$ is an irreducible representation, Schur's Lemma implies that $C_R=c_R I$ for some $c_R \in \R$ known as the Casimir eigenvalue.
For an irreducible representation $R$ of $\mathfrak{su}(N)$ with corresponding Young diagram of $n_{row}$ rows of length $b_1,b_2,\dots,b_{n_{row}}$ and $n_{col}$ columns of length $a_1,a_2,\dots a_{n_{col}}$ and $l$ boxes in total, the Casimir eigenvalue $c_R$ is given by \cite{FuchsSchweigert}
\begin{equation}c_{R}=\frac{1}{2}\left[l(N-l/N)+\sum_{i=1}^{n_{row}}b_i^2-\sum_{i=1}^{n_{col}}a_i^2\right].
\label{eq:casimirvalue}
\end{equation}
For a representation $R$ with a decomposition as in (\ref{eq:semisimpledecomp}), $C_R=\bigoplus_{i} C_{R_i}$ and so each eigenspaces of $C_R$ corresponds to a space $\Lambda_i$ with corresponding Casimir eigenvalue $c_{R_i}$.

Given two representations $R_1$  and $R_2$, we can define a new representation $R_1 \otimes R_2$ called the tensor product representation on the space $\Lambda_1 \otimes \Lambda_2$ by
\[(R_1 \otimes R_2)(T^a)=R_1(T^a) \otimes I_2 + I_1\otimes R_2(T^a)\]
Even when $R_1$ and $R_2$ are irreducible representations, the tensor product representation is not in general irreducible.
The irreducible representations $R_i$ in the decomposition (\ref{eq:semisimpledecomp}) of $R_1 \otimes R_2$ can be calculated using the Young diagrams of $R_1$ and $R_2$.
This process is described in detail in, for example, \cite{FuchsSchweigert}. 
If $R_1$ and $R_2$ have Young diagrams of $l_1$ and $l_2$ boxes respectively, then every irreducible representation in the decomposition of $R_1 \otimes R_2$ has a Young diagram of $l_1+ l_2$ boxes.

\subsection{Alternative $SU(d)$ invariant interaction}
\label{sec:alternativeSUd}
We briefly note that an alternative generalisation of the Heisenberg model has also been studied in the condensed-matter theory literature \cite{Beach09,Lou09,Read89}. The qudits of the system are partitioned into two subsets $A$ and $B$, and the interaction graph is bipartite, with no interactions acting within $A$ or $B$. The total Hamiltonian $H$ is of the form
\[H=\sum_{i \in A,\\ j \in B }\widetilde{h}_{ij} \quad \text{ where }\widetilde{h}=\sum_{a} T^a \otimes (-T^a)^{*}\]
where $^{*}$ denotes complex conjugation.
Since $\sum_a T^a T^a=\frac{d^2-1}{2d}I$ by equation (\ref{eq:casimirvalue}), we have 
\begin{align*}
\widetilde{h}+\frac{d^2-1}{2d} I &=\sum_a T^a \otimes (-T^a)^* + \frac{1}{2}\left(T^aT^a \otimes I +I \otimes (-T^a)^*(-T^a)^*\right)\\
&=\frac{1}{2}\sum_a \widetilde{T}^a \widetilde{T}^a
\end{align*}
where $\widetilde{T}^a= T^a\otimes I + I\otimes (-T^a)^*$. Thus $\widetilde{h}$ is, up to a multiple of the identity, the Casimir operator in the $\widetilde{T}^a$ representation and so commutes with $\widetilde{T}^a$ for all $a$. 
This implies that the total Hamiltonian $H$ is now no longer invariant under conjugation by the unitary $U^{\otimes n}$, but is invariant when conjugated by $U^{\otimes |A|} \otimes (U^{*})^{\otimes |B|}$.

Note that $\widetilde{T}^a$ is the tensor product of the fundamental and antifundamental representation which decomposes into a direct sum of the trivial representation and the adjoint representations (this can be seen using the Young diagram method, as described for example in \cite{FuchsSchweigert}).
Therefore, as $\widetilde{T}^a$ annihilates the state $\ket{\phi}=\frac{1}{\sqrt{d}}\sum_{i}\ket{i}\ket{i}$, $\widetilde{h}+\frac{d^2-1}{2d}I=\frac{1}{2}\sum_a \widetilde{T}^a\widetilde{T}^a$ also annihilates $\ket{\phi}$, and has eigenvalue $\frac{1}{2}c_{\operatorname{adjoint}}=d/2$ on the rest of the space.
Therefore $\widetilde{h}$ is just a linear combination of the identity $I$ and the projector onto $\ket{\phi}$:
\begin{equation}
\label{eq:alternativeSUdproj}
\widetilde{h} =\frac{1}{2d}I-\frac{d}{2} \proj{\phi}
\end{equation}

We will show that this Hamiltonian can simulate an arbitrarily weighted $SU(d)$ invariant interaction $h=\sum_{a} T^a \otimes T^a$ on the $A$ qudits using a second-order mediator gadget.
Consider a system of four qudits with qudits $1,2,3 \in A$ and qudit $4 \in B$. Let $\Pi$ be the projector onto the state $\ket{\phi}_{34}$ and let $H_0= I - \Pi = \frac{2}{d}(\widetilde{h}_{34}+\frac{d^2-1}{2d}I)$, $H_1=0$ and $H_2=\widetilde{h}_{14}+\mu \widetilde{h}_{24}$ for some $\mu \in \R$. Since $\Pi M_4 \Pi = (\tr M) \Pi$ for any $M$ and the $T^a$'s are traceless, $\Pi H_2 \Pi = 0$. By Lemma~\ref{lem:secondorder} we can simulate
\begin{align*}
-\Pi H_2 H_0^{-1} H_2 \Pi & =-\Pi H_2 (I-\Pi) H_2 \Pi=-\sum_{a,b}(T^a_1+\mu T^a_2)\frac{1}{d}\tr(T^aT^b) (T^b_1+\mu T^b_2) \Pi\\
&=-(1+\mu^2)\frac{d^2-1}{4d^2} I - \frac{\mu}{d} \sum_a T_1^a T_2^a,
\end{align*}
where we use that $\sum_a (T^a)^2 = \frac{d^2-1}{2d} I$ in the third equality. By adjusting $\mu$ we can obtain an arbitrarily weighted $h$ interaction up to the identity term.

In order to show that $\widetilde{h}$ is universal, it will therefore suffice to consider only $h$.
We will do this for the rest of the paper.

\subsection{Encoding a logical qubit in a $2d$-qudit gadget}

We now consider a system of $k$ qudits each of dimension $d$, and will use subscript notation to denote which qudit an operator acts on, so $T^a_i$ denotes the action of $T^a$ on qudit $i$ and the identity elsewhere. 
For a set $S$ we use the shorthand $T^a_S=\sum_{i \in S} T^a_i$.
The operators $\{T_S^a\}_a$ form a representation of $\mathfrak{su}(d)$; it is the representation given by the tensor product of the fundamental representation $l=|S|$ times. 

Consider the following Hamiltonian, given by the quadratic Casimir operator in the $\{T^a_S\}_a$ representation:
\begin{align}
C(S)&:=\sum_a T_S^a T_S^a=\sum_{a}\left(\sum_{i \neq j} T^a_i T^a_j + \sum_i T^a_i T^a_i\right)\\
&=\sum_{i\neq j} h_{ij} + \frac{l(d^2-1)}{2d} I.
\end{align}
As discussed above, to understand the eigenspaces of $C(S)$, it suffices to know the irreducible representations contained in the decomposition of $\{T^a_S\}_a$. 
In particular we note that $C(S)$ is a sum of squares of Hermitian matrices so is positive semidefinite, and the Young diagram consisting of a single column of $d$ boxes is a one dimensional irrep, with Casimir eigenvalue zero, corresponding to the state $\ket{\Psi}$, the completely antisymmetric state on $d$ qudits. 
%When $l=kd$ for an integer $k$, the zero eigenvalue is achieved on 1-dimensional irreps corresponding to diagrams of $k$ columns of length $d$ - these correspond to the states $\ket{\Psi}_{S_1} \otimes \dots \otimes \ket{\Psi}_{S_k}$, where $S_1 \cup \dots \cup S_k=S$ is a partitioning of the set $S$ into $k$ subsets of size $d$.
The 1-dimensional irrep is known as the trivial representation because $T^a_S \ket{\Psi}=0$ for all $a$.

We will use a gadget construction to encode a logical qubit within $2d$ physical qudits, using a second-order perturbative gadget (via Lemma \ref{lem:secondorder}) to implement effective interactions across pairs of logical qubits. We consider a system of $2d$ qudits, each of dimension $d$, and each with a label in  $E=\{1,2,\dots, 2d\}$. Let $A=\{3,4,\dots, d+1\}$ and $B=\{d+2,\dots 2d\}$ and consider the Hamiltonian
\[H_0=C(E)+C(A)+C(B) -\frac{(d^2-1)}{d}I. \]
The $-\frac{(d^2-1)}{d}I$ term will simply ensure that the ground state energy of $H_0$ is zero, so that the requirements of Lemma~\ref{lem:secondorder} are met.

First we will show that the ground space of $H_0$ -- which will form our logical qubit -- is indeed two-dimensional. 
In fact the two states in the ground space of $H_0$ sit in the respective ground spaces of $C(E)$, $C(A)$ and $C(B)$. 
The ground space of $C(A)$ is given by the $d$-dimensional space $\Antisymspace$ of antisymmetric states on the $d-1$ qudits in $A$, corresponding to the Young diagram of a single column of $d-1$ boxes. 
Let $\{\ket{i}\}_{i=1}^{d}$ be an orthonormal basis for $\C^d$, then there is a unique (up to a phase) antisymmetric state $\ket{\psi_i}$ in $\linspan \{ \ket{1},\dots,\ket{i-1},\ket{i+1},\dots \ket{d}\}^{\otimes d-1}$. 
These states are clearly orthonormal and form a basis for $\Antisymspace$. 

Then the groundspace of $H_0$ contains
\[\ket{\phi_1}=\ket{\Psi}_{1A}\ket{\Psi}_{2B} \quad \text{ and } \quad \ket{\phi_2}=\ket{\Psi}_{1B}\ket{\Psi}_{2A},\]
where $\ket{\Psi}$ is the completely antisymmetric state on $d$ qudits,
\begin{align}
\ket{\Psi}&=\frac{1}{\sqrt{d!}} \sum_{\sigma \in S_d} \text{sgn}(\sigma)\ket{\sigma(1)}\ket{\sigma(2)}\dots \ket{\sigma(d)}\\
&=\frac{1}{\sqrt{d}}\sum_i \ket{i}\ket{\psi_i}
\end{align}
and $\{\ket{i}\}_i$ and $\{\ket{\psi_i}\}_i$ are the orthonormal bases for $\C^d$ and $\Antisymspace$ as defined above.
Clearly, these states are in the ground space of $C(A)$ and $C(B)$, and $\ket{\Psi}$ is the antisymmetric state on $d$ qudits so $T_E^a$ annihilates $\ket{\phi_1}$ and $\ket{\phi_2}$, implying that these states are also in the ground space of $C(E)$.
To see that these are the only two states in the ground space of $H_0$, we note that the ground space of $C(A)+C(B)$ is spanned by states in the representations given in Figure~\ref{fig:tensorrep}.
The $C(E)$ term forces the ground space of $H_0$ to be the two dimensional space corresponding to the two copies of the Young diagram of two columns of $d$ boxes.

It is important to note that $\ket{\phi_1}$ and $\ket{\phi_2}$ are not orthogonal:
\begin{align}
\braket{\phi_1|\phi_2}&=\frac{1}{d^2}\sum_{i,j,k,l}\left(\bra{i}\bra{\psi_i}\bra{j}\bra{\psi_j}\right)\left(\ket{k}\ket{\psi_l}\ket{l}\ket{\psi_k}\right)\\
&=\frac{1}{d^2}\sum_{i,j,k,l}\delta_{ik}\delta_{il}\delta_{jl}\delta_{jk}=\frac{1}{d^2}\sum_{i}\delta_{ii}=\frac{1}{d}.
\end{align}

\newcommand{\extendedyoung}[4]{

\draw (0,0) grid (#1,-1);
\draw (0,-1) grid (#2,-2);
\ifthenelse{#2>1}{\pgfmathsetmacro{\xp}{2}}{\pgfmathsetmacro{\xp}{#2}}
\draw[dotted] (0,-2) grid (\xp,-3);
\draw (0,-3) grid (\xp,-4);
\draw (0,-4) grid (#3,-5);
\ifthenelse{#4>0}{
\draw (0,-5) grid (#4,-6);}{}}

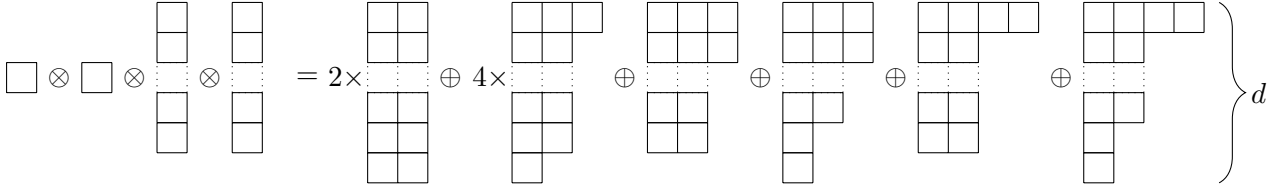
\begin{figure}
\begin{tikzpicture}[scale=0.4]
\pgfmathsetmacro{\s}{0.75}
\begin{scope}[shift={(-10,0)}]
\draw (0,-2) grid (1,-3);
\pgfmathsetmacro{\x}{1+\s}
\node at (\x,-2.5) {$\otimes$};
\pgfmathsetmacro{\x}{\x+\s}  \begin{scope}[shift={(\x,0)}] \draw (0,-2) grid (1,-3); \end{scope}
\pgfmathsetmacro{\x}{\x+1+\s} \node at (\x,-2.5) {$\otimes$};
\pgfmathsetmacro{\x}{\x+\s} \begin{scope}[shift={(\x,0)}] \extendedyoung{1}{1}{1}{0} \end{scope}
\pgfmathsetmacro{\x}{\x+1+\s}\node at (\x,-2.5) {$\otimes$};
\begin{scope}[shift={(3+6*\s,0)}] \extendedyoung{1}{1}{1}{0} \end{scope}
%\draw [decorate,decoration={brace,amplitude=10pt}] (\x+2,0) -- (\x+2,-5)node [black,midway,xshift=25] {$d-1$ };
\end{scope}

\node at (0,-2.5) {$=$};

\begin{scope}[shift={(2,0)}]
\node at (-0.7,-2.5) {$2 \times$};
\extendedyoung{2}{2}{2}{2}
\pgfmathsetmacro{\x}{2+\s}  \node at (\x,-2.5) {$\oplus$};
\pgfmathsetmacro{\x}{\x+1.3+\s} \begin{scope}[shift={(\x,0)}] \node at (-0.7,-2.5) {$4 \times$}; \extendedyoung{3}{2}{2}{1} \end{scope}
\pgfmathsetmacro{\x}{\x+3+\s} \node at (\x,-2.5) {$\oplus$};
\pgfmathsetmacro{\x}{\x+\s} \begin{scope}[shift={(\x,0)}] \extendedyoung{3}{3}{2}{0} \end{scope}
\pgfmathsetmacro{\x}{\x+3+\s} \node at (\x,-2.5) {$\oplus$};
\pgfmathsetmacro{\x}{\x+\s} \begin{scope}[shift={(\x,0)}] \extendedyoung{3}{3}{1}{1} \end{scope}
\pgfmathsetmacro{\x}{\x+3+\s} \node at (\x,-2.5) {$\oplus$};
\pgfmathsetmacro{\x}{\x+\s} \begin{scope}[shift={(\x,0)}] \extendedyoung{4}{2}{2}{0} \end{scope}
\pgfmathsetmacro{\x}{\x+4+\s} \node at (\x,-2.5) {$\oplus$};
\pgfmathsetmacro{\x}{\x+\s} \begin{scope}[shift={(\x,0)}] \extendedyoung{4}{2}{1}{1} \end{scope}

\draw [decorate,decoration={brace,amplitude=10pt}] (\x+4.5,0) -- (\x+4.5,-6)node [black,midway,xshift=15] {$d$ };
\end{scope}
\end{tikzpicture}
\caption{Irreducible representations in the decomposition of the ground space of $C(A)+C(B)$. The rules for taking the tensor product of representations given as Young diagrams can be found for example in \cite{FuchsSchweigert}.}
\label{fig:tensorrep}
\end{figure}

In order to calculate perturbative gadgets we want to understand the action of the physical interaction $h$ defined in (\ref{eq:heisenbergsud}) in this logical qubit space. First we calculate $M_{ij}(T^a_k T^b_l ):=\bra{\phi_i}T^a_k T^b_l \ket{\phi_j}$ for all $a,b,i,j$ and any $k,l \in \{1,2,A,B\}$, and then we will convert to an orthogonal basis later. We only show the calculations for three of these values, as all others can be calculated by symmetric arguments, and recalling  that ($T^a_1+T^a_A) \ket{\Psi}_{1A}=0$.
For example, we can calculate $\bra{\phi_1}T^a_1 T^b_2\ket{\phi_2}= -\bra{\phi_1}T^a_1 T^b_A\ket{\phi_2}= \bra{\phi_1} T^b_1 T^a_1\ket{\phi_2}$.

\begin{align}
\bra{\phi_1}T^a_1 T^b_1\ket{\phi_1} &=\frac{1}{d^2}\sum_{i,j,k,l}\left(\bra{i}\bra{\psi_i}\bra{j}\bra{\psi_j}\right) T_1^a T_1^b\left(\ket{k}\ket{\psi_k}\ket{l}\ket{\psi_l}\right) \\
 &=\frac{1}{d^2}\sum_{i,j,k,l} \bra{i} T^a T^b \ket{k} \delta_{ik}\delta_{jl
}\delta_{jl}\\
 &=\frac{1}{d} \tr(T^aT^b)
 \end{align}
 
 \begin{align}
\bra{\phi_1}T^a_1 T^b_1\ket{\phi_2} &=\frac{1}{d^2}\sum_{i,j,k,l}\left(\bra{i}\bra{\psi_i}\bra{j}\bra{\psi_j}\right) T_1^a T_1^b\left(\ket{k}\ket{\psi_l}\ket{l}\ket{\psi_k}\right) \\
 &=\frac{1}{d^2}\sum_{i,j,k,l} \bra{i} T^a T^b \ket{k} \delta_{il}\delta_{jl}\delta_{jk}\\
 &=\frac{1}{d^2} \tr(T^aT^b)
 \end{align}
 
 \begin{align}
 \bra{\phi_1}T^a_1 T^b_2\ket{\phi_1} &=\frac{1}{d^2}\sum_{i,j,k,l}\left(\bra{i}\bra{\psi_i}\bra{j}\bra{\psi_j}\right) T_1^a T_2^b\left(\ket{k}\ket{\psi_k}\ket{l}\ket{\psi_l}\right) \\
 &=\frac{1}{d^2}\sum_{i,j,k,l} \bra{i} T^a\ket{k} \delta_{ik} \bra{j} T^a\ket{l}\delta_{jl}\\
 &=\frac{1}{d^2} \tr(T^a)\tr(T^b)=0.
 \end{align}

We then have 
\begin{align}
M(T^a_1T^b_1)&=\frac{\tr(T^aT^b)}{d^2}\begin{pmatrix} d & 1\\ 1 & d \\ \end{pmatrix}
& M(T^a_1T^b_A)&=\frac{\tr(T^aT^b)}{d^2}\begin{pmatrix} -d & -1\\ -1 & 0 \\ \end{pmatrix} \\
M(T^a_1T^b_B)&=\frac{\tr(T^aT^b)}{d^2}\begin{pmatrix} 0 & -1\\ -1 & -d \\ \end{pmatrix}
& M(T^a_1T^b_2)&=\frac{\tr(T^aT^b)}{d^2}\begin{pmatrix} 0 & 1\\ 1 & 0 \\ \end{pmatrix} 
\end{align}
Choose $\ket{0_L}$ and $\ket{1_L}$ to be the orthonormal basis of $\linspan\{\ket{\phi_1},\ket{\phi_2}\}$ given by:
%\begin{align}
%\ket{0_L}&= \ket{\phi_1} \\
%\ket{1_L}&= \frac{1}{\sqrt{d^2-1}}\left(\ket{\phi_2}-d \ket{\phi_1}\right)
%\end{align}
\begin{align}
\ket{0_L}&= \sqrt{\frac{d}{2(d+1)}}\left(\ket{\phi_1}+\ket{\phi_2}\right) \\
\ket{1_L}&= \sqrt{\frac{d}{2(d-1)}}\left(\ket{\phi_1}-\ket{\phi_2} \right)
\end{align}

Let $\Pi= \proj{0_L}+\proj{1_L}$ be the projector onto the two dimensional ground space. Then the action of $\Pi T^a_i T^a_j \Pi$ is given in Table~\ref{tab:su(n)terms}. Therefore by Lemma~\ref{lem:firstorder}, choosing $H_1=\alpha h_{1A}+\beta h_{12}$ for $\alpha,\beta\in \R $, we can simulate any logical 1-local interaction in $\linspan\{X_L,Z_L\}$.

\begin{table}
\[\begin{array}{|c|c|}
\hline 
(i,j) & \Pi T^a_i T^a_j \Pi \\
\hline
(1,1), (2,2), (A,A), (B,B) & \frac{1}{2d} I_L \\
(1,A), (2,B) & -\frac{1}{4\sqrt{d^2-1}}X_L -\frac{1}{4(d^2-1)}Z_L -\frac{d^2-2}{4d(d^2-1)} I_L \\
(1,B), (2,A) & \frac{1}{4\sqrt{d^2-1}}X_L -\frac{1}{4(d^2-1)}Z_L -\frac{d^2-2}{4d(d^2-1)} I_L \\
(1,2), (A,B) & \frac{1}{2(d^2-1)}Z_L-\frac{1}{2d(d^2-1)} I_L \\
\hline
\end{array}\]
\caption{Action of $T^a_iT^a_j$ in the ground space of $H_0$. }
\label{tab:su(n)terms}
\end{table}

\subsection{Second-order terms}

We now want to simulate interactions between two logical qubits using a second-order gadget, via Lemma~\ref{lem:secondorder}. $H_1$ is chosen as in the previous section to simulate any 1-local terms desired. Consider two copies of the gadget above with qudit labels $\{1,2,\dots, 2d\}$ and $\{1',2',\dots, 2d'\}$ respectively.
We will choose $H_2=\sum_{i,j} \alpha_{i j} h_{ij'}$, so we need to calculate 
\[ \Pi H_2 (H_0)^{-1} H_2\Pi = \sum_{i,j,k,l} \alpha_{ij} \alpha_{kl}\Pi h_{ij'} (H_0)^{-1} h_{kl'} \Pi.\]
The difficult part of this calculation is to understand how the $(H_0)^{-1}$ term acts. For any state $\ket{\psi}$ in the ground space of $H_0$, it is easy to check that the states $\{T_1^b\ket{\psi}\}_b$ are orthogonal  and that $T_E^a$ acts on this space as the adjoint representation:
 \[T^a_E T^b_1 \ket{\psi}=\left(T^a_1 T^b_1 +\sum_{i\neq 1}T^b_1T^a_i\right)\ket{\psi}
=\left(T^a_1 T^b_1 -T^b_1 T^a_1 \right)\ket{\psi}= [T^a_1,T^b_1]\ket{\psi}.\]
Therefore $T_1^b \ket{\psi}$ is an eigenvector of $C(E)$ with the Casimir eigenvalue corresponding to the adjoint representation, which has Young diagram consisting of one column of length $d-1$ and a second column of length $1$. By equation (\ref{eq:casimirvalue}), this eigenvalue is given by $c_{\text{adjoint}}=d$, which we can also check directly: 
\begin{align}
C(E)T^b_1 \ket{\psi}&=\sum_a T^a_E T^a_E T^b_1\ket{\psi}=\sum_a\left[T_1^a,[T_1^a,T_1^b]\right]\ket{\psi}\\
&=-\sum_{a,c,e}f_{abc}f_{ace} T^e_1\ket{\psi}=-\sum_e \kappa_{be} T^e_1\ket{\psi}=dT^b_1 \ket{\psi}
\end{align}
where we have used the antisymmetry of the structure constants $f_{abc}$ and the definition of the Killing form $\kappa_{ab}=\sum_{c,e} f_{ace}f_{bec}=-2d\tr(T^aT^b)$. 

Furthermore, the operator $T^b_1$ does not act on $A$ or $B$ so the state $T^b_1\ket{\psi}$ is still antisymmetric with respect to permutations within $A$ and $B$ and so is in the zero-energy ground space of $C(A)+C(B)-\frac{d^2-1}{d}I$. Thus $H_0 h_{kl'} \Pi= 2d h_{kl'} \Pi$ and so 
\begin{align*} \Pi h_{ij'} (H_0)^{-1} h_{kl'} \Pi &= \frac{1}{2d}\Pi h_{ij'} h_{kl'} \Pi=\frac{1}{2d}\sum_{a,b} \Pi_E T^a_i T^b_k \Pi_E \otimes \Pi_{E'} T^a_{j'} T^b_{l'} \Pi_{E'}\\
&=\frac{1}{2d}\sum_{a} \Pi_E T^a_i T^a_k \Pi_E \otimes \Pi_{E'} T^a_{j'} T^a_{l'} \Pi_{E'},
\end{align*}
which corresponds to a logical operator that can be read off from Table~\ref{tab:su(n)terms}. We choose 
%$\alpha_{ij}=1$ if $(i,j) \in \{ (1,1), (2,2), (A,A)\}$ and $\alpha_{ij}=0$ otherwise.
$\alpha_{ij}=1$ if $(i,j) \in \{ (1,A), (2,B), (A,1), (B,A), (B,B) \}$ and $\alpha_{ij}=0$ otherwise.
Then by Lemma~\ref{lem:secondorder} we find that the simulated interaction is 
%\[ -\Pi H_2 (H_0)^{-1} H_2\Pi=-\frac{1}{2d}\left(\frac{1}{8-8d^2}X_LX_L - \frac{3}{8(d^2-1)^2}Z_LZ_L+ \text{1-local terms}\right), \]
\[ -\Pi H_2 (H_0)^{-1} H_2\Pi = \frac{1}{8d(d^2-1)}\left(X_LX_L + \frac{3}{d^2-1}Z_LZ_L+ \text{1-local terms}\right), \]
which can be checked either by hand or using a computer algebra package.
As we can produce arbitrary 1-local terms using the arguments of the previous section, this allows us to simulate the interaction $\alpha(XX + \frac{3}{d^2-1} ZZ)$ for an arbitrary positive weight $\alpha$, which falls into a family that was shown to be universal\footnote{Note that the results of~\cite{Piddock-Montanaro} are stated in terms of QMA-completeness, but it is easy to check that, in combination with~\cite{cubitt17}, they imply that universality holds for this interaction.} in \cite{Piddock-Montanaro}.
This completes the proof of the following theorem:

\begin{repthm}{thm:sud}
For any $d \ge 2$, the $SU(d)$ Heisenberg interaction $h:=\sum_{a} T^a \otimes T^a$, where $\{T^a\}$ are traceless Hermitian matrices such that $\tr(T^a T^b)=\frac{1}{2}\delta_{ab}$, is universal.
\end{repthm}

The following corollary is an immediate consequence of Theorem~\ref{thm:sud} and the discussion in Section \ref{sec:alternativeSUd}.
\begin{cor}
\label{cor:alternativeSUd} 
For any $d \ge 2$, the alternative $SU(d)$ Heisenberg interaction $\widetilde{h}:=- \sum_a T^a \otimes (T^a)^{\star}$, where $\{T^a\}$ are traceless Hermitian matrices such that $\tr(T^a T^b)=\frac{1}{2}\delta_{ab}$, is universal even on a bipartite interaction graph.
\end{cor}

\section{Rank 1 projectors}
\label{sec:proj}
In this section we consider the family of $\cS$-Hamiltonians where $\cS$ contains a single rank 1 projector $P$ onto a two qudit state $\ket{\psi} \in (\C^d)^{\otimes 2}$.
We prove universality even in the restricted setting where interactions are only allowed between qudits on a bipartite interaction graph.
We note that this also trivially implies universality without such a restriction.
\begin{repthm}{thm:proj}
Let $P= \proj{\psi}$  be the projector onto the two-qudit state $\ket{\psi} \in (\C^d)^{\otimes 2}$. Then Hamiltonians of the form \[H= \sum_{i \in A, j \in B} \alpha_{ij}P_{ij}\] where $A$ and $B$ are disjoint subsets of qubits and $\alpha_{ij} \in \R$, are universal if $\ket{\psi}$ is entangled.
\\Otherwise, if $\ket{\psi}$ is a product state, then this family of Hamiltonians is classical.
\end{repthm}

\begin{proof}
We first conjugate the entire Hamiltonian by a total unitary $\left(\bigotimes_{i \in A} U\right)\otimes \left(\bigotimes_{j \in B} V\right) $. This allows us to perform a change of basis of the form $(U \otimes V) P_{ij}(U\otimes V)^{\dag}$ for each projector $P_{ij}$.
Therefore, by the Schmidt decomposition, we may assume without loss of generality that $\ket{\psi}= \sum_{i=1}^d \lambda_i \ket{i}\ket{i}$, where $\lambda_i \geqslant 0$ and the $\lambda_i$ are in non-increasing order.
If $\ket{\psi}$ is a product state, then the Hamiltonian is clearly classical, since $P$ is diagonal in this basis - it is the projector onto $\ket{1}\ket{1}$.

So assume that $\ket{\psi}$ is entangled; we first show how to simulate some 1-local operators using mediator qudit gadgets. For three qudits $1,2 \in A$ and $3 \in B$, let $H_0= I-P_{32}$ with groundstate $\ket{\psi}_{32}$ and let $H_{1}=P_{12}$ so that by Lemma~\ref{lem:firstorder} we can simulate
\begin{align*}
\Pi H_1 \Pi &= P_{32}P_{12} P_{32}=\left( \sum_{i,j}\lambda_i \lambda_j I \otimes  \ket{i}\bra{j} \otimes \ket{i}\bra{j}\right) \left( \sum_{k,l} \lambda_k \lambda_l \ket{k}\bra{l} \otimes \ket{k} \bra{l} \otimes I \right) P_{32}
\\
&=\left( \sum_{i,j,l}\lambda_i \lambda_j^2 \lambda_l \ket{j}\bra{l} \otimes  \ket{i}\bra{l} \otimes \ket{i}\bra{j}\right) \left(\sum_{m,n} \lambda_m \lambda_n I \otimes \ket{m}\bra{n} \otimes \ket{m}\bra{n}\right)\\
 &=\left( \sum_{i,j,n}\lambda_i \lambda_j^4 \lambda_n\ket{j}\bra{j} \otimes  \ket{i}\bra{n} \otimes \ket{i}\bra{n}\right) =R_1 P_{32}
\end{align*}
where $R$ is the single qudit operator $R=\sum_j \lambda_j^4 \ket{j}\bra{j}$.

We can now therefore assume we also have access to the 1-local interaction $R$ on any qudit in $A$. Let $H_1= (\alpha+\beta^2)P_{12}$ and $H_2= \beta(P_{12}-R_1)$ for some arbitrary $\alpha, \beta \in \R$, with $H_0= I-P_{32}$ as before. We note that $\Pi H_2\Pi=0$, so that by Lemma~\ref{lem:secondorder}, we can simulate
\begin{align*}
\Pi \left[H_1- H_2 (H_0)^{-1} H_2 \right]\Pi&= P_{32}\left[ (\alpha+\beta^2) P_{12} - \beta^2 P_{12}(I-P_{32})P_{12}\right] P_{32}\\
&=\alpha P_{32}P_{12} P_{32}+\beta^2 (P_{32}P_{12} P_{32})^2=(\alpha R_1+\beta^2 R_1^2) P_{32} 
\end{align*}
By a symmetric argument, we can also simulate the 1-local interaction $\alpha R + \beta^2 R^2$ on any qudit in $B$.
To complete the proof, we consider the following two separate cases:
\begin{description}
\item{\bf (i) $R$ has a degenerate eigenspace with non-zero eigenvalue.}

Suppose there exists $\mu >0$ such that $J=\left\{i  \:|\:  \lambda_i = \mu\right\}\subseteq  \{1,2,\dots,d\}$ has two or more elements.
Then $R^2-2\mu^4 R +\mu^8I=(R-\mu^4 I)^2$ is positive semidefinite with ground space projector $\Pi= \sum_{i \in J} \proj{i}$.
By projecting all qudits into this subspace with Lemma \ref{lem:firstorder}, we can simulate a Hamiltonian of interactions of the form
\[ (\Pi \otimes \Pi)P (\Pi \otimes \Pi) = \mu^2 \sum_{i,j \in J} \ket {i}\bra{j} \otimes \ket{i} \bra{j}\]
on a bipartite lattice. This interaction is exactly the alternative $SU(d')$ invariant interaction for $d'=|J|$ (see equation \ref{eq:alternativeSUdproj}), which is universal by Corollary~\ref{cor:alternativeSUd}.

\item{\bf (ii) The eigenspaces of $R$ with non-zero eigenvalue are non-degenerate.}

Without loss of generality, assume that the $ \lambda_i$ are ordered in non-increasing order. The assumption that $\ket{\psi}$ is entangled implies that $\lambda_1, \lambda_2>0$. Since we are not in case (i), we know that $\lambda_1>\lambda_2$ and $\lambda_2> \lambda_i$ for all $i \neq 1,2$.
Then the operator $H_0= R^2-(\lambda_1^4+ \lambda_2^4)R +\lambda_1^4\lambda_2^4 I$ has two-dimensional ground space with projector $\Pi=\proj{1}+\proj{2}$.
Therefore by Lemma \ref{lem:firstorder}, we can simulate 
\begin{align*}
 \Pi P \Pi &= \sum_{i,j \in \{1,2\}}\lambda_i \lambda_j \ket {i}\bra{j} \otimes \ket{i} \bra{j}\\
&=\frac{\lambda_1\lambda_2}{2}(XX-YY)+\frac{\lambda_1^2+\lambda_2^2}{4}(ZZ+I)+\frac{\lambda_1^2-\lambda_2^2}{4}(ZI+IZ)
\end{align*}
where we have identified $\ket{1}$ and $\ket{2}$ with the qubit basis states $\ket{0}$ and $\ket{1}$, and $X,Y,Z$ are the standard qubit Pauli matrices.

The 2-local part of this interaction was shown to be universal in \cite{Piddock-Montanaro}, even when the interactions are restricted to a bipartite interaction graph. 
It remains to note that the gadget for removing the 1-local part of an interaction presented in \cite{Cubitt-Montanaro} takes place on a bipartite interaction graph.
\end{description}
\end{proof}

\section{$SU(2)$ Heisenberg interaction on qudits of dimension $d$}

Next we consider the $SU(2)$ Heisenberg interaction in local dimension $d$.
Let $S^x, S^y, S^z$ form a $d$-dimensional irreducible representation of $\mathfrak{su}(2)$ corresponding to the qubit operators $\sigma^x=X/2,\sigma^y=Y/2,\sigma^z=Z/2$. As a representation they must satisfy $[S^a,S^b]=\sum_{c}i\epsilon_{abc}S^c$, where $\epsilon_{abc}$ is the completely antisymmetric Levi-Civita symbol which satisfies the following standard identities:
\begin{equation}
\label{eq:LCidentity}
\sum_{a} \epsilon_{abc} \epsilon_{aef}=\delta_{be}\delta_{cf}-\delta_{bf}\delta_{ce} 
\quad \Rightarrow \quad
\sum_{a,b} \epsilon_{abc} \epsilon_{abf}=2\delta_{cf}.
\end{equation}
Then the $SU(2)$ Heisenberg interaction on qudits of dimension $d$ is defined by 
\[h=\sum_{a} S^a \otimes S^a. \]
We first prove some preliminary technical results that will be useful later on. 

The trivial representation of $\mathfrak{su}(2)$ has Young diagram of a single column of two boxes.
Let $R^{(d)}$ be the unique $d$-dimensional representation such that $R^{(d)}(\sigma^a)=S^a$, which has a Young diagram of a single row of $d-1$ boxes. 
The Casimir eigenvalue of the $d$-dimensional representation is given by $\lambda:=(d^2-1)/4$ by equation (\ref{eq:casimirvalue}), and so $\sum_a S^aS^a=\lambda I$. 

The tensor product of two $d$-dimensional representations has a direct sum decomposition into all odd-dimensional representations of sizes $1,3,\dots,2d-1$ (this can be seen using the Young diagram method, as described for example in \cite{FuchsSchweigert}):
\begin{equation}
\label{eq:su2tensordecomp}
R^{(d)}\otimes R^{(d)}=R^{(1)} \oplus R^{(3)} \oplus \dots \oplus R^{(2d-1)}
\end{equation}
Letting $s=(d-1)/2$, this is the familiar decomposition of the total spin of two particles of spin $s$.

For two qudits of dimension $d$ labelled $E$ and $F$, let $H_0=h_{EF}+\lambda I=\frac{1}{2}\sum_{a}(S_E^a+S_F^a)(S_E^a+S_F^a)$, which is half the Casimir operator in the representation $\{S^a_E+S^a_F\}_a$, so has eigenspace decomposition as given in equation (\ref{eq:su2tensordecomp}), with eigenvalues half of the corresponding Casimir eigenvalue for that representation.

Let $\ket{\psi_{EF}}$ be the state corresponding to the trivial one dimensional representation in the decomposition, for which $(S_E^a+S_F^a)\ket{\psi_{EF}}=0$ for all $a$. In the standard choice of basis this is given by 
\[\ket{\psi_{EF}}=\frac{1}{\sqrt{d}}\sum_{i=0}^{d-1}(-1)^i\ket{i}_E\ket{d-i}_F.\]
The following identities involving $\ket{\psi_{EF}}$ can be derived from the fact that $\bra{\psi_{EF}}M_E\ket{\psi_{EF}}=\frac{1}{d}\tr(M)$ for any single qudit interaction $M$ and the trace formulas provided in \cite{Okubo77}.
\begin{equation}
\label{eq:psiSpsi}
\bra{\psi_{EF}}S_E^a\ket{\psi_{EF}}=0, 
\qquad \bra{\psi_{EF}}S_E^aS_E^b\ket{\psi_{EF}}=\frac{\lambda}{3}\delta_{ab},
\qquad \bra{\psi_{EF}}S_E^aS_E^bS_E^c\ket{\psi_{EF}}=\frac{i\lambda}{6}\epsilon_{abc}
\end{equation}
\begin{equation}
\label{eq:psiSSSSpsi}
\bra{\psi_{EF}}S_E^aS_E^bS_E^cS_E^e\ket{\psi_{EF}}=\frac{\lambda}{15}\left((\lambda-2)\delta_{ac}\delta_{be}+(\lambda+\tfrac{1}{2})(\delta_{ab}\delta_{ce}+\delta_{ae}\delta_{bc})\right)
\end{equation}
In particular the second equation of (\ref{eq:psiSpsi}) shows that the states $\{S_E^a \ket{\psi_{EF}}\}_{a=1}^3$ are orthogonal; in fact they span the space on which $S_E^a+S_F^a$ acts as the 3 dimensional adjoint representation in the decomposition, since $(S_E^a+S_F^a)\ket{\psi_{EF}}=0$ implies $(S_E^a+S_F^a) S_E^b\ket{\psi_{EF}}=[S_E^a,S_E^b]\ket{\psi_{EF}}$.
We can check that $H_0$ has eigenvalue 1 on this space:
\[H_0 S_E^b\ket{\psi_{EF}}=\frac{1}{2}\sum_{a}[S_E^a,[S_E^a,S_E^b]]\ket{\psi_{EF}}=\frac{1}{2}\sum_{a,c,e}-\epsilon_{ace}\epsilon_{abc}S_E^e\ket{\psi_{EF}}=S_E^b\ket{\psi_{EF}}. \]
 Finally we wish to show that the states $\left(\frac{1}{2}\{S_E^b,S_E^c\}-\frac{\lambda}{3}\delta_{bc}\right)\ket{\psi_{EF}}$  are in the 5-dimensional eigenspace of $H_0$ with eigenvalue 3.
\begin{align*}
H_0 S_E^b S_E^c \ket{\psi_{EF}}&=\frac{1}{2}\sum_{a} (S_E^a+S_F^a)(S_E^a+S_F^a) S_E^b S_E^c \ket{\psi_{EF}}
=\frac{1}{2}\sum_{a} [S_E^a,[S_E^a,S_E^b S_E^c]] \ket{\psi_{EF}}\\
&=\frac{1}{2}\sum_{a} \left( [S_E^a,[S_E^a,S_E^b ]]S_E^c+2[S_E^a,S_E^b] [S_E^a,S_E^c ]+S_E^b[S_E^a[S_E^a,S_E^c ]] \right)\ket{\psi_{EF}}\\
&=-\frac{1}{2}\sum_{a,e,f}\left(\epsilon_{abe}\epsilon_{aef}S_E^fS_E^c+2\epsilon_{abe}\epsilon_{acf}S_E^eS_E^f+\epsilon_{ace}\epsilon_{aef}S_E^bS_E^f\right)\ket{\psi_{EF}}\\
&=\left(2S_E^bS_E^c-\sum_{e,f}(\delta_{bc}\delta_{ef}-\delta_{bf}\delta_{ce})S_E^eS_E^f\right)\ket{\psi_{EF}}
=\left(2S_E^bS_E^c-\delta_{bc}\lambda I + S^c_ES^b_E\right)\ket{\psi_{EF}}
\end{align*}
where we have used equation (\ref{eq:LCidentity}) and $\sum_{e}S^eS^e=\lambda I$. This implies that $H_0\left(\frac{1}{2}\{S_E^b,S_E^c\}-\frac{\lambda}{3}\delta_{bc}I\right)\ket{\psi_{EF}}=3\left(\frac{1}{2}\{S_E^b,S_E^c\}-\frac{\lambda}{3}\delta_{bc}I\right)\ket{\psi_{EF}}$ as desired.

\subsection{Simulating $h^2$ with $h$}
\begin{lem}
\label{lem:htoh2}
A Hamiltonian consisting entirely of $SU(2)$ Heisenberg interactions $h$ can simulate a Hamiltonian of the form $\sum_{ij} \alpha_{ij} h_{ij}+\beta_{ij} h_{ij}^2$ for arbitrary $\alpha_{ij},\beta_{ij}\in \R$ and $\beta_{ij}\ge 0 $. 
\end{lem}

\begin{proof}
To apply an arbitrary interaction of the form $\alpha h +\beta h^2$ across qudits 1 and 2, we will use a mediator gadget with a pair of mediator qudits labelled $E,F$ under the heavy interaction $H_0=h_{EF}+\lambda I$ for $\lambda=\frac{d^2-1}{4}$ as in the previous section. Let $\Pi=\proj{\psi_{EF}}$ be the projector onto the ground state of $H_0$.

This will be a fourth-order gadget so we must define Hamiltonians $H_1,H_2,H_3,H_4$ in order to apply Lemma~\ref{lem:fourthorder}.
Let \[H_4=\mu_2(h_{1E}+h_{2E})=\mu_2\sum_{a}(S_1^a+S_2^a)S_E^a=\mu_2\sum_{a}\widetilde{S}^aS_E^a, \] where  $\widetilde{S}^a =S_1^a+S_2^a$, and let $H_1=\mu_1 h_{12}$, $H_2=\frac{2\mu_2^2\lambda}{3}(h_{12}+\lambda I)$, and  $H_3=-\frac{\mu_2^3\lambda}{3}(h_{12}+\lambda I)$, where $\mu_1$, $\mu_2$ are real coefficients to be chosen later.
Note that $h_{12}+\lambda I = \frac{1}{2} \sum_a \widetilde{S}^a \widetilde{S}^a$.
$H_1,H_2,H_3$ all commute with $\Pi$, so are block diagonal with respect to the split $ \mathcal{H}_- \oplus \mathcal{H}_+$.
We can use equation (\ref{eq:psiSpsi}) to check that the remaining condition of Lemma~\ref{lem:fourthorder} is satisfied,
\[\Pi H_4 \Pi=\mu_2\sum_{a}(S_1^a+S_2^a)\bra{\psi_{EF}}S_E^a\ket{\psi_{EF}}\Pi=0.\]
Since $H_0 S_E^b\ket{\psi_{EF}}= S_E^b\ket{\psi_{EF}}$ (when viewing $H_0$ as an operator only on $E$ and $F$), we have $H_0 H_4 \Pi=H_4\Pi$. This significantly simplifies the calculations required to determine the effective interaction produced using Lemma~\ref{lem:fourthorder}:
\begin{align*}
\Pi H_4 H_0^{-1}H_4 \Pi&=\Pi (H_4)^2 \Pi=\mu_2^2\sum_{a,b} (S_1^a+S_2^a)(S_1^b + S_2^b)\bra{\psi_{EF}}S_E^aS_E^b\ket{\psi_{EF}}\Pi\\
&=\frac{\mu_2^2\lambda}{3}\sum_{a,b} \delta_{ab}\widetilde{S}^a \widetilde{S}^b\Pi
=\frac{2\mu_2^2\lambda}{3}(h_{12}+\lambda I)\Pi=\Pi H_2 \Pi;
\end{align*}
\begin{align*}
\Pi H_4 H_0^{-1}H_4 H_0^{-1}H_4 \Pi&=\Pi (H_4)^3 \Pi=\mu_2^3\sum_{a,b,c} (S_1^a+S_2^a)(S_1^b + S_2^b)(S_1^c+S_2^c)\bra{\psi_{EF}}S_E^aS_E^bS_E^c\ket{\psi_{EF}}\Pi\\
&=\frac{\mu_2^3\lambda}{6}\sum_{a,b,c}i\epsilon_{abc}\widetilde{S}^a \widetilde{S}^b \widetilde{S}^c\Pi
=\frac{\mu_2^3\lambda}{6}\sum_{c}  \widetilde{S}^c  \widetilde{S}^c \Pi=\frac{\mu_2^3\lambda}{3}(h_{12}+\lambda I) \Pi\\
&=-\Pi H_3\Pi.
\end{align*}
In the final set of equations we have used the following useful identity which holds for any operators $\widetilde{S}^a$ which form a representation of $\mathfrak{su}(2)$ and thus satisfy $[\widetilde{S}^a,\widetilde{S}^b]=\sum_c i\epsilon_{abc} \widetilde{S}^c$:
\begin{align}
\label{eq:EabcSaSb}
\sum_{a,b}i\epsilon_{abc}\widetilde{S}^a\widetilde{S}^b&=\sum_{a,b}\frac{i}{2}\left(\epsilon_{abc}\widetilde{S}^a\widetilde{S}^b +\epsilon_{bac}\widetilde{S}^b\widetilde{S}^a\right)=\frac{i}{2}\sum_{a,b}\epsilon_{abc}[\widetilde{S}^a,\widetilde{S}^b]\\
&=-\frac{1}{2}\sum_{a,b}\epsilon_{abc}\epsilon_{abe}\widetilde{S}^e=-\delta_{ce}\widetilde{S}^e=-\widetilde{S}^c.
\end{align}

Let $A=\Pi H_4 H_0^{-1}H_2 H_0^{-1}H_4 \Pi$ and $B=\Pi H_4 H_0^{-1}H_4 H_0^{-1}H_4 H_0^{-1}H_4 \Pi$, so that by Lemma~\ref{lem:fourthorder}  $\Delta H_0 + \Delta^{3/4} H_4 + \Delta^{1/4} H_3 + \Delta^{1/2}H_2+H_1$ simulates $(H_1)_{--}+A-B$. 
First we calculate $A$ using equation (\ref{eq:psiSpsi}) to find
\[A=\Pi H_4 H_2 H_4 \Pi=\frac{\mu_2^4\lambda}{3}\sum_{a,b,c} \widetilde{S}^a\widetilde{S}^b \widetilde{S}^b \widetilde{S}^c \bra{\psi_{EF}}S_E^aS_E^c\ket{\psi_{EF}}\Pi
=\frac{\mu_2^4\lambda^2}{9}\sum_{a,b}\widetilde{S}^a\widetilde{S}^b \widetilde{S}^b \widetilde{S}^a \Pi.\]
Calculating $B$ is more complicated:
\[B=\Pi (H_4)^2 H_0^{-1}(H_4)^2 \Pi=\mu_2^4\sum_{a,b,c,e}\widetilde{S}^a\widetilde{S}^b \widetilde{S}^c \widetilde{S}^e  \bra{\psi_{EF}}S_E^aS_E^bH_0^{-1}S_E^cS_E^e\ket{\psi_{EF}}\Pi. \]
We therefore need to calculate $ \bra{\psi_{EF}}S_E^aS_E^bH_0^{-1}S_E^cS_E^e\ket{\psi_{EF}}$, which can be done by recalling from above that $\left(\frac{1}{2}\{S_E^b,S_E^c\}-\frac{\lambda}{3}\delta_{bc}I\right)\ket{\psi_{EF}}$ is in the eigenspace of $H_0$ with eigenvalue 3, and $[S_E^b,S_E^c]=\sum_e f_{bce} S_E^e$ for some coefficients $f_{bce}$, so $[S_E^b,S_E^c]\ket{\psi_{EF}}$ is in the eigenspace of $H_0$ with eigenvalue 1. Then we have
\begin{align*}
\bra{\psi_{EF}}S_E^aS_E^bH_0^{-1}S_E^cS_E^e\ket{\psi_{EF}}&=\bra{\psi_{EF}}S_E^aS_E^bH_0^{-1}\left(\frac{1}{2}\{S_E^c,S_E^e\}-\frac{\lambda}{3}\delta_{ce}I +\frac{1}{2}[S_E^c,S_E^e]+\frac{\lambda}{3}\delta_{ce}I\right)\ket{\psi_{EF}}\\
&=\bra{\psi_{EF}}S_E^aS_E^b\left(\frac{1}{3}\left(\frac{1}{2}\{S_E^c,S_E^e\}-\frac{\lambda}{3}\delta_{ce}I \right) +\frac{1}{2}[S_E^c,S_E^e]\right)\ket{\psi_{EF}}\\
&=\bra{\psi_{EF}}S_E^aS_E^b\left(\frac{2}{3}S_E^cS_E^e-\frac{1}{3}S_E^eS_E^c-\frac{\lambda}{9}\delta_{ce} I \right)\ket{\psi_{EF}}\\
&=\frac{\lambda}{45}\left( (\lambda-\tfrac{9}{2})\delta_{ac}\delta_{be}+(\lambda+3)\delta_{ae}\delta_{bc}+(\tfrac{1}{2}-\tfrac{2}{3}\lambda)\delta_{ab}\delta_{ce}\right)
\end{align*}
where we have used equations (\ref{eq:psiSpsi}) and (\ref{eq:psiSSSSpsi}) in the last equality.
And so we have
\begin{align*}
A-B=\frac{\mu_2^4\lambda}{45}\sum_{a,b}\left( (\tfrac{9}{2}-\lambda)\widetilde{S}^a\widetilde{S}^b \widetilde{S}^a \widetilde{S}^b+(4\lambda-3)\widetilde{S}^a\widetilde{S}^b \widetilde{S}^b \widetilde{S}^a+(\tfrac{2}{3}\lambda-\tfrac{1}{2})\widetilde{S}^a\widetilde{S}^a \widetilde{S}^b \widetilde{S}^b\right).
\end{align*}
Then we substitute in the following relations which are an easy consequence of equation (\ref{eq:EabcSaSb}): 
\begin{align*}\sum_{a,b}\widetilde{S}^a\widetilde{S}^b\widetilde{S}^a\widetilde{S}^b
&=\sum_{a,b}\left(\widetilde{S}^a\widetilde{S}^a\widetilde{S}^b\widetilde{S}^b+ \widetilde{S}^a[\widetilde{S}^b,\widetilde{S}^a]\widetilde{S}^b\right) 
=\sum_{a,b}\left(\widetilde{S}^a\widetilde{S}^a\widetilde{S}^b\widetilde{S}^b+\sum_c i\epsilon_{bac}\widetilde{S}^a\widetilde{S}^c\widetilde{S}^b\right) \\
&=\sum_{a,b}\widetilde{S}^a\widetilde{S}^a\widetilde{S}^b\widetilde{S}^b -\sum_c\widetilde{S}^c\widetilde{S}^c,
\end{align*}
\begin{align*}\sum_{a,b}\widetilde{S}^a\widetilde{S}^b\widetilde{S}^b\widetilde{S}^a
&=\sum_{a,b}\left(\widetilde{S}^a\widetilde{S}^b\widetilde{S}^a\widetilde{S}^b+ \widetilde{S}^a\widetilde{S}^b[\widetilde{S}^b,\widetilde{S}^a]\right) =\sum_{a,b}\left(\widetilde{S}^a\widetilde{S}^b\widetilde{S}^a\widetilde{S}^b+\sum_c i\epsilon_{bac}\widetilde{S}^a\widetilde{S}^b\widetilde{S}^c\right)\\ &=\sum_{a,b}\widetilde{S}^a\widetilde{S}^b\widetilde{S}^a\widetilde{S}^b +\sum_c\widetilde{S}^c\widetilde{S}^c
=\sum_{a,b}\widetilde{S}^a\widetilde{S}^a\widetilde{S}^b\widetilde{S}^b
\end{align*}
%
%and similarly $\sum_{a,b}\widetilde{S}^a\widetilde{S}^b\widetilde{S}^a\widetilde{S}^b =\sum_{a,b}\widetilde{S}^a\widetilde{S}^a\widetilde{S}^b\widetilde{S}^b +\sum_c\widetilde{S}^c\widetilde{S}^c$ 
to get 
\begin{align*}
A-B&=\frac{\mu_2^4\lambda}{45}\left(\left(\frac{11}{3}\lambda+1\right)\sum_{a,b}\widetilde{S}^a\widetilde{S}^a\widetilde{S}^b\widetilde{S}^b +\left(\lambda-\frac{9}{2}\right)\sum_ c\widetilde{S}^c\widetilde{S}^c \right)\\
&=\mu_2^4\tfrac{\lambda}{135}\left(4(11\lambda+3)h_{12}^2+(88\lambda^2+30\lambda-27)h_{12}+(44\lambda^2+18\lambda-27)\lambda I \right)\Pi
\end{align*}
where we have used $\sum_c \widetilde{S}^c \widetilde{S}^c=2(h_{12}+\lambda I)$.

Let $\mu_1=\alpha-\mu_2^4\frac{\lambda}{135}(88\lambda^2+30\lambda-27)$ and  $\mu_2=(135\beta/4(11\lambda^2+3\lambda))^{1/4}$, noting that $11\lambda^2+3\lambda$ is positive for all $d\ge2$. Then by Lemma~\ref{lem:fourthorder} we simulate $\Pi H_1 \Pi +A-B=(\alpha h_{12}+\beta h_{12}^2 +c I ) \Pi$ for some $c \in \R$.

Finally, since this is a fourth-order gadget, we must check if there is any cross-gadget interference when we use multiple gadgets in parallel. 
Let $\Pitot$ be the projector onto the ground space of all gadgets being applied in parallel. 
By Corollary~\ref{cor:4thorderinterference}, the interference between gadgets $i$ and $j$ is given by 
\[-\frac{1}{2} \Pitot\left[H_4^{(i)},H_4^{(j)}\right]^2 \Pitot. \]
If $H_4^{(i)}$ and $H_4^{(j)}$ commute then clearly there is no interference. Assume without loss of generality that gadget $i$ simulates an interaction between qudits  $1$ and $2$ with $H_4^{(i)}=\mu_2^{(i)}( h_{1E_i}+h_{2E_i})$ and gadget $j$ simulates an interaction between qudits $1$ and $3$ with $H_4^{(j)}=\mu_2^{(j)} (h_{1E_j}+h_{3E_j})$. Normalising by a factor of $(\mu_2^{(i)})^2(\mu_2^{(j)})^2$ for convenience, the cross-gadget interference is proportional to 
\begin{align*}
-\frac{1}{2(\mu_2^{(i)})^2(\mu_2^{(j)})^2} \Pitot\left[H_4^{(i)},H_4^{(j)}\right]^2 \Pitot
&=-\frac{1}{2}  \Pitot \left[ \sum_{a}(S_1^a+S_2^a)S_{E_i}^a,\sum_b (S_1^b+S_3^b)S_{E_j}^b\right]^2\Pitot\\
&=-\frac{1}{2}  \sum_{a,b,c,e}[S_1^a,S_1^b] [S_1^c,S_1^e]\Pitot S_{E_i}^a S_{E_i}^c  S_{E_j}^b S_{E_j}^e\Pitot\\
&=-\frac{1}{2} \frac{\lambda^2}{9}\sum_{a,b}[S_1^a,S_1^b][S_1^a,S_1^b]
=\frac{\lambda^2}{18} \sum_{a,b,c,e}\epsilon_{abc}\epsilon_{abe}S_1^cS_1^e\\
&=\frac{\lambda^2}{9} \sum_{c}S_1^cS_1^c=\frac{\lambda^3}{9}I
\end{align*}
where we have used equation (\ref{eq:psiSpsi}) in the third equality. Therefore the cross-gadget interference is proportional to the identity, which corresponds only to an unimportant energy shift, and so can be ignored.
\end{proof}

\subsection{$h^2$ and $h$ simulate qutrit Heisenberg, which simulates $h^2$ and $h$}
\label{sec:h2h}

Let $C$ be the Casimir operator corresponding to the $\{S^a_1 + S^a_2\}_a$ representation of $\mathfrak{su}(2)$. Given access to $h^2$ and $h$ interactions, we can produce the two-qudit interaction 
\begin{align*}
(C-2I)^2 &=\left(\sum_a (S_1^a+S_2^a)(S_1^a+S_2^a)-2I\right)^2=(2h_{12}+2\lambda I-2I)^2\\
&=4\left(h_{12}^2+2(\lambda-1)h_{12}+(\lambda-1)^2 I\right),
\end{align*}
where as before $\lambda = (d^2-1)/4$. This operator is clearly positive semidefinite and has eigenvalue zero only on the 3-dimensional representation in the decomposition (\ref{eq:su2tensordecomp}), since the 3-dimensional representation has Casimir eigenvalue 2.
We will use this 3-dimensional space to encode a logical qutrit. For any 4 qudits $(1,2)$, $(3,4)$, where each pair is restricted to this space, the operator
\[h_{13}+h_{14}+h_{23}+h_{24}=\sum_a (S_1^a+S_2^a)(S_3^a+S_4^a)\]
%\[\Pi\left(h_{13}+h_{14}+h_{23}+h_{24}\right)\Pi=\Pi\left(\sum_a (S_1^a+S_2^a)(S_3^a+S_4^a)\right)\Pi\]
acts as a logical qutrit $SU(2)$ Heisenberg interaction.
So by Lemma~\ref{lem:firstorder} we can simulate any qutrit Hamiltonian of $SU(2)$ Heisenberg interactions using qudit interactions $h$ and $h^2$. Then, by Lemma~\ref{lem:htoh2}, it is possible to simulate any Hamiltonian $H=\sum_{ij} \alpha_{ij} h'_{ij}+\beta_{ij} (h'_{ij})^2$, where $\beta_{ij}\ge 0$ and now $h'$ and $(h')^2$ are the corresponding qutrit interactions. 
In particular one can set $\alpha_{ij}=\beta_{ij}$ and simulate $\sum_{ij} \beta_{ij}( h'_{ij}+(h'_{ij})^2)$. Then $h'+(h')^2$ is the $SU(3)$ Heisenberg interaction, which is universal by Theorem~\ref{thm:sud} (even with non-negative weights). This completes the proof of the following theorem:

\begin{repthm}{thm:su2}
For any $d \ge 2$, the $SU(2)$ Heisenberg interaction $h= S^x \otimes S^x + S^y \otimes S^y + S^z \otimes S^z$, where $S^x$, $S^y$, $S^z$ are representations of the Pauli matrices $X$, $Y$, $Z$, is universal.
\end{repthm}

\section{Bilinear-biquadratic interaction in dimension 3}

We finally consider an important variant of the $SU(2)$ Heisenberg model: the bilinear-biquadratic spin-1 Heisenberg model (i.e.\ in local dimension 3). Write $X_3$, $Y_3$, $Z_3$ for matrices such that $\{iX_3,iY_3,iZ_3\}$ generate a $3$-dimensional irreducible representation of $\mathfrak{su}(2)$. For example, we can take
\[ X_3 = \frac{1}{\sqrt{2}} \begin{pmatrix} 0 & 1 & 0\\ 1& 0& 1\\ 0 & 1 & 0\end{pmatrix},\;\;\;\;Y_3 = \frac{i}{\sqrt{2}} \begin{pmatrix} 0 & -1 & 0\\ 1& 0& -1\\ 0 & 1 & 0\end{pmatrix},\;\;\;\;Z_3 = \begin{pmatrix} 1 & 0 & 0\\ 0& 0& 0\\0 & 0 & -1\end{pmatrix}; \]
note that these obey the same commutation relations as the Pauli matrices (up to a scaling constant). Then the Heisenberg interaction is
\[ h = X_3 \otimes X_3 + Y_3 \otimes Y_3 + Z_3 \otimes Z_3. \]
Consider the algebra generated by $h$. We have $h^3 = h - 2 h^2 + 2 I$, so up to scaling and an identity term any nontrivial interaction in this algebra can be written as $h^{(\theta)} := (\cos \theta) h + (\sin \theta) h^2$ for some $\theta$.
Let $\alpha=\cos \theta$ and $\beta=\sin \theta$. Because of our freedom to choose the signs of interactions, we can further assume that $0 \le \theta \le \pi$, and thus $\beta \ge 0$. Then any Hamiltonian produced from such interactions can be written, up to an overall identity term, as
\[ H = \sum_{i < j} a_{ij} h^{(\theta)}_{ij}. \]
This model is known as the (general) bilinear-biquadratic Heisenberg model and has been a popular object of study~\cite{Affleck89,Harada02,Kennedy90,Lauchli06}. The special case $\theta = \arctan 1/3$ is the interaction proportional to $h + \frac{1}{3}h^2$ occurring in the famous AKLT model~\cite{AKLT}, which was handled in Lemma \ref{lem:AKLTmodel}. We also already showed that the cases $\theta \in \{0,\pi/4\}$ are universal in the previous section ($\pi/4$ corresponds to the $SU(3)$ Heisenberg interaction); here we prove universality for all other values of $\theta$.

It is easy to check that $h$ has three eigenspaces with eigenvalues $-2$, $-1$, 1 and dimensions 1, 3, 5 respectively. 
Therefore $h^{(\theta)}$ has eigenvalues $4\beta-2\alpha, \beta-\alpha,\beta+\alpha$ with respect to the same eigenspaces.
In addition, $h^2$ is proportional to the projector onto $\ket{\psi}=\ket{02}-\ket{11}+\ket{20}$ plus a multiple of the identity.
Depending on $\theta$, $h^{(\theta)}$ has the following properties:
\begin{itemize}
\item $\theta = 0$: $h^{(\theta)} = h$. The Heisenberg model.
\item $0 < \theta < \arctan 1/3$: ground state nondegenerate and equal to $\ket{02}-\ket{11}+\ket{20}$.
\item $\theta = \arctan 1/3$: ground space 4-fold degenerate (the AKLT model).
\item $\arctan 1/3 < \theta < \pi/2$: ground space 3-fold degenerate and spanned by
\be \label{eq:basis3dim} \{\ket{01}-\ket{10},\ket{12}-\ket{21},\ket{02}-\ket{20}\}. \ee
\item $\theta = \pi/2$: ground space 8-fold degenerate and the orthogonal complement of $\ket{02}-\ket{11}+\ket{20}$. The case $h^{(\theta)} = h^2$.
\item $\pi/2 < \theta < \pi$: ground space 5-fold degenerate. 
\end{itemize}

The special case $\theta = \pi/4$ gives the qutrit swap operator (up to rescaling and subtracting an identity term), which is in addition $SU(3)$-invariant. For $\theta > \pi/4$, the highest energy state is nondegenerate and is $\ket{02}-\ket{11}+\ket{20}$.

\subsection{Mediator gadget}
We first consider the case where the state $\ket{\psi}=\ket{02}-\ket{11}+\ket{20}$ is either the unique ground state or highest excited state of $h^{(\theta)}$.
\begin{lem}
\label{lem:mediatorsu2}
Let $\theta \in (0,\arctan 1/3) \cup (\pi/4, \pi)\setminus \{\arctan 2\}$. Then $h^{(\theta)}$ is universal.
\end{lem}

\begin{proof}
Our strategy will be to use a second-order gadget via Lemma \ref{lem:secondorder} to implement the effective interaction $h^{(\theta')}$ for any choice of $\theta'$. In particular this allows us to simulate the interaction $h^{(\pi/4)}$ which is the qutrit swap operator -- the unique $SU(3)$ invariant interaction shown to be universal in  Theorem~\ref{thm:sud}.
 To use this approach, we need to define Hamiltonians $H_0$, $H_1$, $H_2$ on a system of 4 qutrits. We label these qutrits $1,2,3,4$ where qutrits 3 and 4 are mediator qutrits, and the effective interaction $h^{(\theta')}$  is simulated on qutrits 1 and 2.

The condition on $\theta$ implies that $\beta>0$ and $\alpha> 3\beta$ or $\alpha< \beta$. 
Consider the operator $h^{(\theta)}+(2\alpha-4\beta)I$, which annihilates $\ket{\psi}=\ket{02}-\ket{11}+\ket{20}$, and has eigenvalues $\alpha-3\beta$ and $3\alpha-3\beta$ on the two eigenspaces of $h$ with dimension 3 and 5 respectively, which in turn correspond to eigenvalues $-1$ and $+1$.
If $\alpha> 3\beta$ then both of these eigenvalues are positive and we set $H_0=h^{(\theta)}_{34}+(2\alpha-4\beta)I$, while if $\alpha<\beta$ then both of these eigenvalues are negative and the proof will continue analogously with $H_0=-(h^{(\theta)}_{34}+(2\alpha-4\beta)I)$.

In either case, $\Pi=I \otimes \ket{\psi_{34}}\bra{\psi_{34}}$ is the projector onto the ground space of $H_0$.
Let $H_1=\lambda_1 h^{(\theta)}_{12}$ for some $\lambda_1 \in \R$, so that $H_1$ commutes with $\Pi$, and $\Pi H_1\Pi =\lambda_1 h_{12}^{(\theta)} \Pi$. 
 Then we choose
 \[H_2=\lambda_2 \left(h^{(\theta)}_{13}+h^{(\theta)}_{23}-\frac{8\beta}{3} I\right)= \lambda_2(\alpha-\beta/2)A+\lambda_2\beta B\] where $A=h_{13}+h_{23}$, $B=h_{13}^2+\tfrac{1}{2}h_{13}+ h_{23}^2+\tfrac{1}{2}h_{23}-\tfrac{8}{3}I$, and $\lambda_2 \in \R$. 
 It is easy to check that for any $\ket{\phi_{12}}$, $h_{13}\ket{\phi_{12}}\ket{\psi_{34}}$ and $h_{23}\ket{\phi_{12}}\ket{\psi_{34}}$ are in the eigenspace of $h_{34}$ with eigenvalue $-1$, and therefore that $A\Pi$ has support only on the eigenspace of $H_0$ with eigenvalue $\alpha-3\beta$.
 Similarly, one can check that $(h_{13}^2 +\frac{1}{2}h_{13}-\frac{4}{3}I)\ket{\phi_{12}}\ket{\psi_{34}}$ and $(h_{23}^2 +\frac{1}{2}h_{23}-\frac{4}{3}I)\ket{\phi_{12}}\ket{\psi_{34}}$ are in the eigenspace of $h_{34}$ with eigenvalue $+1$, which implies that $B\Pi$ has support only on the eigenspace of $H_0$ with eigenvalue $3\alpha-3\beta$. 
 
Therefore neither $A\Pi$ or $B\Pi$ have support on the eigenspace of $H_0$ with eigenvalue 0, and so $\Pi H_2 \Pi =0$ as required to apply Lemma \ref{lem:secondorder}. The second-order term is given by
\[\Pi H_2 H_0^{-1} H_2 \Pi=\lambda_2^2\frac{(\alpha-\beta/2)^2}{\alpha-3\beta}\Pi A^2\Pi +\lambda_2^2\frac{\beta^2}{3\alpha-3\beta}\Pi B^2\Pi.\]
Calculating $\Pi A^2 \Pi$ and $\Pi B^2 \Pi$ separately we find that
\[\Pi A^2 \Pi=\Pi (h_{13}^2 +h_{13}h_{23}+h_{23}h_{13}+h_{23}^2)\Pi=\frac{4}{3}(2I +h_{12})\Pi\]
\[\Pi B^2\Pi=\left(\frac{2}{3}h_{12}^2+\frac{1}{3}h_{12}+\frac{2}{9}I\right)\Pi.\]
Then by Lemma~\ref{lem:secondorder}, we can simulate the interaction
\[\Pi H_1 \Pi - \Pi H_2H_0^{-1}H_2\Pi =\left(\lambda_1 h^{(\theta)}_{12}+\lambda_2^2 \tilde{h}_{12}(\theta)\right)\Pi\]
where 
\[\tilde{h}_{12}(\theta)=\frac{2}{9( \alpha - \beta)}\left(\beta^2 h_{12}^2+\frac{6 \alpha^3 - 12 \alpha^2 \beta + 8 \alpha \beta^2 - 3 \beta^3}{\alpha-3\beta}h_{12}+\frac{2 (18 \alpha^3 - 36 \alpha^2 \beta + 23 \alpha \beta^2 - 6 \beta^3)}{3(\alpha-3\beta)}I\right).\]
By repeating the same calculation with $H_2=\lambda_2(h^{(\theta)}_{13}-h^{(\theta)}_{23})$, it is possible to simulate the interaction $\lambda_1 h_{12}^{(\theta)}-\lambda_2^2 \tilde{h}_{12}(\theta)$ instead. For all $\theta$ satisfying the conditions in the lemma, it is easy to check that the 2-local part of $\tilde{h}_{12}(\theta)$ is linearly independent of $h_{12}^{(\theta)}$. So, by choosing $\lambda_1$, $\lambda_2$ appropriately, we can use this gadget to simulate any desired interaction $h^{(\theta')}$ (with an arbitrary weight), and in particular the case $\theta' = \pi/4$.
\end{proof}

\subsubsection{Logical qubit gadget}

In the next case we consider, $h^{(\theta)}$ has a 3-dimensional ground space.

\begin{lem}
\label{lem:logicalsu2}
Let $\theta \in (\arctan 1/3,\arctan 5)$. Then $h^{(\theta)}$ is universal.
\end{lem}

\begin{proof}
In this case, the condition on $\theta$ implies that $0< \beta/5<\alpha < 3 \beta$ and that $h^{(\theta)}$'s ground space is 3-dimensional.
We will construct a second-order gadget that encodes each logical qutrit into one of these 3-dimensional ground spaces of two physical qutrits.
Using Lemma~\ref{lem:secondorder}, we choose $H_0$, $H_1$ and $H_2$ such that the effective interaction between logical qutrits is proprtional to $h+h^2$, the SU(3) invariant SWAP interaction shown to be universal in Theorem~\ref{thm:sud}.
%Then $h^{(\theta)}+(\alpha-\beta)I$ has eigenvalues $0,3\beta-\alpha,2\alpha$. 

By the anti-interference discussion presented in \cite[Lemma 36]{cubitt17}, it will suffice to consider just two logical qutrits encoded in 4 physical qutrits. Let one logical qutrit be encoded into the ground space of $h_{12}^{(\theta)}$ in a pair of physical qutrits labelled $1,2$ and a second logical qutrit be encoded into the ground space of $h^{(\theta)}_{34}$ in a pair of physical qutrits labelled $3,4$. The overall heavy Hamiltonian $H_0$, with an appropriate multiple of the identity to ensure the ground state energy is zero, is given by 
\[H_0=h^{(\theta)}_{12}+h^{(\theta)}_{34}+2(\alpha-\beta)I.\]
Let $\Pi$ be the projector onto the 9 dimensional ground space of $H_0$, in which the two logical qutrits are encoded.
One can check that for $i \in \{1,2\}$ and $j \in \{3,4\}$,
\[\Pi h^{(\theta)}_{ij}\Pi=\left(\frac{1}{4} h^{(\theta)}_L + \beta I\right) \Pi\]
where $h^{(\theta)}_L$ denotes the action of $h^{(\theta)}$ in the logical qutrit space, with respect to the basis (\ref{eq:basis3dim}). 
Let $H_2=\lambda_2(h^{(\theta)}_{13}-h^{(\theta)}_{24})$ so that $\Pi H_2 \Pi=0$. Using a computer algebra package we can calculate the second-order term, remembering that $H_0$ has zero energy on its ground space, and that the $H_0^{-1}$ denotes the inverse computed on the higher energy space only:
\[-\Pi H_2 H_0^{-1} H_2 \Pi=
\frac{\lambda_2^2}{2\alpha(\alpha-3\beta)}\left((-3\alpha^3 + 6 \alpha^2 \beta - 8\alpha\beta^2 + \beta^3)h
-\frac{1}{2}(5 \alpha^3 - 7 \alpha^2\beta + 9 \alpha\beta^2 + \beta^3)h^2+c I\right)\Pi \]
for some $c \in \R$. 

Let $H_1=4\lambda_1 h_{13}$ so that $\Pi H_1\Pi= \lambda_1(h_{L}^{(\theta)} + 4\beta I) \Pi=\lambda_1(\alpha h_L +\beta h_L^2  + 4\beta I ) \Pi$. Then by Lemma \ref{lem:secondorder}, choosing  $H_0$ and $H_2$ as above and setting $\lambda_1=\alpha-\beta$, $\lambda_2=2\sqrt{\alpha}$ will simulate
\[\frac{5\alpha^3 - 8 \alpha^2 \beta + 13 \alpha \beta^2 - 2 \beta^3}{3 \beta-\alpha }\left(h_L+h^2_{L}\right)+\tilde{c}I\]
for some $\tilde{c}$, which is the $SU(3)$ Heisenberg interaction as desired, up to rescaling and deletion of an identity term.
We note that $3\beta-\alpha>0$ and 
\[5\alpha^3 - 8 \alpha^2 \beta + 13 \alpha \beta^2 - 2 \beta^3=(5\alpha-\beta)(\alpha-\sqrt{2} \beta)^2 +(10\sqrt{2}-7)\alpha^2\beta +(3-2\sqrt{2})\alpha\beta^2>0\]
since $\alpha,\beta>0$ and $5\alpha-\beta>0$. Therefore this gadget can only produce positively-weighted interactions, but this restriction is allowed in Theorem~\ref{thm:sud}.
\end{proof}

Combining Theorem \ref{thm:su2}, Lemma~\ref{lem:AKLTmodel}, Lemma~\ref{lem:mediatorsu2} and Lemma~\ref{lem:logicalsu2} yields our final result:

\begin{repthm}{thm:bilinbiq}
Let $h^{(\theta)} := (\cos \theta) h + (\sin \theta) h^2$, where $\theta \in [0,2\pi)$ is an arbitrary parameter and $h$ is the spin-1 Heisenberg interaction. For any $\theta$, $h^{(\theta)}$ is universal.
\end{repthm}

\subsection*{Acknowledgements}

We would like to thank Johannes Bausch and Toby Cubitt for helpful discussions on the topic of this work. SP was supported by the EPSRC. AM was supported by EPSRC Early Career Fellowship EP/L021005/1. No new data were created during this study.

% ------------------------------------------------------------------------------
%\newpage

\appendix

\section{Proofs for fourth-order perturbative gadgets}
\label{app:fourthorderproofs}

In this appendix, we prove Lemma~\ref{lem:fourthorder} and Lemma~\ref{lem:4thorderinterference}.

\begin{replem}{lem:fourthorder}[Fourth-order simulation]
Let $H_0$, $H_1$, $H_2$, $H_3$, $H_4$ be Hamiltonians acting on the same space, such that: $\max\{\|H_1\|,\|H_2\|,\|H_3\|,\|H_4\|\} \le \Lambda$; $H_2$ and $H_3$ are block-diagonal with respect to the split $\mathcal{H}_+ \oplus \mathcal{H}_-$; $(H_4)_{--}=0$. Suppose there exists a local isometry $V$ such that $\Im(V)=\mathcal{H}_-$ and 
\begin{equation} \| V H_{\operatorname{target}} V^\dag - \Pi_-\left(H_1 +H_4 H_0^{-1} H_2 H_0^{-1} H_4 -H_4 H_0^{-1} H_4 H_0^{-1} H_4 H_0^{-1} H_4\right)\Pi_- \| \le \epsilon/2 
\label{eq:4thorder}
\end{equation}
and also that
\begin{equation} (H_2)_{--} = \Pi_- H_4 H_0^{-1}H_4 \Pi_- \quad \text{ and } \quad (H_3)_{--} = -\Pi_- H_4 H_0^{-1}H_4 H_0^{-1}H_4 \Pi_-. 
\label{eq:H2H3}
\end{equation}
Then $H_{\operatorname{sim}} = \Delta H_0 + \Delta^{3/4} H_4 + \Delta^{1/4} H_3 + \Delta^{1/2}H_2+H_1$ $(\Delta/2,\eta,\epsilon)$-simulates $H_{\operatorname{target}}$, provided that $\Delta \ge O(\Lambda^{20}/\epsilon^4+\Lambda^4/\eta^4)$.
\end{replem}

\begin{proof} 
We will follow the presentation of the Schreiffer-Wolff transformation provided in~\cite{Bravyi-Hastings} and~\cite{Bravyi-DiVincenzo-Loss}. 
Let $A=\Delta^{3/4} H_4 + \Delta^{1/4} H_3 + \Delta^{1/2}H_2+H_1$, so that $H_{\operatorname{sim}}=\Delta H_0+A$.
The Schreiffer-Wolff transformation is a unitary operator $e^S$ which maps the low-energy space of $H_{\operatorname{sim}}$ onto $\mathcal{H}_-$, the ground space of $H_0$.
Define $\widetilde{V}=e^{-S}V$, which therefore maps exactly onto the low energy space of $H_{\operatorname{sim}}$.
And, using equation (22) of \cite{Bravyi-Hastings}, we have $\|V-\widetilde{V}\|=\|I-e^{-S}\|=O(\|S\|)=O(\|A\|/\Delta)=O(\Lambda/\Delta^{1/4})\le \eta$, so $\widetilde{V}$ satisfies condition 1 of Definition~\ref{dfn:sim}.

To check condition 2 of Definition~\ref{dfn:sim}, it is necessary to bound
\[ \|H_{\operatorname{sim}}|_{\le  \Delta}-\widetilde{V}H_{\operatorname{target}}\widetilde{V}^{\dagger} \|
=\| V \widetilde{V}^{\dagger} H_{\operatorname{sim}} \widetilde{V} V^{\dagger} -V H_{\operatorname{target}} V^{\dagger} \| 
= \| H_{\eff} - V H_{\operatorname{target}} V^{\dagger}\| \]
where $H_{\eff}= (e^S H_{\operatorname{sim}}e^{-S})_{--}$, which is in general a very complicated operator. To deal with this, we expand $H_{\eff}$ as a Taylor series in $1/ \Delta$.
The first three terms are given in \cite{Bravyi-Hastings} as 
\[H_{\eff,1}=A_{--} \qquad \text{ and } \qquad H_{\eff,2}=-\frac{1}{\Delta}A_{-+}H_0^{-1}A_{+-}\]
\[H_{\eff,3}= \frac{1}{\Delta^2}A_{-+}H_0^{-1}A_{++}H_0^{-1}A_{+-}-\frac{1}{2\Delta^2}(A_{-+}H_0^{-2}A_{+-}A_{--} +\hc)\]
The fourth-order term in the Taylor series can be derived using the techniques of \cite{Bravyi-DiVincenzo-Loss}, where they consider the more general situation where $H_0$ acts non-trivially on its low energy space.
Let $A_{\od}=\Pi_-A\Pi_+ +\Pi_+A\Pi_-$ and $A_{\text{d}}=\Pi_-A\Pi_- +\Pi_+A\Pi_+$  and $S_1=\Delta^{-1}[H_0^{-1},A_{\od}]$.
In the special case we are considering where $(H_0)_{--}=0$, the fourth-order term is given according to equation (3.22) of \cite{Bravyi-DiVincenzo-Loss}  as 
\begin{align*}H_{\eff,4}&=\Pi_-\left(\frac{1}{8}[S_1,[S_1,[S_1,A_{\od}]]] 
-\frac{1}{2} [A_{\od},[\Delta^{-1}H_0^{-1},[A_{\text{d}},[\Delta^{-1}H_0^{-1},[A_{\text{d}},S_1]]]]]\right) \Pi_-\\
&=\frac{1}{2\Delta^3}\Pi_-\Bigl(A H_0^{-2}  A\Pi_- A H_0^{-1} A   -A H_0^{-1} A H_0^{-1} A H_0^{-1} A +A H_0^{-2} A H_0^{-1} A \Pi_- A\\
&\qquad \qquad\qquad+A H_0^{-1} A H_0^{-2} A \Pi_- A -A H_0^{-3} A \Pi_- A \Pi_- A+\hc \Bigr)\Pi_- 
\end{align*}
where the $\hc$ refers to the Hermitian conjugate of all terms contained in the brackets, where the second equality follows from some tedious algebra or the use of a computer algebra package.

Next we substitute in $A=\Delta^{3/4} H_4 + \Delta^{1/4} H_3 + \Delta^{1/2}H_2+H_1$ to get 
\[H_{\eff,1}=\Delta^{1/4} (H_3)_{--} + \Delta^{1/2}(H_2)_{--}+(H_1)_{--}\]
\[H_{\eff,2}=-\Delta^{1/2} \Pi_- H_4 H_0^{-1}H_4 \Pi_-   +O(\Lambda^{2}/\Delta)\]
\begin{align*}
H_{\eff,3}&=\Delta^{1/4}\Pi_- H_4 H_0^{-1}H_4 H_0^{-1}H_4 \Pi_- + \Pi_- H_4 H_0^{-1}H_2 H_0^{-1}H_4 \Pi_- \\ &\qquad-\frac{1}{2}\left( \Pi_-H_4 H_0^{-2}H_4 \Pi_-H_2 \Pi_- +\hc\right)+O(\Lambda^3/\Delta^{1/4}) 
\end{align*}
\[H_{\eff,4}=\frac{1}{2}\Pi_-\Bigl( H_4 H_0^{-2} H_4 \Pi_- H_4 H_0^{-1} H_4 -H_4 H_0^{-1}H_4 H_0^{-1}H_4 H_0^{-1}H_4+\hc\Bigr)\Pi_-+O(\Lambda^4/\Delta^{1/4})\]
Combining these expressions with equations (\ref{eq:4thorder}) and (\ref{eq:H2H3}), and noting that some terms cancel because $\Pi_-H_2 \Pi_- = \Pi_- H_4 H_0^{-1}H_4 \Pi_-$, we have
\[ \| H_{\eff} - V H_{\operatorname{target}} V^{\dagger}\|\le \|H_{\eff}-\sum_{i=1}^{4} H_{\eff,i}\|+\epsilon/2+O(\Lambda^{2}/\Delta)+O(\Lambda^3/\Delta^{1/4}) +O(\Lambda^4/\Delta^{1/4}).\]
Given $\Delta>O( \Lambda^{20}/\epsilon^4)$, we may assume that the sum of the last three terms is less than $\epsilon/4$. By equation (23) of \cite{Bravyi-Hastings}, we have $\|H_{\eff}-\sum_{i=1}^{4} H_{\eff,i}\| = O(\Delta^{-4}\|A\|^{5})=O(\Lambda^5/\Delta^{1/4})<\epsilon/4$.
\end{proof}

\begin{replem}{lem:4thorderinterference}
Consider a Hilbert space $\cH=\cH_0 \otimes \bigotimes_{i\ge 1} \cH_i$ with multiple fourth-order mediator gadgets labelled by $i\ge 1$, each with heavy Hamiltonian $H_0^{(i)}$ which acts non-trivially only on $\cH_i$, and interaction terms $H_1^{(i)}$, $H_2^{(i)}$, $H_3^{(i)}$, $H_4^{(i)}$ which act non-trivially only on $\cH_i \otimes \cH_0$. Let $\Pi_-^{(i)}$ denote the projector onto the ground space of $H_0^{(i)}$, and $\Pi_+^{(i)} = I - \Pi_-^{(i)}$.
Suppose that for each $i$, these terms satisfy the conditions of Lemma~\ref{lem:fourthorder}; in particular, $H_0^{(i)} \Pi_-^{(i)} = 0$, $H_2^{(i)}$ and $H_3^{(i)}$ are block diagonal with respect to the $\Pi_-^{(i)}$, $\Pi_+^{(i)}$ split, $\Pi^{(i)}_- H_4^{(i)} \Pi^{(i)}_-=0$ and 
\[ \Pi_-^{(i)}H_2^{(i)}\Pi_-^{(i)} = \Pi^{(i)}_- H_4^{(i)} (H_0^{(i)})^{-1}H_4^{(i)} \Pi^{(i)}_- \text{ and }
\Pi^{(i)}_-H_3^{(i)} \Pi_-^{(i)} = -\Pi^{(i)}_- H_4^{(i)} (H_0^{(i)})^{-1}H_4^{(i)} (H_0^{(i)})^{-1}H_4^{(i)} \Pi^{(i)}_-.\]
%
%Let $H_0=\sum_{i}H_0^{(i)}$, $H_1=\sum_{i}H_1^{(i)}$, $H_2=\sum_{i}H_2^{(i)}$, $H_3=\sum_{i}H_3^{(i)}$ and $H_4=\sum_{i}H_4^{(i)}$.
For each $j \in \{0,\dots,4\}$, let $H_j = \sum_i H_j^{(i)}$, and let  $\Lambda \ge\max\{\|H_1\|,\|H_2\|,\|H_3\|,\|H_4\|\}$.

Suppose there exists a local isometry $V$ such that $\Im(V)$ is the ground space of $H_0$ and $\|V H_{\operatorname{target}}V^{\dagger}-M\|\le \epsilon/2$ where $M$ is equal to
\begin{align*}M=\sum_i& \Pi_-\left(H_1^{(i)}+H_4^{(i)} (H_0^{(i)})^{-1} H_2^{(i)} (H_0^{(i)})^{-1} H_4^{(i)} -H_4^{(i)} (H_0^{(i)})^{-1} H_4^{(i)} (H_0^{(i)})^{-1} H_4^{(i)} (H_0^{(i)})^{-1} H_4^{(i)}\right)\Pi_-\\
&+\sum_{i\neq j}\Pi_-\Big(H_4^{(i)} (H_0^{(i)})^{-1} H_4^{(j)}(H_0^{(j)})^{-1}H_4^{(j)} (H_0^{(i)})^{-1} H_4^{(i)}\\[-11pt]
&\qquad\qquad\qquad-H_4^{(i)} (H_0^{(i)})^{-1} H_4^{(j)} (H_0^{(i)}+H_0^{(j)})^{-1} H_4^{(j)} (H_0^{(i)})^{-1} H_4^{(i)}\\
&\qquad\qquad\qquad-H_4^{(i)} (H_0^{(i)})^{-1} H_4^{(j)} (H_0^{(i)}+H_0^{(j)})^{-1} H_4^{(i)} (H_0^{(j)})^{-1} H_4^{(j)}\Big)\Pi_-
\end{align*}
where $\Pi_-$  is the projector onto the ground space of $H_0$.

Then $\Delta H_0 + \Delta^{3/4} H_4 + \Delta^{1/4} H_3 + \Delta^{1/2}H_2+H_1$ $(\Delta/2,\eta,\epsilon)$ simulates $H_{\operatorname{target}}$, provided that $\Delta \ge O(\Lambda^{20}/\epsilon^4+\Lambda^4/\eta^4)$
\end{replem}

\newcommand{\Pbar}[1]{\overline{\Pi^{\{#1\}}}}
\begin{proof}
First we note that since the $H_0^{(i)}$ operators act on different subsystems for each $i$, all the $\Pi_-^{(i)}$ operators commute and $\Pi=\prod_i \Pi_-^{(i)}$.
For a set $S$, let $\overline{\Pi^S}$ be the projector onto the excited (i.e.\ not ground) space of all gadgets with label $i \in S$ and onto the ground space of all other gadgets. This is defined by
\[\overline{\Pi^{S}}=\left(\prod_{i \in S} \Pi_+^{(i)}\right)\left(\prod_{j  \notin S} \Pi_-^{(j)}\right).\]
These projectors are orthogonal in the sense that $\overline{\Pi^S}\overline{\Pi^T}=0$ unless $S=T$.
By definition, $\overline{\Pi^{S}}$ commutes with $H_0$, and the following relation holds:
\begin{equation}
\label{eq:PiSH0}
H_0^{-1}\overline{\Pi^S}=\left( \sum_{i \in S} H_0^{(i)}\right)^{-1}\overline{\Pi^S} =\overline{\Pi^S} H_0^{-1}.
\end{equation}
Since $\Pi_-^{(i)}H_4^{(i)}\Pi_-^{(i)}=0$, we have $H_4^{(i)}\Pi_-^{(i)}=\Pi_+^{(i)} H_4^{(i)}\Pi_-^{(i)}$ for all $i$. 
This implies the following relations:
\begin{equation}
\label{eq:PiH4Pi}
H_4^{(i)}\Pi_-=\Pbar{i}H_4^{(i)}\Pi_- \quad \text{ and } \quad(I-\Pi_-)H_4^{(i)}\Pbar{j} =\Pbar{i,j}H_4^{(i)}\Pbar{j}\text{ for all } i,j.
\end{equation}
We will now use equations (\ref{eq:PiSH0}) and (\ref{eq:PiH4Pi}) to check that the conditions of Lemma~\ref{lem:fourthorder} hold.
\begin{align*}
\Pi_- H_4 H_0^{-1}  H_4 \Pi_-&=\sum_{i,j}\Pi_- H_4^{(i)} H_0^{-1}  H_4^{(j)}\Pi_- =\sum_{i,j}\Pi_- H_4^{(i)} \Pbar{i}H_0^{-1} \Pbar{j} H_4^{(j)}\Pi_-\\
&=\sum_{i,j}\Pi_- H_4^{(i)} (H_0^{(i)})^{-1} \Pbar{i}\Pbar{j} H_4^{(j)}\Pi_-=\sum_{i}\Pi_- H_4^{(i)} (H_0^{(i)})^{-1} H_4^{(i)}\Pi_-\\
&=\sum_{i} \Pi_- H_2^{(i)}\Pi_-=\Pi_- H_2\Pi_-;
\end{align*}
\begin{align*}
\Pi_- H_4  H_0^{-1}  H_4 H_0^{-1}  H_4 \Pi_-&=\sum_{i,j,k}\Pi_- H_4^{(i)} H_0^{-1}  H_4^{(j)} H_0^{-1}  H_4^{(k)} \Pi_-\\
&=\sum_{i,j,k}\Pi_- H_4^{(i)}\Pbar{i} H_0^{-1}  H_4^{(j)} H_0^{-1} \Pbar{k} H_4^{(k)} \Pi_-\\
&=\sum_{i,j,k}\Pi_- H_4^{(i)} (H_0^{(i)})^{-1} \Pbar{i} H_4^{(j)}\Pbar{k} (H_0^{(k)})^{-1}  H_4^{(k)} \Pi_-\\
&=\sum_{i}\Pi_- H_4^{(i)} (H_0^{(i)})^{-1} H_4^{(i)} (H_0^{(i)})^{-1}  H_4^{(i)} \Pi_-\\
&=-\sum_{i} \Pi_- H_3^{(i)} \Pi_- =-\Pi_- H_3 \Pi_-,
\end{align*}
where in the fourth equality we have used the fact that $ \Pbar{i} H_4^{(j)}\Pbar{k}=0$ unless $i=j=k$, which again follows from $\Pi^{(i)} H_4^{(i)} \Pi^{(i)}=0$.

Finally we use equations (\ref{eq:PiSH0}) and (\ref{eq:PiH4Pi}) to calculate the fourth-order terms from Lemma~\ref{lem:fourthorder}:
\begin{align*}
\Pi_- H_4 H_0^{-1} H_2 H_0^{-1} H_4 \Pi_-&=\sum_{i,j,k}\Pi_- H_4^{(i)} H_0^{-1} H_2^{(j)} H_0^{-1} H_4^{(k)}\Pi_-\\
&=\sum_{i,j,k}\Pi_- H_4^{(i)} (H_0^{(i)})^{-1} \Pbar{i}H_2^{(j)} \Pbar{k}(H_0^{(k)})^{-1} H_4^{(k)}\Pi_-\\
&=\sum_{i,j}\Pi_- H_4^{(i)} (H_0^{(i)})^{-1} \Pbar{i}H_2^{(j)} \Pbar{i}(H_0^{(i)})^{-1} H_4^{(i)}\Pi_-\\
&=\sum_{i}\Pi_- H_4^{(i)} (H_0^{(i)})^{-1} H_2^{(i)}(H_0^{(i)})^{-1} H_4^{(i)}\Pi_-\\
&\qquad \qquad +\sum_{i\neq j}\Pi_- H_4^{(i)} (H_0^{(i)})^{-1} H_4^{(j)} (H_0^{(j)})^{-1}H_4^{(j)}(H_0^{(i)})^{-1} H_4^{(i)}\Pi_-,
\end{align*}
where in the third equality we note that $[H_2^{(j)},\Pi^{(k)}]=0$ for all $j,k$ since $H_2^{(j)}$ is block diagonal with respect to the $\Pi^{(j)}_-$, $\Pi^{(j)}_+$ split, which implies that $\Pbar{i}H_2^{(j)} \Pbar{k}=\Pbar{i} \Pbar{k} H_2^{(j)}=\delta_{ik}\Pbar{i}H_2^{(j)} \Pbar{k}$;
and in the final equality we used the fact that for $i\neq j$,  $\Pbar{i}H_2^{(j)}\Pbar{i}=\Pbar{i}\Pi_{-}^{(j)}H_2^{(j)}\Pi_{-}^{(j)}\Pbar{i}=\Pbar{i}H_4^{(j)} (H_0^{(j)})^{-1} H_4^{(j)} \Pbar{i}$. Next,
\begin{align*}
\Pi_-H_4 H_0^{-1} H_4 H_0^{-1} H_4 H_0^{-1} H_4 \Pi_-&=\sum_{i,j,k,l} \Pi_-H_4^{(i)} H_0^{-1} H_4^{(j)} H_0^{-1} H_4^{(k)} H_0^{-1} H_4^{(l)} \Pi_-\\
&=\sum_{i,j,k,l} \Pi_-H_4^{(i)}\Pbar{i} H_0^{-1} H_4^{(j)} H_0^{-1} H_4^{(k)} H_0^{-1} \Pbar{l}\Pi_- \\
&=\sum_{i,j,k,l} \Pi_-H_4^{(i)} (H_0^{(i)})^{-1} \Pbar{i}H_4^{(j)} H_0^{-1} H_4^{(k)}  \Pbar{l}(H_0^{(l)})^{-1}H_4^{(l)} \Pi_- \\
&=\sum_{i,j,k,l} \Pi_-H_4^{(i)} (H_0^{(i)})^{-1} H_4^{(j)} \Pbar{i,j}H_0^{-1}\Pbar{k,l} H_4^{(k)}(H_0^{(l)})^{-1}H_4^{(l)} \Pi_- 
\end{align*}
Note that $H_0^{-1}$ commutes with $\Pbar{i,j}$, and so there is a factor $\Pbar{i,j}\Pbar{k,l}$ which is zero unless $\{i,j\}=\{k,l\}$. There are three such possibilities:
\[\Pbar{i,j}H_0^{-1}\Pbar{k,l} =\left\{
\begin{array}{cc}
(H_0^{(i)})^{-1}& i=j=k=l\\
(H_0^{(i)}+H_0^{(j)})^{-1}&i=k \neq j=l\\
(H_0^{(i)}+H_0^{(j)})^{-1}& i=l \neq j=k
\end{array}\right.\]
Substituting these three possibilities back into the previous expression above, and summing over $i,j,k,l$, we find that $\Pi_- H_4 H_0^{-1} H_2 H_0^{-1} H_4 \Pi_-$ $ - \Pi_-H_4 H_0^{-1} H_4 H_0^{-1} H_4 H_0^{-1} H_4 \Pi_-$ is equal to the terms given in the statement of the lemma.
\end{proof}

\bibliographystyle{plain}
\bibliography{universalH}

\begin{thebibliography}{10}

\bibitem{Affleck89}
I.~Affleck.
\newblock Quantum spin chains and the {H}aldane gap.
\newblock {\em Journal of Physics: Condensed Matter}, 1(19):3047, 1989.

\bibitem{AKLT}
I.~Affleck, T.~Kennedy, E.H. Lieb, and H.~Tasaki.
\newblock Rigorous results on valence-bond ground states in antiferromagnets.
\newblock {\em Phys.\ Rev.\ Lett.}, 59(7):799, 1987.

\bibitem{aharonov09}
D.~Aharonov, D.~Gottesman, S.~Irani, and J.~Kempe.
\newblock The power of quantum systems on a line.
\newblock {\em Comm.\ Math.\ Phys.}, 287(1):41--65, 2009.
\newblock {\tt arXiv:0705.4077}.

\bibitem{bausch16}
J.~Bausch, T.~Cubitt, and M.~Ozols.
\newblock The complexity of translationally-invariant spin chains with low
  local dimension, 2016.
\newblock {\tt arXiv:1605.01718}.

\bibitem{Bausch17}
J.~Bausch and S.~Piddock.
\newblock The complexity of translationally-invariant low-dimensional spin
  lattices in {3D}, 2017.
\newblock {\tt arXiv:1702.08830}.

\bibitem{Beach09}
K.~Beach, F.~Alet, M.~Mambrini, and S.~Capponi.
\newblock {$SU (N)$} {H}eisenberg model on the square lattice: A continuous-{N}
  quantum {M}onte {C}arlo study.
\newblock {\em Phys.\ Rev.\ B}, 80(18):184401, 2009.

\bibitem{bookatz14}
A.~Bookatz.
\newblock {QMA}-complete problems.
\newblock {\em Quantum Inf.\ Comput.}, 14(5\&6):361--383, 2014.
\newblock {\tt arXiv:1212.6312}.

\bibitem{bravyi06a}
S.~Bravyi, A.~Bessen, and B.~Terhal.
\newblock {M}erlin-{A}rthur games and stoquastic complexity, 2006.
\newblock {\tt quant-ph/0611021}.

\bibitem{Bravyi-DiVincenzo-Loss}
S.~Bravyi, D.~DiVincenzo, and D.~Loss.
\newblock Schrieffer--{W}olff transformation for quantum many-body systems.
\newblock {\em Ann. of Phys.}, 326(10):2793--2826, 2011.
\newblock {\tt arXiv:1105.0675}.

\bibitem{Bravyi-Hastings}
S.~Bravyi and M.~Hastings.
\newblock On complexity of the quantum {I}sing model.
\newblock {\em Comm.\ Math.\ Phys.}, 349(1):1--45, 2017.
\newblock {\tt 1410.0703}.

\bibitem{cao13}
Y.~Cao, R.~Babbush, J.~Biamonte, and S.~Kais.
\newblock Towards experimentally realizable hamiltonian gadgets.
\newblock {\em preprint}, 2013.
\newblock {\tt arXiv:1311.2555}.

\bibitem{cao17}
Y.~Cao and S.~Kais.
\newblock Efficient optimization of perturbative gadgets, 2017.
\newblock {\tt arXiv:1709.02705}.

\bibitem{cao15}
Y.~Cao and D.~Nagaj.
\newblock Perturbative gadgets without strong interactions.
\newblock {\em Quantum Inf.\ Comput.}, 15(13\&14):1197--1222, 2015.
\newblock {\tt arXiv:1408.5881}.

\bibitem{cirac12}
J.~I. Cirac and P.~Zoller.
\newblock Goals and opportunities in quantum simulation.
\newblock {\em Nature Physics}, 8:264--266, 2012.

\bibitem{creignou01}
N.~Creignou, S.~Khanna, and M.~Sudan.
\newblock {\em Complexity Classifications of Boolean Constraint Satisfaction
  Problems}.
\newblock SIAM, 2001.

\bibitem{Cubitt-Montanaro}
T.~Cubitt and A.~Montanaro.
\newblock Complexity classification of local {H}amiltonian problems.
\newblock {\em SIAM J.\ Comput.}, 45(2):268--316, 2016.
\newblock {\tt arXiv:1311.3161}.

\bibitem{cubitt17}
T.~Cubitt, A.~Montanaro, and S.~Piddock.
\newblock Universal quantum {H}amiltonians, 2017.
\newblock {\tt arXiv:1701.05182}.

\bibitem{frieze97}
A.~Frieze and M.~Jerrum.
\newblock Improved approximation algorithms for {MAX k-CUT} and {MAX
  BISECTION}.
\newblock {\em Algorithmica}, 18(1):67--81, 1997.

\bibitem{FuchsSchweigert}
J.~Fuchs and C.~Schweigert.
\newblock {\em Symmetries, Lie algebras and representations: A graduate course
  for physicists}.
\newblock Cambridge University Press, 2003.

\bibitem{georgescu14}
I.~Georgescu, S.~Ashhab, and F.~Nori.
\newblock Quantum simulation.
\newblock {\em Rev. Mod. Phys.}, 86:153, 2014.
\newblock {\tt arXiv:1308.6253}.

\bibitem{gharibian15}
S.~Gharibian, Y.~Huang, Z.~Landau, and S.~W. Shin.
\newblock Quantum {H}amiltonian complexity.
\newblock {\em Foundations and Trends in Theoretical Computer Science},
  10(3):159--282, 2015.
\newblock {\tt arXiv:1401.3916}.

\bibitem{gottesman13}
D.~Gottesman and S.~Irani.
\newblock The quantum and classical complexity of translationally invariant
  tiling and {H}amiltonian problems.
\newblock {\em Theory of Computing}, 9(2):31--116, 2013.
\newblock {\tt arXiv:0905.2419}.

\bibitem{hallgren13}
S.~Hallgren, D.~Nagaj, and S.~Narayanaswami.
\newblock The local hamiltonian problem on a line with eight states is
  {QMA}-complete.
\newblock {\em Quantum Inf.\ Comput.}, 13(9\&10):0721--0750, 2013.
\newblock {\tt arXiv:1312.1469}.

\bibitem{Harada02}
K.~Harada and N.~Kawashima.
\newblock Quadrupolar order in isotropic {H}eisenberg models with biquadratic
  interaction.
\newblock {\em Phys.\ Rev.\ B}, 65(5):052403, 2002.

\bibitem{jonsson00}
P.~Jonsson.
\newblock Boolean constraint satisfaction: complexity results for optimization
  problems with arbitrary weights.
\newblock {\em Theoretical Computer Science}, 244:189--203, 2000.

\bibitem{jordan08}
S.~Jordan and E.~Farhi.
\newblock Perturbative gadgets at arbitrary orders.
\newblock {\em Physical Review A}, 77(6):062329, 2008.
\newblock {\tt arXiv:0802.1874}.

\bibitem{Kempe-Kitaev-Regev}
J.~Kempe, A.~Kitaev, and O.~Regev.
\newblock The complexity of the local {H}amiltonian problem.
\newblock {\em SIAM J.\ Comput.}, 35(5):1070--1097, 2006.
\newblock {\tt quant-ph/0406180}.

\bibitem{Kennedy90}
T.~Kennedy.
\newblock Exact diagonalisations of open spin-1 chains.
\newblock {\em Journal of Physics: Condensed Matter}, 2(26):5737, 1990.

\bibitem{Kitaev-Shen-Vyalyi}
A.~Yu. Kitaev, A.~H. Shen, and M.~N. Vyalyi.
\newblock {\em Classical and Quantum Computation}, volume~47 of {\em Graduate
  Studies in Mathematics}.
\newblock AMS, 2002.

\bibitem{Science}
G.~De las Cuevas and T.~Cubitt.
\newblock Simple universal models capture all classical spin physics.
\newblock {\em Science}, 351(6278):1180--1183, 2016.
\newblock {\tt arXiv:1406.5955}.

\bibitem{Lauchli06}
A.~L{\"a}uchli, F.~Mila, and K.~Penc.
\newblock Quadrupolar phases of the $s=1$ bilinear-biquadratic {H}eisenberg
  model on the triangular lattice.
\newblock {\em Phys.\ Rev.\ Lett.}, 97:087205, 2006.
\newblock {\tt cond-mat/0605234}.

\bibitem{linden99}
N.~Linden, S.~Popescu, and A.~Sudbery.
\newblock Nonlocal parameters for multiparticle density matrices.
\newblock {\em Phys.\ Rev.\ Lett.}, 83:243, 1999.
\newblock {\tt quant-ph/9801076}.

\bibitem{Lou09}
J.~Lou, A.W. Sandwik, and N.~Kawashima.
\newblock Antiferromagnetic to valence-bond-solid transitions in
  two-dimensional {$SU (N)$} {H}eisenberg models with multispin interactions.
\newblock {\em Phys.\ Rev.\ B}, 80(18):180414, 2009.

\bibitem{mattis93}
D.~Mattis.
\newblock {\em The Many-body Problem: An Encyclopedia Of Exactly Solved Models
  In One Dimension}.
\newblock World Scientific, 1993.

\bibitem{Okubo77}
S.~Okubo.
\newblock {Casimir invariants and vector operators in simple and classical Lie
  algebras}.
\newblock {\em Journal of Mathematical Physics}, 18(5976):2382--1857, 1977.

\bibitem{Oliveira-Terhal}
R.~Oliveira and B.~Terhal.
\newblock The complexity of quantum spin systems on a two-dimensional square
  lattice.
\newblock {\em Quantum Inf.\ Comput.}, 8:0900, 2008.

\bibitem{parkinson10}
J.~Parkinson and D.~Farnell.
\newblock {\em An Introduction to Quantum Spin Systems}.
\newblock Springer, 2010.

\bibitem{Piddock-Montanaro}
S.~Piddock and A.~Montanaro.
\newblock The complexity of antiferromagnetic interactions and {2D} lattices.
\newblock {\em Quantum Inf.\ Comput.}, 17(7\&8):636--672, 2017.
\newblock {\tt arXiv:1506.04014}.

\bibitem{Read89}
N.~Read and S.~Sachdev.
\newblock {Valence-bond and spin-Peierls ground states of low-dimensional
  quantum antiferromagnets}.
\newblock {\em Phys.\ Rev.\ Lett.}, 62(14):1694, 1989.

\bibitem{schaefer78}
T.~Schaefer.
\newblock The complexity of satisfiability problems.
\newblock In {\em Proc. 10\textsuperscript{th} Annual ACM Symp. Theory of
  Computing}, pages 216--226, 1978.

\bibitem{Schuch-Verstraete}
N.~Schuch and F.~Verstraete.
\newblock Computational complexity of interacting electrons and fundamental
  limitations of {D}ensity {F}unctional {T}heory.
\newblock {\em Nature Physics}, 5:732--735, 2009.
\newblock {\tt arXiv:0712.0483}.

\bibitem{thapper16}
J.~Thapper and S.~\v{Z}ivn\'y.
\newblock The complexity of finite-valued {CSP}s.
\newblock {\em J. ACM}, 63(4), 2016.
\newblock {\tt arXiv:1210.2987}.

\bibitem{thompson88}
G.~Thompson.
\newblock Normal forms for skew-symmetric matrices and {H}amiltonian systems
  with first integrals linear in momenta.
\newblock {\em Proceedings of the American Mathematical Society},
  104(3):910--916, 1988.

\bibitem{wu82}
F.~Wu.
\newblock The {P}otts model.
\newblock {\em Rev. Mod. Phys.}, 54(1):235--268, 1982.

\end{thebibliography}

\end{document}